\documentclass[journal,onecolumn]{IEEEtran}

\usepackage{bbm}
\usepackage{amsmath}
\usepackage{amsfonts}
\usepackage{mathrsfs}
\usepackage{hyperref}
\usepackage{amssymb}
\usepackage{bm}
\usepackage{epsfig}
\usepackage{color}

\usepackage{graphicx}
\usepackage{epstopdf}

 \newtheorem{thm}{Theorem}[section]
\newtheorem{rem}[thm]{Remark}

\newtheorem{lem}[thm]{Lemma}
 \newtheorem{prop}[thm]{Proposition}
\newtheorem{cor}[thm]{Corollary}

 \newtheorem{cons}[thm]{Construction}

\newcommand{\B}{{\mathcal B}}
\newcommand{\C}{{\mathcal C}}
\newcommand{\HH}{{\mathcal H}}
\newcommand{\MM}{{\mathcal M}}

\newcommand{\vu}{\mathbf{u}}
\newcommand{\vv}{\mathbf{v}}

\newcommand{\vc}{\mathbf{c}}
\newcommand{\ve}{\mathbf{e}}
\newcommand{\vj}{\mathbf{j}}
\newcommand{\vvx}{\mathbf{x}}
\newcommand{\supp}{{\rm supp}}

\newcommand{\bbZ}{{\mathbb{Z}}}

\begin{document}
\title{New bounds and constructions for multiply constant-weight codes}

\author{Xin Wang, Hengjia Wei, Chong Shangguan, and Gennian Ge
\thanks{The research of H. Wei was supported by the Post-Doctoral Science Foundation of China under Grant No.~2015M571067, and Beijing Postdoctoral Research Foundation.

The research of G. Ge was supported by the National Natural Science Foundation of China under Grant Nos. 61171198, 11431003 and 61571310, and the
Importation and Development of High-Caliber Talents Project of Beijing Municipal Institutions.
}
\thanks{X. Wang is with the School of Mathematical Sciences, Zhejiang University,
Hangzhou 310027,  China (e-mail: 11235062@zju.edu.cn).}
\thanks{H. Wei is with the School of Mathematical Sciences, Capital Normal University,
Beijing 100048, China (e-mail: ven0505@163.com).}
\thanks{C. Shangguan is with the School of Mathematical Sciences, Zhejiang University,
Hangzhou 310027,  China (e-mail: 11235061@zju.edu.cn).}
\thanks{G. Ge is with the School of Mathematical Sciences, Capital Normal University,
Beijing 100048, China (e-mail: gnge@zju.edu.cn). He is also with Beijing Center for Mathematics and Information Interdisciplinary Sciences, Beijing, 100048, China.}
}

\maketitle

\begin{abstract}
Multiply constant-weight codes (MCWCs) were introduced recently to improve the reliability of certain physically unclonable function response. In this paper, the bounds of MCWCs and the constructions of optimal MCWCs are studied. Firstly, we derive three different types of  upper bounds which improve the Johnson-type bounds given by Chee  {\sl et al.}  in some parameters. The asymptotic lower bound of MCWCs is also examined. Then we obtain  the asymptotic existence of two classes of optimal MCWCs, which shows that the Johnson-type bounds for MCWCs with distances $2\sum_{i=1}^mw_i-2$ or $2mw-w$ are asymptotically exact. Finally, we construct a class of optimal MCWCs with total weight four and distance six by establishing the connection between such MCWCs and a new kind of combinatorial structures. As a consequence, the maximum sizes of MCWCs with total weight less than or equal to four are determined almost completely.
\end{abstract}


\section{Introduction}
Modern cryptographic practice rests on the use of one-way functions, which are easy to evaluate but difficult to invert. Unfortunately, commonly used one-way functions are either based on unproven conjectures or have known vulnerabilities. Physically unclonable functions (PUFs), introduced by Pappu {\sl et al.} \cite{phy}, provide innovative low-cost authentication methods and robust structures against physical attacks. Recently, PUFs have become a trend to provide security in low cost devices such as Radio Frequency Identifications (RFIDs) and smart cards \cite{loop,prf,phy,puf}. Multiply constant-weight codes (MCWCs) establish the connection between the design of the Loop PUFs \cite{loop} and coding theory, thus were put forward in \cite{MCWC}. In an MCWC, each codeword is a binary word of length $mn$ which is partitioned into $m$ equal parts and has weight exactly $w$ in each part \cite{MCWC}. The more general definition of MCWCs with different lengths and weights in different parts can be found in \cite{zhang}. This definition generalizes the classic definitions of constant-weight codes (CWCs) (where $m=1$) and doubly constant-weight codes (where $m=2$) \cite{jon,lev}.

The theory of MCWCs is at a rudimentary stage. In \cite{zhang} Chee {\sl et al.}  extended techniques of Johnson \cite{jon} and established certain preliminary upper and lower bounds for possible sizes of MCWCs. They also showed that these bounds are asymptotically tight up to a constant factor. In \cite{Cheeoptimal}, Chee {\sl et al.} gave some combinatorial constructions for MCWCs which yield several new infinite families of optimal MCWCs. In particular, by establishing the connection between MCWCs and combinatorial designs and using some existing results in design theory, they determined the maximum sizes of MCWCs with total weight less than or equal to four,  leaving an infinite class open. In the same paper, they also showed that the Johnson-type bounds are asymptotically tight for fixed weights and distances by applying Kahn's Theorem \cite{Kahn} on the size of the matching in hypergraphs. Furthermore,  in \cite{Cheegraph}, they demonstrated that one of the Johnson-type bounds is asymptotically exact for the distance $2mw-2$. This was achieved by applying the theory of  edge-colored digraph-decompositions \cite{LW}.

In this paper, we continue the study on  the bounds of MCWCs and the constructions of optimal MCWCs. Our main contributions are as follows:

\begin{itemize}
  \item We extend the techniques of Agrell {\sl et al.} \cite{var} and improve the Johnson-type bounds derived in \cite{zhang}. We also show that the generalised Gilbert-Varshmov (GV) bound \cite{gil,varsh} is better than the asymptotic lower bounds derived in \cite{zhang}, where the concatenation techniques are employed.
  \item  We obtain the asymptotic existence of two classes of optimal MCWCs. One of them generalizes the known result of \cite{Cheegraph} for MCWCs with different weights in different parts. The other shows that another Johnson-type bound is asymptotically exact for distance $2mw-w$.
  \item We consider the open case of optimal MCWCs in \cite{Cheeoptimal}, i.e., doubly constant-weight codes with  weight two in each part and distance six. We establish an equivalence relation  between such MCWCs and certain kind of combinatorial structures, which are called skew almost-resolvable squares. Accordingly, several new constructions are proposed. As a consequence, the maximum sizes of MCWCs with total weight less than or equal to four are determined almost completely, leaving a very small number of lengths open.
\end{itemize}

The rest of this article is organized as follows. Section 2 collects the necessary definitions and notations. Section 3 gives three forms of upper bounds, which can improve the previous Johnson-type bounds. Section 4 studies the asymptotic lower bounds of MCWCs.
Section 5 presents the asymptotic existence of two classes of optimal MCWCs. Section 6 handles the optimal MCWCs with total weight four. A conclusion is made in Section 7.

\section{Definitions and Notations}

\subsection{Multiply Constant-weight Codes}
All sets considered in this paper are finite if not obviously
infinite. We use $[n]$ to denote the set $\{1,2,\ldots,n\}$. If $X$ and $R$ are finite sets, $R^X$ denotes the set of
vectors of length $|X|$. Each component of a vector $\vu\in
R^X$ takes value in $R$ and is indexed by an element of $X$, that is,
$\vu=(\vu_x)_{x\in X}$, and $\vu_x\in R$ for each $x\in X$. A
{\it $q$-ary code of length $n$} is a set $\C \subseteq \bbZ_q^X$ for some
$X$ with size $n$. The elements of $\C$ are called {\it codewords}. The {\it support} of a vector $\vu \in \bbZ_q^X$, denoted $\supp(\vu)$, is the set $\{x \in X : \vu_x \not= 0\}$. The
{\it Hamming norm} or the {\it Hamming weight} of a vector $\vu\in\bbZ_q^X$ is
defined as $\|\vu\|=| \supp(\vu)|$. The distance
induced by this norm is called the {\it Hamming distance}, denoted $d_H$,
so that $d_H(\vu,\vv)=\| \vu-\vv \|$, for $\vu,\vv\in\bbZ_q^X$. A code $\C$ is said to {\it have distance $d$} if the Hamming distance between any two distinct codewords of $\C$ is at least $d$. A $q$-ary code of length $n$ and distance $d$ is called an $(n,d)_q$ code. When $q=2$, an $(n,d)_2$ code is simply called an $(n,d)$ code.

Let $m$, $N$ be positive integers and $X$ be a set of size $N$. Suppose that $X$ can be partitioned as $X=X_1\cup X_2\cup\cdots \cup X_m$ with $|X_i|=n_i$, $i=1,2,\ldots,m$. An $(N,d)$ code $\C \subseteq \bbZ_2^X$ is said to be of {\it multiply constant-weight} and denoted by MCWC$(w_1,n_1;w_2,n_2;\cdots;w_m,n_m;d)$, if each codeword has the weight $w_1$ in the coordinates indexed by $X_1$, weight $w_2$ in the coordinates indexed by $X_2$, and so on and so forth. When $w_1=w_2=\cdots=w_m=w$ and $n_1=n_2=\cdots=n_m=n$, we simply denote this multiply constant-weight code of length $N=mn$ by MCWC$(m,n,d,w)$.

The largest size of an $(n,d)_q$ code is denoted by $A_q(n,d)$. When $q=2$, the size is simply denoted by $A(n,d)$. The largest size of an MCWC$(w_1, n_1;w_2, n_2; \ldots;w_m, n_m; d)$ is denoted by $T(w_1, n_1;w_2, n_2;\ldots;w_m, n_m; d)$; the largest size of an MCWC$(m,n,d,w)$ is denoted by $M(m,n,d,w)$; and the largest size of a CWC$(n,d,w)$ is denoted by $A(n,d,w)$. The code achieving the largest size is said to be {\it optimal}.

Next, we will restate the known results about MCWCs without proof, more details can be found in \cite{zhang}.
The authors of \cite{zhang} first use the concatenation technique to construct MCWCs from the classic $q$-ary codes.

\begin{prop}\label{concatenation}(\cite{zhang})
Let $q\leqslant A(n,d_1,w)$, we have
\[
M(m,n,d_1d_2,w)\geq A_q(m,d_2).
\]
Specially,
$M(m,qw,2d,w)$ $\geq$ $A_q(mw,d)$.
\end{prop}

As MCWC is a generalization of CWC, the techniques of Johnson for CWC \cite{jon} can be naturally
extended to give the recursive bounds as follows:

\begin{prop}(\cite{zhang})\label{jonp}
\begin{equation}\label{jon1}
T(w_1,n_1;w_2,n_2;\ldots;w_m,n_m;d)\leq \lfloor \frac{n_i}{w_i}T(w_1,n_1;\ldots;w_i-1,n_i-1;\ldots;w_m,n_m;d)\rfloor,
\end{equation}
\begin{equation}\label{jon2}
T(w_1,n_1;w_2,n_2;\ldots;w_m,n_m;d)\leq \lfloor \frac{n_i}{n_i-w_i}T(w_1,n_1;\ldots;w_i,n_i-1;\ldots;w_m,n_m;d)\rfloor,
\end{equation}
\begin{equation}\label{jon3}
T(w_1,n_1;w_2,n_2;\ldots;w_m,n_m;d)\leq \lfloor \frac{u}{w_1^2/n_1+w_2^2/n_2+\cdots+w_m^2/n_m-\lambda}\rfloor,
\end{equation}
where $d=2u$ and $\lambda=w_1+w_2+\cdots+w_m-u$.
\end{prop}

\begin{prop}(\cite{zhang})\label{mjonp}
\begin{equation}\label{mjon1}
M(m,n,d,w)\leq \lfloor \frac{n^m}{w^m}M(m,n-1,d,w-1)\rfloor,
\end{equation}
\begin{equation}
M(m,n,d,w)\leq \lfloor \frac{n^m}{(n-w)^m}M(m,n-1,d,w)\rfloor,
\end{equation}
\begin{equation}\label{mjon3}
M(m,n,d,w)\leq \lfloor \frac{d/2}{d/2+mw^2/n-mw}\rfloor.
\end{equation}
\end{prop}

\subsection{Association Schemes}
Let $X$ be a finite set with at least two elements and, for any integer $n\geq 1$, let ${\cal R}=\{R_0,R_1,\ldots,R_n\}$ be a family of $n+1$ relations $R_i$ on $X$. The pair $(X,{\cal R})$ will be called an {\it association scheme with $n$ classes} if the  following  three conditions are satisfied:

\begin{itemize}
\item[1.]The set $\cal R$ is a partition of $X^2$ and $R_0$ is the diagonal relation, i.e., $R_0=\{(x,x)|x\in X\}$.
\item[2.]For $i=0,1,\ldots,n$, the inverse $R_{i}^{-1}=\{(y,x)|(x,y)\in R_i\}$ of the relation $R_i$ also belongs to $\cal R$.
\item[3.]For any triple of integers $i,j,k=0,1,\ldots,n$, there exists a number $p^{(k)}_{i,j}=p^{(k)}_{j,i}$ such that, for all $(x,y)\in R_k$:
    \begin{equation}
    |\{z\in X|(x,z)\in R_i, (z,y)\in R_j\}|=p^{(k)}_{i,j}.
    \end{equation}
    The $p^{(k)}_{i,j}$'s are called the {\it intersection numbers} of the scheme $(X,{\cal R})$.
\end{itemize}

Any relation $R_i$ can be described by its {\it adjacency matrix} $D_i\in {\mathbb{C}}(X,X)$, defined as follows:
\[
D_i(x,y)=\left\{
     \begin{array}{ll}
     1, &  (x,y)\in R_i,\\
     0, &  (x,y)\not\in R_i.
     \end{array}
     \right.
\]

We call the linear space
\[
A=\{  \sum_{i=0}^{n}\alpha_i D_i|\alpha_i\in \mathbb{C} \}
\]
the {\it Bose-Mesner algebra} of the association scheme $(X,\cal R)$. There is a set of pairwise orthogonal idempotent matrices $J_0,J_1,\ldots,J_n$, which forms another basis of this Bose-Mesner algebra.

Given two bases $\{D_k\}$ and $\{J_k\}$ of the Bose-Mesner algebra of a scheme, let us consider the linear transformations from one  into the other:
\[
D_k=\sum_{i=0}^{n}P_k(i)J_i,~~~~k=0,1,\ldots,n.
\]

From these we construct a square matrix $P$ of order $n+1$ whose $(i,k)$-entry is $P_k(i)$:
\[
P=[P_k(i):0\leq i,k\leq n].
\]

Since $P$ is nonsingular, there exists a unique square matrix $Q$ of order $n+1$ over $\mathbb{C}$ such that
\[
PQ=QP=|X|I.
\]
The matrices $P$ and $Q$ are called the {\it eigenmatrices} of the association scheme.

Let ${\cal R}=\{R_0,R_1,\ldots,R_n\}$ be a set of $n+1$ relations on $X$ of an association scheme. For a nonempty subset $Y$ of $X$, let us define the {\it inner distribution} of $Y$ with respect to $\cal R$ to be the $(n+1)$-tuple $\alpha=(\alpha_0,\alpha_1,\ldots,\alpha_n)$ of nonnegative rational numbers $\alpha_i$ given by
\[
\alpha_i=|Y|^{-1}|R_i\cap Y^2|.
\]

In \cite{del}, Delsarte gave a key observation about the inner distribution and the eigenmatrix Q.

\begin{thm}\label{lp}(\cite{del})
The components $\alpha Q_k$ of the row vector $\alpha Q$ are nonnegative.
\end{thm}

Let $w$ and $n$ be integers, with $1\leq w\leq n$. In the Hamming space of dimension $n$ over $\mathbb{F}=\{0,1\}$, we consider the subset $X$ of $\mathbb{F}^n$ as follows:
\[
X=\{x\in \mathbb{F}^n|w_H(x)=w\},
\]
and we define the distance relations $R_0,R_1,\ldots,R_w$:
\[
R_i=\{(x,y)\in X^2|d(x,y)=2i\}.
\]
For given $n$ and $w$, with $1\leq w\leq n/2$, we call $(X,\cal R)$ the {\it Johnson scheme} $J(w,n)$, i.e., binary codes with length $n$ and constant weight $w$.

Given an integer $k$, with $0\leq k\leq w$, we define {\it Eberlein polynomial} $E_k(u)$, in the indeterminate $u$, as follows:
\[
E_k(u)=\sum_{i=0}^{k}(-1)^i{u \choose i}{w-u \choose k-i}{n-w-u \choose k-i}.
\]

\begin{thm}(\cite{del})
The eigenmatrices $P$ and $Q$ of the Johnson scheme $J(w,n)$ are given by
\[
P_k(i)=E_k(i),
\]
\[
Q_i(k)=\frac{\mu_i E_k(i)}{{w \choose i}{n-w \choose i}},
\]
where $\mu_i=\frac{n-2i+1}{n-i+1}{n \choose i}$.
\end{thm}

\subsection{Design Theory}
To give our constructions of optimal MCWCs, we need the following notations and results in design theory.

Let $K$ be a subset of positive integers and  $\lambda$ be a
positive integer. A {\it pairwise balanced design} ($(v, K, \lambda)$-PBD or $(K, \lambda)$-PBD of order $v$) is a pair
($X,\B$), where $X$ is a finite set ({\it the point set}) of cardinality
$v$ and $\B$ is a family of subsets ({\it blocks}) of $X$ that satisfy
(1) if $B\in \B$, then $|B|\in K$ and (2) every pair of distinct
elements of $X$ occurs in exactly $\lambda$ blocks of $\B$. The
integer $\lambda$ is the {\it index} of the PBD. When $K=\{k\}$, a $(v, \{k\}, \lambda)$-PBD is also known as a {\it balanced incomplete block design} (BIBD), which is denoted by BIBD$(v,k,\lambda)$.


\begin{thm}[\cite{CD}]\label{PBD579}
\label{PBD5-9} For any odd integer $v\geq 5$, a $(v,\{5,7,9\},1)$-PBD exists with exceptions $v\in [11,19]\cup\{23\}\cup[27,33]\cup\{39\}$, and possible exceptions $v\in \{43,51,59,71,75,83,87,95,99,107,111,113,115,119,139,179\}$.
\end{thm}

An {\it $\alpha$-parallel class} of blocks in a BIBD $(X,{\cal B})$ is a subset ${\cal B}' \subset {\cal B}$ such that
each point $x \in X$ is contained in exactly $\alpha$ blocks in
${\cal B}'$. When $\alpha =1$, we simply call it a {\it parallel
class}, as usual. If the block set $\cal B$ can be partitioned into
$\alpha$-parallel classes, then the BIBD is called {\it $\alpha$-resolvable} (or just {\it resolvable} if $\alpha=1$). We will use $\alpha$-resolvable BIBDs to construct optimal MCWCs.

A {\it group divisible design} (GDD) is a triple $(X,{\cal G},{\cal B})$ where $X$ is a set of points, ${\cal G}$ is a partition of $X$ into {\it groups}, and ${\cal B}$ is a collection of subsets of $X$ called {\it blocks} such that any pair of distinct points from $X$ occurs either in some group or in exactly one block, but not both.  A $K$-GDD of type $g^{u_1}_1 g^{u_2}_2\ldots g^{u_s}_s$ is a GDD in which every block has size from the set $K$ and in which there are $u_i$ groups of size $g_i, i=1,2,\ldots,s$. When $K=\{k\}$, we simply write $k$ for $K$.  A $k$-GDD of type $m^k$ is also called a {\em transversal
design} and denoted by TD$(k,m)$.


\begin{thm}[\cite{Abel,CD}]
\label{TD} Let $m$ be a positive integer. Then:
\begin{enumerate}
\item a TD$(4,m)$ exists if $m\not\in \{2,6\}$;
\item a TD$(5,m)$ exists if $m\not\in \{2,3,6,10\}$;
\item a TD$(6,m)$ exists if $m\not\in\{2,3,4,6,10,22\}$;
\item a TD$(m+1,m)$ exists if $m$ is a prime power.
\end{enumerate}
\end{thm}

\subsection{Decomposition of Edge-colored Complete Digraphs}
Denote the set of all ordered pairs of a finite set $X$ with
distinct components by $\overline{X \choose 2}$. An {\it edge-colored digraph} is a
triple $G =(V,C,E)$, where $V$ is a finite set of {\it vertices}, $C$ is
a finite set of {\it colors} and $E$ is a subset of $\overline{X \choose 2}\times C$. Members
of $E$ are called {\it edges}. The {\it complete edge-colored digraph} on
$n$ vertices with $r$ colors, denoted by $K_n^{(r)}$, is the edge-colored
digraph $(V,C,E)$, where $|V|=n$, $|C| = r$ and $E = \overline{X \choose 2}\times C$.

A family $\cal F$ of edge-colored subgraphs of an edge-colored
digraph $K$ is a {\it decomposition} of $K$ if every edge of $K$ belongs
to exactly one member of $\cal F$. Given a family of edge-colored
digraphs $\cal G$, a decomposition $\cal F$ of $K$ is a {\it $\cal G$-decomposition of
$K$} if each edge-colored digraph in $\cal F$ is isomorphic to some
$G \in \cal G$. In \cite{LW}, Lamken and Wilson exhibited the asymptotic existence
of decompositions of $K_n^{(r)}$ for a fixed family of digraphs. To state their result, we require more concepts.

Consider an edge-colored digraph $G = (V,C,E)$ with
$|C| = r$. Let $((u,v), c) \in E$ denote a directed edge from
$u$ to $v$, colored by $c$. For any vertex $u$ and color $c$, define
the {\it indegree} and {\it outdegree} of $u$ with respect to $c$, to be the
number of directed edges of color $c$ entering and leaving $u$,
respectively. Then for vertex $u$, we define the {\it degree vector}
of $u$ in $G$, denoted by $\tau(u,G)$, to be the vector of length $2r$,
$\tau(u,G) =(\textup{in}_1(u,G), \textup{out}_1(u,G),\ldots,\textup{in}_r(u,G),\textup{out}_r(u,G))$.
Define $\alpha(G)$ to be the greatest common divisor of the integers
$t$ such that the $2r$-vector $(t, t, \ldots, t)$ is a nonnegative integral
linear combination of the degree vectors $\tau(u,G)$ as $u$ ranges
over all vertices of all digraphs $G \in \cal G$.

For each $G = (V,C,E) \in \cal G$, let $\mu(G)$ be the {\it edge vector}
of length $r$ given by $\mu(G)=(m_1(G),m_2(G),\ldots,m_r(G))$ where $m_i(G)$ is the number of edges with color $i$ in $G$. We
denote by $\beta(G)$ the greatest common divisor of the integers $m$ such that $(m, m, \ldots, m)$ is a nonnegative integral linear
combination of the vectors $\mu(G)$, $G \in \cal G$. Then $\cal G$ is said to
be {\it admissible} if $(1,1,\ldots,1)$ can be expressed as a positive
rational combination of the vectors $\mu(G)$, $G \in \cal G$.

\begin{thm}[Lamken and Wilson \cite{LW}]\label{graphdecom} Let $\cal G$ be an admissible family of edge-colored digraphs with $r$ colors. Then there exists a constant $n_0=n_0(\cal G)$ such that a $\cal G$-decomposition of $K_n^{(r)}$ exists for every $n\geq n_0$ satisfying $n(n-1)\equiv 0 \pmod{\beta(\cal G)}$ and
$n-1\pmod{\alpha(\cal G)}$.
\end{thm}

In the same paper, the above theorem had also been extended to the multiplicity case. Consider the problem of finding a family $\cal F$ of subgraphs of $K_n^{(r)}$ each of which is isomorphic to a member of $\cal G$, so that each edge of  $K_n^{(r)}$ of color $i$ occurs in exactly $\lambda_i$ of the members of $\cal F$. We can think of this as a $\cal G$-decomposition of $K_n^{[\lambda_1, \lambda_2 , \ldots, \lambda_r]}$, which denotes the digraph on $n$ vertices where there are exactly $\lambda_i$ edges of color $i$ joining $x$ to $y$ for any ordered pair $(x, y)$ of distinct vertices.

Let $\bm{\lambda}$ $=(\lambda_1, \lambda_2 , \ldots, \lambda_r)$ be a vector of positive integers. Let $\alpha(\cal{G}; \bm{\lambda})$ denote
the least positive integer $t$ such that the constant vector $t\bm{\lambda}$ is an integral
linear combination of $\tau(u, G)$ as $u$ ranges over all vertices of all digraphs $G \in \cal G$. Let $\beta(\cal{G}; \bm{\lambda})$
denote the least positive integer $m$ such that the constant vector $m\bm{\lambda}$ is an
integral linear combination of $\mu(G)$, $G\in \cal G$.  We say $\cal G$ is {\it $\bm{\lambda}$-admissible} when
the vector $\bm{\lambda}$ is a positive rational linear combination of $\mu(G)$, $G\in \cal G$.

\begin{thm}[Lamken and Wilson \cite{LW}]\label{mgraphdecom} Let $\cal G$ be a $\bm{\lambda}$-admissible family of edge-$r$-colored
digraphs, where $\bm{\lambda}=(\lambda_1, \lambda_2 , \ldots,$ $\lambda_r)$. Then there exists a constant $n_0=n_0({\cal G}, \bm{\lambda})$ such that a $\cal G$-decomposition of $K_n^{[\lambda_1, \lambda_2 , \ldots, \lambda_r]}$ exists for every $n\geq n_0$ satisfying: $n(n-1)\equiv 0 \pmod{\beta(\cal G; \bm{\lambda})}$ and
$n-1\pmod{\alpha(\cal G; \bm{\lambda})}$.
\end{thm}

\section{Upper Bounds}
For the simplicity of illustration, when handling the general bounds of MCWCs, we only consider the special case of MCWC$(m,n,d,w)$. However, it is easy to see that our methods used can also be applied to the general case.
\subsection{Bounds from Spherical Codes}
We start with the definition of a spherical code. Different from the classic code, the spherical code is defined on the Euclidean space. A {\it spherical code} is a finite subset of $S(n)$, where $S(n):=\{\vvx\in R^n: \|\vvx\|=1\}$.  Here $\|*\|$ is the Euclidean norm. The {\it distance} between two codewords is defined by $d_E(\vc_1,\vc_2):=\|\vc_1-\vc_2\|$. However, to characterize the codeword separation in a spherical code, the {\it minimum angle $\phi$} or {\it the maximum cosine $s$} is often used instead of the Euclidean distance. The relation between these three parameters is
\[
s:=\cos \phi=1-\frac{d_{E}^2}{2}.
\]
We will generally use $s$ as the separation parameter. The largest size of an $n$-dimensional spherical code with maximum cosine $s$ is defined by $A_S(n,s)$.

When $s\leq 0$, the value of $A_S(n,s)$ has been determined completely. \cite{acz,dav,erd,ran,sar}:
\[
\begin{array}{ll}
        A_S(n,s)=\lfloor1-\frac{1}{s}\rfloor, & if~s\leq-\frac{1}{n};\\
        A_S(n,s)=n+1, & if~-\frac{1}{n}\leq s<0;\\
        A_S(n,0)=2n. &
\end{array}
\]

Before proceeding further, let us remark that, under a suitable mapping, a binary code can be viewed as a spherical code. Thus an upper bound on the cardinality of the spherical code serves as an upper bound for the binary code. This observation
can improve previous upper bounds in some cases.

Define $$\HH(n)=\{0,1\}^n,$$ $$\MM(m,n,w)=\{\vvx\in\HH(n):\vvx \cdot \vu_i=w\},$$ where $\vu_i=\ve_i\otimes \vj_n$, $\ve_i$ is the standard $m$-dimensional unit vector and $\vj_n$ is the $n$-dimensional all-one vector. Then any subset of  $\HH(n)=\{0,1\}^n$ is a binary code of length $n$ and any subset of $\MM(m,n,w)$ is an MCWC$(m,n,d,w)$ for some distance $d$.

Let $\Omega(*)$ denote the mapping $0\rightarrow1$ and $1\rightarrow-1$ from binary Hamming space to Euclidean space.
Then
\[
\Omega(\MM(m,n,w))=\{\vvx\in \Omega(\HH(n)):\vvx\cdot \vu_i=n-2w ~~for ~~1\leq i\leq m \}.
\]

For any point $\vvx\in \MM(m,n,w)$, $\vvx$ satisfies $(\Omega(\vvx)-\vvx_0)\cdot \vu_i =0$ and $\|\Omega(\vvx)-\vvx_0\|=r$, where
\[
\vvx_0=(1-\frac{2w}{n}) \vj_{mn},
\]
and
\[
r=2\sqrt{\frac{mw(n-w)}{n}}.
\]

Hence $\Omega(\MM(m,n,w))$ is a subset of the $(nm-m)$-dimensional hypersphere of radius $r$ centered at $\vvx_0$.

From the above analysis, we can get the following bound:

\begin{thm}\label{31}
\[
\begin{array}{ll}
M(m,n,2d,w) \leq \lfloor \frac{d}{b}\rfloor, & if ~b\geq\frac{d}{nm-m+1},\\

M(m,n,2d,w) \leq m(n-1)+1, & if ~0<b<\frac{d}{n},
\end{array}
\]
where
\[
b=d-\frac{mw(n-w)}{n}.
\]
\end{thm}

\begin{proof}
Let $\C$ be an MCWC$(m,n,2d,w)$. Translating $\Omega(\C)$ by $\vvx_0$ and scaling the radius by $1/r$, in accordance with the above analysis, yields an $(nm-m)$-dimensional spherical code with the maximum cosine $s=1-\frac{dn}{mw(n-w)}$. Thus
\[
\begin{array}{ll}
        M(m,n,2d,w) \leq A_S(m(n-1),s), & if ~s \geq -1;\\
        M(m,n,2d,w)=1, & if  ~s < -1.
\end{array}
\]

Using $A_S(mn-m,s)$ as an upper bound for $|\Omega(\C)|$ completes the proof.
\end{proof}

\begin{rem}
The first bound in Theorem \ref{31} is equivalent to the last Johnson-type bound (\ref{jon3}) and the second bound improves the Johnson-type bound of Proposition \ref{jonp} when $0<b<\frac{d}{n}$.
\end{rem}

\subsection{Plotkin-type Bounds}
The following proposition is well-known, while we provide a sketch of the proof for the sake of completeness.

\begin{prop}(\cite{var})\label{plo}
Let $\C$ be an $(n,d)$ code, then
\[
|\C|\leq \frac{d/2}{d/2-\sum_{i=1}^{n}f_{i}(1-f_{i})}
\]
provided that the denominator is positive, where $f_i$ denotes the proportion of codewords that have a $1$ in position $i$.
\end{prop}

\begin{proof}
The proof follows from the technique of double counting. On one hand,
\[
d_{av}=\frac{1}{M(M-1)}\sum_{c_1,c_2\in C}d(c_1,c_2)\geq d,
\]
where $M=|\C|$.
On the other hand,
\[
d_{av}=\frac{2M}{M-1}\sum_{i=1}^{n}f_i(1-f_i).
\]
By the double counting principle,
\[
\frac{2M}{M-1}\sum_{i=1}^{n}f_i(1-f_i)\geq d.
\]

\end{proof}

For MCWCs, we will have more restrictions concerning $f_i$, so we expect to get a better bound.

\begin{thm}
\begin{equation}\label{pplo}
M(m,n,2d,w)\leq \max\{\frac{d}{d-\sum_{i=1}^{mn}f_i(1-f_i)}\}
\end{equation}
where the maximum is taken over all $f_i $  $(1\le i \le mn)$ that satisfy the constraints below:
\begin{eqnarray}
f_1+f_2+\cdots+f_n=w\\
f_{n+1}+f_{n+2}+\cdots+f_{2n}=w\\
\vdots\\
f_{(m-1)n+1}+f_{(m-1)n+2}+\cdots+f_{mn}=w.
\end{eqnarray}
\end{thm}

\begin{proof}
The proof follows from the definition of MCWCs and Proposition \ref{plo}.
\end{proof}

\begin{cor}\label{cor}
\begin{equation}\label{jon}
M(m,n,2d,w)\leq \lfloor\frac{d}{b}\rfloor,
\end{equation}
where
\[
b=d-\frac{mw(n-w)}{n}.
\]
\end{cor}

\begin{proof}
To get an upper bound of MCWCs, we only need to determine the minimum value of $\sum_{i=1}^{n}f_{i}^{2}$, when $f_1+f_2+\cdots+f_n=w$.
We use the method of Lagrange Multiplier. Let $\gamma$ be an auxiliary variable. We consider the following function:
\[
g(f_1,f_2,\ldots,f_n,\gamma)=\sum_{i=1}^{n}f_{i}^{2}+\gamma(f_1+f_2+\cdots+f_n-w).
\]
Then
\[
\frac{\partial g}{\partial f_i}=2f_i+\gamma=0,
\]
\[
\frac{\partial g}{\partial \gamma}=\sum_{i=1}^{n}f_i-w=0.
\]
Thus when $f_i=\frac{w}{n}$, the original function will achieve the minimum value. Substituting $f_i$ with $\frac{w}{n}$ in the sum of (\ref{pplo}), we obtain (\ref{jon}).
\end{proof}

\begin{rem}
The bound (\ref{jon}) is equivalent to the  Johnson-type bound (\ref{mjon3}) of Proposition \ref{mjonp}, however when we  impose the additional constraint that $f_i$ must be multiples of $1/M$, the problem will be set in the discrete domain $\{0,1/M,2/M,\ldots,1\}$ instead of the continuous domain $[0,1]$. Similar with the above discussion of Corollary \ref{cor}, we will get an implicit expression of the upper bound.
\end{rem}

\begin{cor}
If $b>0$, then
\[
M(m,n,2d,w)\leq \lfloor d/b\rfloor,
\]
where
\[
\begin{array}{l}
b=d-\frac{mw(n-w)}{n}+\frac{nm}{M^2}\{Mw/n\}\{M(n-w)/n\},\\
M=M(m,n,2d,w),\\
\{x\}=x-\lfloor x\rfloor.
\end{array}
\]
\end{cor}

\subsection{Linear Programming Bounds}
Let $\C$ be an MCWC$(m,n,2d,w)$.
The distance distribution of  $\C$  can be defined as follows:
\[
A_{2i_1,2i_2,\ldots,2i_m}:=\frac{1}{|\C|}\sum_{\vc\in \C}A_{2i_1,2i_2,\ldots,2i_m}(\vc),
\]
where $A_{2i_1,2i_2,\ldots,2i_m}(\vc):=|\{\vc_1\in \C:(\vc_1\oplus\vc)\cdot \vu_j=2i_j\}|$, $\vu_j:=\ve_j\otimes \vj_n$, $\ve_j$ is the standard $m$-dimensional unit vector and $\vj_n$ is the $n$-dimensional all-one vector.

\begin{cor}\label{con}
Let $\C$ be an MCWC$(m,n,2d,w)$, then
\[
\sum_{i_1=0}^{w}\sum_{i_2=0}^{w}\cdots\sum_{i_m=0}^{w}Q_{k_1}(i_1)Q_{k_2}(i_2)\cdots Q_{k_m}(i_m)A_{2i_1,2i_2,\ldots,2i_m}\geq 0.
\]
\end{cor}

\begin{proof}
For $v=1,2,\ldots,m$, suppose $(X^{(v)};R_{0}^{(v)},\cdots,R_{w}^{(v)})$ is an association scheme with intersection numbers $p_{ijk}^{(v)}$, incidence matrices $D_i^{(v)}$, idempotents $J_i^{(v)}$, and eigenvalues $P_k^{(v)}(i)$, $Q_k^{(v)}(i)$. Then the Cartesian product $(X^{(1)}\times X^{(2)}\times \cdots\times X^{(m)}; R_{i_1\ldots i_m}=R_{i_1}^{(1)}\times\cdots\times R_{i_m}^{(m)},0\leq i_j\leq m$ for $1\leq j\leq m)$ is an association scheme with eigenmatrice $Q_{k_1}^{(1)}(i_1)Q_{k_2}^{(2)}(i_2)\cdots Q_{k_m}^{(m)}(i_m)$. Hence $\C$ is a code in the product of $m$ Johnson schemes. The result follows from Theorem \ref{lp}.
\end{proof}

\begin{thm}
\[
M(m,n,2d,w)\leq 1+\lfloor \max\sum_{i_1=0}^{w}\sum_{i_2=0}^{w}\cdots\sum_{i_m=0}^{w}A_{2i_1,\ldots,2i_m}\rfloor,
\]
where

\[A_{2i_1,\ldots,2i_m}\geq 0,\]
\[A_{2i_1,\ldots,2i_m}=0,~~for~\sum_{j=1}^{m}i_j<d;\]
and
\begin{equation}
\sum_{i_1=0}^{w}\sum_{i_2=0}^{w}\cdots\sum_{i_m=0}^{w}Q_{k_1}(i_1)Q_{k_2}(i_2)\cdots Q_{k_m}(i_m)A_{2i_1,2i_2,\ldots,2i_m}\geq 0.
\end{equation}
\end{thm}

\section{Asymptotic Lower Bounds}
In this section, we consider the asymptotic rate of $M(m,n,d,w)$ when $m$ is large, $n$ is a function of $m$, $d=\lfloor \delta mn\rfloor$ and $w=\lfloor \omega n\rfloor$ for $0<\delta, \omega<1$.  Define the value $\mu(\delta,\omega)$ as follows:
\[
\mu(\delta,\omega):=\limsup_{m \rightarrow \infty}\frac{\log_2 M(m,n,\lfloor\delta mn\rfloor,\lfloor\omega n\rfloor)}{mn}.
\]

In \cite{zhang},  Chee {\sl et al.} used the concatenation technique to give the following asymptotic lower bound.
\begin{prop}(\cite{zhang})\label{propalb}
For $\delta\leq 1/2$, we have
\[\mu(\delta,1/2)\geq 1-H(\delta),
\]
where $H(x)$ denotes the binary entropy function defined by
\[H(x):=-x\log_{2}x-(1-x)\log_2(1-x),\]
for all $0\leq x\leq 1$.
\end{prop}

In this section, we will generalise Proposition~\ref{propalb} and give a general form of the asymptotic lower bound. After that, we will give a generalised Gilbert-Varshamov bound for MCWCs and show that this classic method can provide a better bound.

The first bound follows from Proposition \ref{concatenation}. We choose the $q$-ary code that can achieve the Gilbert-Varshamov bound as outer codes. For convenience, we assume $\frac{1}{\omega}$ and $\delta mn$ are integers.

\begin{thm}
For $\omega\leq 1/2$ and $\delta\leq \max\{1/2,2\omega\}$, we have
\[
\mu_{c}(\delta,\omega)\geq\omega\log_2(\frac{1}{\omega})(1-H_{\frac{1}{\omega}}(\frac{\delta}{2\omega})),
\]
where $H_{q}(x):=x\log_q(q-1)-x\log_q x-(1-x)\log_q(1-x)$ for $0<x\leq\frac{q-1}{q}$.
\end{thm}

\begin{proof}
Applying Proposition \ref{concatenation}, we get $M(m,n,\delta mn,\omega n)\geq A_{\frac{1}{\omega}}(mwn,\frac{\delta mn}{2})$. Since $A_q(n,d)\geq q^{(1-H_q(d/n))n}$, then
\[
M(m,n,\delta mn,\omega n)\geq (\frac{1}{\omega})^{(1-H_{\frac{1}{\omega}}(\frac{\delta}{2\omega}))mwn},
\]
thus
\[
\mu_c(\delta,\omega)\geq \omega\log_2(\frac{1}{\omega})(1-H_{\frac{1}{\omega}}(\frac{\delta}{2\omega})).
\]
\end{proof}

\begin{rem}
Actually, there exist algebraic geometric codes leading to an asymptotic improvement upon Gilbert-Varshamov bound when the alphabet size $q\geq 49$ \cite{TVZ,X}. Since the improvement is slight, we still use the Gilbert-Varshamov bound for the sake of simplicity.
\end{rem}

The Gilbert-Varshamov bound is one of the most well-known and fundamental results in coding theory. In fact, it can be easily applied to various kinds of codes. For MCWC$(m,n,2d,w)$, the volume of the Hamming ball of radius $2d-1$ is
\[
\sum_{i_1+i_2+\ldots+i_m\leq d-1}{w \choose i_1}{n-w \choose i_1}\cdots{w \choose i_m}{n-w \choose i_m}.
\]

\begin{thm}
For $\omega\leq 1/2$ and $\delta\leq \max\{1/2,2\omega\}$, we have
\[
\mu_{GV}(\delta,\omega)\geq H_2(\omega)-\omega H_2(\frac{\delta}{2\omega})-(1-\omega)H_2(\frac{\delta}{2(1-\omega)}).
\]
\end{thm}

\begin{proof}
Since
\[
M(m,n,\delta mn,\omega n)\geq \frac{{n \choose \omega n}^m}{\sum_{i_1+i_2+\ldots+i_m\leq \frac{\delta mn}{2}-1}{\omega n \choose i_1}{(1-\omega)n \choose i_1}\cdots{\omega n \choose i_m}{(1-\omega)n \choose i_m}}
\]
\[
\geq \frac{{n \choose \omega n}^m}{\sum_{0\leq i\leq \frac{\delta mn}{2}}{\omega mn \choose i}{(1-\omega)mn \choose i}},
\]
we have
\[
\mu_{GV}(\delta,\omega)\geq\frac{\log_2\frac{2^{nmH_2(\omega)}}{2^{\omega nmH_2(\frac{\delta}{2\omega})}2^{(1-\omega)mnH_2(\frac{\delta}{2(1-\omega)})}}}{mn}
\]
\[
\geq H_2(\omega)-\omega H_2(\frac{\delta}{2\omega})-(1-\omega)H_2(\frac{\delta}{2(1-\omega)}).
\]
\end{proof}

At the end of this section, we compare the two bounds given above and show that the generalised Gilbert-Varshamov bound offers a better one.

\begin{thm}{\label{com}}
\[
\mu_{GV}(\delta,\omega)\geq\mu_{c}(\delta,\omega),
\]
equality holds only when $w=\frac{1}{2}$ or $\delta=2(\omega-\omega^2)$.
\end{thm}

\begin{proof}
Let
\[
f(\delta,\omega)=\mu_{GV}(\delta,\omega)-\mu_{c}(\delta,\omega)
\]
\[
=H_2(\omega)-(1-\omega)H_2(\frac{\delta}{2(1-\omega)})+\frac{\delta}{2}\log_2(\frac{1}{\omega}-1)-(1-\omega)\log_2(1-\omega).
\]
For simplicity, letting $x=\frac{\delta}{2}$, we get
\[
f(x,\omega)=-(2-2\omega-x)\log_2(1-\omega)+x\log_2(\frac{x}{\omega})+(1-\omega-x)\log_2(1-\omega-x).
\]
We will derive the proof by considering two cases of $\omega\leq\frac{1}{4},x\leq\omega$ and $\frac{1}{4}<\omega\leq \frac{1}{2},x\leq\frac{1}{4}$ separately.
\begin{enumerate}
  \item[(a)] $\omega\leq\frac{1}{4},x\leq\omega$.

        When $x=0$, $f(0,\omega)=-(1-\omega)\log_2(1-\omega)>0$.

        When $x=\omega$, $f(\omega,\omega)=(3\omega-2)\log_2(1-\omega)-(2\omega-1)\log_2(1-2\omega)$.
        We want to show $f(\omega,\omega)\geq 0$.
        Since $f(0,0)=0$ and $f(\frac{1}{4},\frac{1}{4})=2-\frac{5}{4}\log_2 3>0$, we need to show that $g(\omega)=f(\omega,\omega)$ is monotonely increasing.
        \[
        g^{'}(\omega)=3\log_2(1-\omega)-2\log_2(1-2\omega)+\frac{\omega}{\omega-1},
        \]
        \[
        g^{''}(\omega)=\frac{\omega(3-2\omega)}{(\omega-1)^2(1-2\omega)}>0.
        \]
        Since $g^{'}(0)=0$ and $g^{'}(\frac{1}{4})=\frac{5}{3}+3\log_2(\frac{3}{4})>0$, we get $g^{'}(\omega)\geq 0$, thus $f(\omega,\omega)\geq 0$.

Moreover $\frac{\partial f(x,\omega)}{\partial x}=\log_2\frac{x(1-\omega)}{\omega(1-\omega-x)}=0$, we get $x=\omega-\omega^2$. Since $f(\omega-\omega^2,\omega)=0$, with the above analysis, we get $f(\delta,\omega)\geq0$.

  \item[(b)] $\frac{1}{4}<\omega\leq \frac{1}{2},x\leq\frac{1}{4}$.

   When $x=0$,  $f(0,\omega)=-(1-\omega)\log_2(1-\omega)>0$.

When $x=\frac{1}{4}$, $f(\frac{1}{4},\omega)=-(\frac{7}{4}-2\omega)\log_2(1-\omega)+\frac{1}{4}\log_2(\frac{1}{4\omega})+(\frac{3}{4}-\omega)\log_2(\frac{3}{4}-\omega).$
We want to show $f(\frac{1}{4},\omega)\geq 0$.
Since $f(\frac{1}{4},\frac{1}{4})=-\frac{5}{4}\log_2(\frac{3}{4})-\frac{1}{2}>0$ and $f(\frac{1}{4},\frac{1}{2})=0$, we show the function $f(\frac{1}{4},\omega)$ is monotonely decreasing.
\[
f^{'}(\frac{1}{4},\omega)=\frac{1}{\ln2}(2\ln(1-\omega)-\ln(\frac{3}{4}-\omega)+1-\frac{1}{4(1-\omega)}-\frac{1}{4\omega}),
\]
\[
f^{''}(\frac{1}{4},\omega)=\frac{1}{\ln2}(\frac{}{}+\frac{1-2\omega}{4\omega^2(1-\omega)^2})\geq0.
\]
Since $f^{'}(\frac{1}{4},\frac{1}{4})=\frac{1}{\ln2}(\ln(\frac{9}{8})-\frac{1}{3})<0$ and $f^{'}(\frac{1}{4},\frac{1}{2})=0$, we get $f^{'}(\frac{1}{4},\omega)\leq0$, thus $f(\frac{1}{4},\omega)\geq0$.

The remainder of the proof is the same as the first case. Then, we have already proven this theorem.
\end{enumerate}
\end{proof}

\section{Two Infinite Classes of Optimal Codes}

In \cite{Cheegraph}, Chee {\sl et al.}  demonstrated that certain Johnson-type bounds
are asymptotically exact for constant-composition codes, nonbinary constant-weight codes and MCWCs by constructing several infinite classes of optimal codes achieving these bounds. Especially, for MCWCs they showed that the bound~(\ref{jon1})  is asymptotically exact for distance $2mw-2$.

\begin{thm}[Chee {\sl et al.} \cite{Cheegraph}]\label{Cheegraphdecom} Fix $m$ and $w$. There exits an integer $n_0$ such that
\begin{equation*}
M(m,n,2mw-2,w)=\frac{n(n-1)}{w^2}
\end{equation*}
for all $n\geq n_0$ satisfying $n-1\equiv 0 \pmod{w^2}$.
\end{thm}

In this section, we will generalize Theorem~\ref{Cheegraphdecom} to the case where the weight $w_i$ may not be equal. We determine the value of $T(w_1,n;w_2,n;\ldots;w_m,n;2\sum_{i=1}^m w_i-2)$ for some modulo classes of $n$ when $n$ is sufficiently large.
We also establish the connection between $\alpha$-resolvable BIBDs and MCWCs and employ Theorem~\ref{mgraphdecom} to establish the asymptotic existence of a class of $\alpha$-resolvable BIBDs. As a consequence, we prove that the bound (\ref{jon3}) is asymptotically exact for distance $2mw-w$.


\subsection{Optimal MCWCs with Distance $2\sum_{i=1}^m w_i-2$}

Let $w_1\geq w_2\geq \cdots \geq w_m$ be nonnegative integers. Let $w=\sum_{i=1}^m w_i$. The Johnson-type bound~(\ref{jon1}) shows that
\begin{equation*}
\begin{split}
T(w_1,n;w_2,n;\ldots;w_m,n;2w-2) \leq & \begin{cases} \frac{n(n-1)}{w_1(w_1-1)}, \textup{\ \ if $w_1>w_2$;} \\
 \frac{n(n-1)}{w_1^2}, \textup{\ \ if $w_1=w_2$.}\\
\end{cases}
\end{split}
\end{equation*}
We will show that this bound is asymptotically tight. To apply Theorem~\ref{graphdecom}, we first define the  family of edge-colored digraphs $\cal G$. We use the $m^2$ ordered pairs from $[m]$ as colors. Define  $\overline{w}=[w_1,w_2\ldots,w_m]$. Let $G(\overline{w})$ be the digraph with vertex set

\begin{equation}\label{partition0}
V(G(\overline{w}))= W_1 \cup W_2 \cup \cdots \cup W_{m}
\end{equation}
where $W_i$'s are disjoint vertex sets with $|W_i|=w_i$. Here, for  all distinct
$x, y \in V(G(\overline{w}))$, there is an edge from $x$ to $y$ of color $(i, j)$ where $i$ and $j$ are
such that $x \in W_i$ and $y \in W_j$.
Then in the graph $G(\overline{w})$,  there are $w_iw_j$ edges colored $(i,j)$ with $i\not = j$, and $w_i(w_i-1)$ edges colored $(i,i)$.
For $i, j \in [m]$, let $G_{ij}$ be a digraph with two vertices
and one directed edge of color $(i,j)$. To define
${\cal G}(\overline{w})$, we consider the following two cases depending on whether $w_1=w_2$:
\begin{enumerate}
  \item When $w_1 > w_2$, we have $w_1(w_1-1)\geq w_1w_2$. Let $r$ be the largest integer such that $w_1-1 = w_2 = \cdots = w_r$. Then set ${\cal G}(\overline{w}) = \{G(\overline{w})\}\cup \{G_{ij}: (i,j)\in ([m]\times [m]) \backslash \{(1,i),(i,1): 1\leq i\leq r\}\}$.
  \item When $w_1 = w_2$, we have $w_1w_2>w_1(w_1-1)$. Let $r$ be the largest integer such that $w_1 = \cdots = w_r$. Then set ${\cal G}(\overline{w}) = \{G(\overline{w})\}\cup \{G_{ij}: (i,j)\in ([m]\times [m]) \backslash \overline{[r] \choose 2}\}$.
\end{enumerate}

\begin{prop}
Suppose that a $G(\overline{w})$-decomposition of $K_n^{(m^2)}$ exists. Then
\begin{equation*}
\begin{split}
T(w_1,n;\ldots;w_m,n;2w-2) = & \begin{cases} \frac{n(n-1)}{w_1(w_1-1)}, \textup{\ \ if $w_1>w_2$;} \\
 \frac{n(n-1)}{w_1^2}, \textup{\ \ if $w_1=w_2$.}\\
\end{cases}
\end{split}
\end{equation*}
\end{prop}

\begin{proof}
Let $V$ be the vertex set of $K_n^{(m^2)}$ and $\cal F$ be the $G(\overline{w})$-decomposition. Let $X=\{1,2,\ldots,m\}\times V$. The code is constructed in $2^X$. For each $F\in \cal F$ isomorphic to $G(\overline{w})$, there is a unique partition of the vertex set $V(F)=\cup_{i=1}^m S_{i}$ so that the edge from $x$ to $y$ in $F$ has color $(i,j)$
if $x \in S_i$ and $y \in S_j$. Construct a codeword $\mathbf{u}$ such that $\mathbf{u}_{(i,x)} = 1$ if $x \in S_i$, and $\mathbf{u}_{(i,x)} = 0$ otherwise. Since $|S_i|=w_i$, this code is an MCWC$(w_1,n;\ldots;w_m,n;d)$ with some distance $d$. Noting that every colored edge appears at most once in the member of $\cal F$ isomorphic to $G(\overline{w})$, we have $|\rm{supp}(\mathbf{u})\cap \rm{supp}(\mathbf{v})|\leq 1$ for any two codewords $\mathbf{u}$ and $\mathbf{v}$. Thus this code has distance $2w-2$.

Finally, let $m$ be the number of digraphs in $\cal F$ isomorphic to $G(\overline{w})$. It is easy to see that $m=\frac{n(n-1)}{w_1(w_1-1)}$ if $w_1>w_2$ and $m=\frac{n(n-1)}{w_1^2}$ otherwise.
\end{proof}

Noting that $m_{(i,j)}(G(\overline{w}))=w_iw_j$, $i\not=j$, $m_{(i,i)}(G(\overline{w}))=w_i(w_i-1)$ and $m_{(i,j)}(G_{ij})=1$, we have
\begin{equation*}
\begin{split}
\beta({\cal G}(\overline{w})) = & \begin{cases} w_1(w_1-1), \textup{\ \ if $w_1>w_2$;} \\
 w_1^2, \textup{\ \ if $w_1=w_2$.}\\
\end{cases}
\end{split}
\end{equation*}

Since $\textup{in}_{(i,j)}(G(\overline{w}))=w_j$, $\textup{out}_{(i,j)}(G(\overline{w}))=w_i$ for any $i \not =j$, $\textup{in}_{(i,i)}(G(\overline{w}))=\textup{out}_{(i,i)}(G(\overline{w}))=w_i-1$, it is easy to check that
\begin{equation*}
\begin{split}
\alpha({\cal G}(\overline{w})) = & \begin{cases} w_1(w_1-1), \textup{\ \ if $w_1>w_2$;} \\
 w_1, \textup{\ \ if $w_1=w_2$.}\\
\end{cases}
\end{split}
\end{equation*}

Then applying Theorem~\ref{graphdecom}, we can obtain the following result.

\begin{thm}Let $w_1\geq w_2\geq \cdots \geq w_m$ be nonnegative integers and $w=\sum_{i=1}^m w_i$. There exits an integer $n_0$ such that
\begin{equation*}
\begin{split}
T(w_1,n;\ldots;w_m,n;2w-2) = & \begin{cases} \frac{n(n-1)}{w_1(w_1-1)}, \textup{\ \ if $w_1>w_2$;} \\
 \frac{n(n-1)}{w_1^2}, \textup{\ \ if $w_1=w_2$.}\\
\end{cases}
\end{split}
\end{equation*}
for all $n\geq n_0$ satisfying $n-1\equiv 0 \pmod{w_1(w_1-1)}$ if $w_1 > w_2$, or $n-1\equiv 0\pmod{w_1^2}$ otherwise.
\end{thm}

\subsection{Optimal MCWCs with Distance $2mw-w$}
We first establish a connection between $\alpha$-resolvable BIBDs and optimal MCWCs.
\begin{prop}\label{alphar2MCWC} If there exits an $\alpha$-resolvable BIBD$(v,k,\lambda)$, then $M(m,n,d,w)=v$, where $m=\frac{\lambda(v-1)}{\alpha(k-1)}$, $n=\frac{\alpha v}{k}$, $d=2(\frac{\lambda(v-1)}{k-1}-\lambda)$, and $w=\alpha$.
\end{prop}

\begin{proof}  The Johnson-type bound~(\ref{jon3}) shows that $M(m,n,d,w)\leq v$ where $m=\frac{\lambda(v-1)}{\alpha(k-1)}$, $n=\frac{\alpha v}{k}$, $d=2(\frac{\lambda(v-1)}{k-1}-\lambda)$, and $w=\alpha$.

Let $(X, \cal B)$ be an $\alpha$-resolvable BIBD$(v,k,\lambda)$. Since there are $\frac{\lambda(v-1)}{\alpha(k-1)}$ $\alpha$-parallel classes in $\cal B$,  each of which consists of $\frac{\alpha v}{k}$ blocks, we can arrange all the blocks in an $m\times n$ array with  $m=\frac{\lambda(v-1)}{\alpha(k-1)}$ and $n=\frac{\alpha v}{k}$, such that the blocks in each row form an $\alpha$-parallel class. Now, for each point $x\in X$, construct a codeword $\mathbf{u}$ with $\mathbf{u}_{(i,j)}=1$ if the block in the entry $(i,j)$ contains $x$, and $\mathbf{u}_{(i,j)}=0$ otherwise. Since each point appears in $\alpha$ times in each row, the code constructed above is an MCWC$(m,n,d,\alpha)$ of size $v$ for some distance $d$. Since any two distinct points of $X$ appear together in exactly $\lambda$ blocks, the supports of any two codewords intersect in exactly $\lambda$ points. Thus the code has distance $d=2(mw-\lambda)=2(\frac{\lambda(v-1)}{k-1}-\lambda)$.
\end{proof}

In the remaining of this subsection, we employ Theorem~\ref{mgraphdecom} to show that when $\alpha=\lambda$ and $k\mid \alpha$, an $\alpha$-resolvable BIBD$(v,k,\lambda)$ exists for all sufficient $v$ with $v\equiv 1\pmod{k-1}$.

We first define the family of edge-$r$-colored digraphs $\cal G$ with $r=k^2-k$. We use the $(k-1)^2$ ordered pairs from $[k-1]$ and the $k-1$ singletons $(i)$, $i=1, 2, ..., k-1$ as colors. Let $\bm{\lambda}$ be a vector of length $k^2-k$ with each entry being $\lambda$. For each $(k-1)$-tuple $\mathbf{t}=(t_1,t_2\ldots,t_{k-1})$ of nonnegative integers summing to $k$, let $G(\mathbf{t})$ be the digraph with $k+1$ vertices

\begin{equation}\label{partition}
V(G(\mathbf{t}))=\{w\} \cup T_1 \cup T_2 \cup \cdots \cup T_{k-1}
\end{equation}
where $T_i$'s are disjoint vertex sets with $|T_i|=t_i$ and $w$ is another vertex not in any $T_i$. Here, for  all distinct
$x, y \in V(G(\mathbf{t}))$, there is an edge from $x$ to $y$ of color $(i, j)$ where $i$ and $j$ are
such that $x \in T_i$ and $y \in T_j$, and an edge of color $(i)$ from the special vertex
$w$ to each $x$ in $T_i$. Let $\cal G$ be the collection of all such $G(\mathbf{t})$.

\begin{prop}
If there exits a $\cal G$-decomposition of the edge-$r$-colored $K_m^{[\lambda, \lambda, \ldots, \lambda]}$ with $r=k^2-k$ and $m=\frac{v-1}{k-1}$, then a $\lambda$-resolvable BIBD$(m(k-1)+1,k,\lambda)$ exists.
\end{prop}

\begin{proof}
Let $V$ be the vertex set of $K_m^{[\lambda, \lambda, \ldots, \lambda]}$ and let $X=\{\infty\} \cup (V\times [k-1])$. Let
$B_x=\{\infty\}\cup (\{x\}\times \{1,2, \ldots, k-1\})$, ${\cal B}=\{B_x: x \in V\}$.
The elements $V$ will be used to index the $\lambda$-parallel classes, which are denoted as ${\cal P}_x$, $x \in V$; $B_x$ will be in ${\cal P}_x$. For each $F \in \cal F$, there will be a unique
partition of the $k+1$ vertices $V(F)\subset V$ as
$$V(F)=\{w\} \cup S_1 \cup S_2 \cup \cdots \cup S_{k-1}$$
as in (\ref{partition}). Let
$$A_F=\cup_{i=1}^{k-1} S_i \times \{i\};$$
we take $\lambda$ copies of this block in the parallel class ${\cal P}_w$. Let ${\cal A}=\{A_F : F \in {\cal F}\}$ and let ${\cal B}^{\lambda}$ be a multi-set containing each member of $\cal B$ $\lambda$ times. It is easy
to check that $(X, {\cal A} \cup {\cal B}^{\lambda})$ is a $((k-1)m+1, k, \lambda)$-BIBD, and that each ${\cal P}_w$
is a $\lambda$-parallel class. For example, the $\lambda$ blocks in ${\cal P}_w$ that contains a
point $(y,i), y\not=w$ are $A_F$'s where $F$'s are the graphs in $\cal F$ that contain the edge
of color $(i)$ from $w$ to $y$.
\end{proof}

With the same argument as that in the proof of \cite[Therem~10.1]{LW}, one can show that $m(m-1)(\lambda, \lambda, \ldots, \lambda)$ is an integral linear combination of the vectors $\mu(G(\mathbf{t}))$, $G(\mathbf{t})\in \cal G$, and $(m-1)(\lambda, \lambda, \ldots, \lambda)$ is an integral linear combination of the vectors $\tau(x, G(\mathbf{t}))$ as $x$ ranges over all vertices of all digraphs $G(\mathbf{t})\in \cal G$. Thus, the two conditions of Theorem~\ref{mgraphdecom} are satisfied. Applying this theorem we can obtain the following result.

\begin{thm}\label{alpharBIBD}
Given positive integers $k$ and $\lambda$ with $k \mid \lambda$, there exits a constant $m_0=m_0(k,\lambda)$ such that a $\lambda$-resolvable BIBD$(m(k-1)+1,k,\lambda)$ exists for all $m\geq m_0$.
\end{thm}

Combining Proposition~\ref{alphar2MCWC} and Theorem~\ref{alpharBIBD}, we can get the following result.

\begin{thm} Given positive integers $k$ and $w$ with $k \mid w$, there exits a constant $m_0=m_0(k,w)$ such that
$$M(m,n,2(mw-w),w)=m(k-1)+1$$
with $n=w(m(k-1)+1)/k$ for all $m\geq m_0$.
\end{thm}

\section{Optimal MCWCs with Weight Four}

In \cite{Cheeoptimal}, the authors determined the maximum size of MCWCs for total weight less than or equal to four, except when $m=2$, $w_1=w_2=2$, $d = 6$ and $n_1 \leq n_2 \leq 2n_1-1$, with both $n_1$ and $n_2$ being odd. We consider this open class in this section. The Johnson-type bound~(\ref{jon1}) yields that:

\begin{lem} Let $n_1,n_2$ be two odd integers with $0< n_1 \leq n_2 \leq 2n_1-1$. Then
T$(2,n_1;2,n_2;6) \leq \lfloor \frac{n_2(n_1-1)}{4}\rfloor$.
\end{lem}


We will show  the above bound can be achieved for most cases. Firstly, we introduce a new combinatorial structure and establish the connection between such a structure and the optimal MCWC$(2,n_1;2,n_2;6)$.

\subsection{Skew Almost-resolvable Squares}


Let $V$ be a set of $v$ points and $S$ be a set of $s$ points. A {\it skew almost-resolvable square}, denoted SAS$(s,v)$, is an $s\times s$ array, where the rows and the columns are indexed by the elements of $S$, and each cell is either empty or contains a pair of points from $V$, such that:
\begin{enumerate}
  \item for every two cells $(i,j)$ and $(j,i)$ with $i \not =j $ at most one is filled;
  \item the cells on the diagonal are all empty;
  \item no pair of points from $V$ appears in more than one cell;
  \item for each $i\in S$, the pairs in row $i$ together with those in column $i$ form a partition of $V\backslash\{x\}$ for some $x\in V$.
\end{enumerate}

\begin{prop}\label{MCWC2SAS}
Let $v\equiv 1\pmod {4}$ and $s\equiv 1\pmod {2}$ with  $v \leq s \leq 2v-1$. There exists an MCWC$(2,v;2,s;6)$ of size $\frac{s(v-1)}{4}$ if and only if an SAS$(s,v)$ exists.
\end{prop}

\begin{proof} Let $A$ be an SAS$(s,v)$ on $V$ with rows and columns indexed by $S$. We may assume that $V$ and $S$ are distinct. Let $X=V\cup S$. The code is constructed in $2^X$.
For each filled cell $(i,j)$ of $A$ with $A(i,j)=\{a,b\}$, construct a codeword $\vu$ where $\vu_x=1$ if $x\in\{a,b,i,j\}$, and $\vu_x=0$ otherwise. Then we  get an MCWC$(2,v;2,s;d)$ for some distance $d$. Note that Properties 1), 3) and 4) guarantee that any pair of points of $X$ appear in at most one codeword's support. The supports of any two distinct codewords $\vu$ and $\vv$ intersect in at most one point and then the code  has distance $6$. According to Property 4), for each $i \in S$, there are $\frac{v-1}{2}$ cells filled in row $i$ and column $i$. Thus we have $\frac{s(v-1)}{4}$ cells filled in total and the code has size $\frac{s(v-1)}{4}$.

Conversely, let $X=X_1\cup X_2$ with $|X_1|=v$ and $|X_2|=s$. Let $\C$ be an MCWC$(2,v;2,s;6)$ of size $\frac{s(v-1)}{4}$ in $2^X$. Construct an $s\times s$ array with rows and columns indexed by the elements of $S$. For each codeword $\vu\in \cal C$ with ${\rm supp}(\vu)=\{a,b,i,j\}$, $a,b\in X_1$ and $i,j\in X_2$, fill in the cell $(i,j)$ with the pair $\{a,b\}$.  It is easy to check that this array is an SAS$(s,v)$.
\end{proof}

In the above definition of SASs, if we replace the condition 4) by the following one, we get the definition of SAS$^*(s,v)$s.
\begin{enumerate}
  \item[4)'] there exits an $i_0\in S$ such that for each $i\in S\backslash \{i_0\}$, the pairs in row $i$ and column $i$ form a partition of $V\backslash\{x\}$ for some $x\in V$; the pairs in row $i_0$ and column $i_0$ form a partition of $V\backslash\{x,y,z\}$ for some distinct $x,y,z\in V$.
\end{enumerate}

Similarly, we have the following result, the proof of which is exactly the same as that of Proposition~\ref{MCWC2SAS} and we omit it here.

\begin{prop}\label{MCWC2SAS*}
Let $v\equiv 3\pmod {4}$ and $s\equiv 1\pmod {2}$ with  $v \leq s \leq 2v-1$. There exists an MCWC$(2,v;2,s;6)$ of size $\lfloor\frac{s(v-1)}{4}\rfloor$ if and only if an SAS$^*(s,v)$ exists.
\end{prop}

In the following, we will discuss a useful construction method, i.e.,  frame  construction,  which will allow us to construct infinite families of SASs and SAS$^*$s.

Let $V$ be a set of $v$ points and $S$ be a set of $s$ points. Let $\{H_1,H_2,\ldots,H_n\}$ be a partition of $V$ with $|H_i|=h_i$ and $\{S_1,S_2,\ldots,S_n\}$ be a partition of $S$ with $|S_i|=s_i$. A {\it skew frame-resolvable square} (SFS) of type $\{(s_i,h_i):1\leq i \leq n\}$ is an $s\times s$ array, where the rows and the columns are indexed by the elements of $S$, and each cell is either empty or contains a pair of points from $V$, such that:
\begin{enumerate}
  \item for every two cells $(i,j)$ and $(j,i)$ with $i \not =j $ at most one is filled;
  \item the subarray indexed by $S_i\times S_i$ is empty, and it is called {\it hole};
  \item no pair of points from $V$ appears in more than one cell;
  \item no pair of points from $H_i$ appears in any cell;
  \item for each $l\in S_i$, the pairs in row $l$ together with those in column $l$ form a partition of $V\backslash H_i$.
\end{enumerate}
We will use an exponential notation $(s_1,g_1)^{n_1} \cdots (s_n,g_n)^{n_t}$ to indicate that there are $n_i$ occurrences of $(s_i, g_i)$ in the partitions.

We can use GDDs to give the recursive construction of SFSs.

\begin{cons}\label{WFC} Let $(X,{\cal G},{\cal B})$ be a GDD, and let
$s,v : X \rightarrow \bbZ^{+} \cup \{0\}$ be two weight functions on $X$.
Suppose that for each block $B \in \cal B$, there exists an SFS
of type $\{(s(x),v(x)) : x \in B\}$. Then there is an SFS of type $\{
(\sum_{x \in G} s(x), \sum_{x \in G} v(x)) : G \in \cal G \}$.
\end{cons}

\begin{proof} For each $x\in X$, let $S(x)$ be an index set of $s(x)$ elements, where $S(x)$ and $S(y)$ are disjoint for any $x\not =y \in X$. For each $B \in \B$, we construct an SFS of type $\{(s(x),v(x)) : x \in B\}$ ${\cal A}_B$ on $\cup_{x \in B} (\{x\}\times \{1, 2,\ldots,v(x)\})$ and index its rows and columns using the elements of the set $\cup_{x\in\B}S(x)$.

Denote $S=\cup_{x\in X} S(x)$ and $V=\cup_{x \in X} (\{x\}\times \{1, 2,\ldots,v(x)\})$. We construct the requisite SFS $\cal A$ on $V$ and index its rows and columns by $S$ as follows: for each cell of $\cal A$ indexed by $(\alpha,\beta)$, if $\alpha \in S(x)$, $\beta \in S(y)$ with $x\not =y$ and there exists a block $B \in \B$ containing $x,y$, then we place the entry from ${\cal A}_B$ indexed by $(\alpha, \beta)$ in the cell of $\cal A$; otherwise the cell is empty.

For each $G_i\in \cal G$, denote $S_i=\cup_{x\in G_i} S(x)$ and $H_i=\cup_{x\in G_i} (\{x\}\times \{1, 2,\ldots,v(x)\})$. It is easy to check that Properties 1)~--~4) in the definition of SFSs are satisfied. Now, for each $\alpha \in S_i$, we consider the pairs in row $\alpha$ and column $\alpha$. Assume that $\alpha \in S(x)$ for some $x\in G_i$. Since for each $y\not \in G_i$, there exits a unique block containing both $x$ and $y$, the set $\{B\backslash\{x\}: x\in B \in \cal B\}$ forms a partition of $X\backslash G_i$. Note that for each  ${\cal A}_B$ with $x\in B$, the pairs in row $\alpha$ and column $\alpha$ of ${\cal A}_B$ form a partition of $\cup_{y\in B, y\not= x} (y\times \{1,2,\ldots, v(y)\})$. Then the pairs in row $\alpha$ and column $\alpha$ in $\cal A$ form a partition of $$\bigcup_{x\in B, B\in \cal B}\left( \bigcup_{y\in B, y\not= x} \left(y\times \{1,2,\ldots, v(y)\}\right) \right)=\bigcup_{y\in X\backslash G_i} \left(y\times \{1,2,\ldots, v(y)\}\right)=V\backslash H_i.$$

Thus we have proved that $\cal A$ is an SFS of type $\{ (\sum_{x \in G} s(x), \sum_{x \in G} v(x)) : G \in \cal G \}$.
\end{proof}

Let $V$ be a set of $v$ points and $S$ be a set of $s$ points. Let $W$ be a subset of $V$ with $|W|=w$ and $T$ be a subset of $S$ with $|T|=t$. A {\it holey skew almost-resolvable square}, denoted HSAS$(s,v;t,w)$, is an $s\times s$ array, where the rows and the columns are indexed by the elements of $S$, and each cell is either empty or contains a pair of points from $V$, such that:
\begin{enumerate}
  \item for every two cells $(i,j)$ and $(j,i)$ with $i \not =j $ at most one is filled;
  \item the subarray indexed by $T\times T$ is empty, and it is called {\it hole};
  \item no pair of points from $V$ appears in more than one cell;
  \item no pair of points from $W$ appears in any cell;
  \item for each $t\in T$, the pairs in row $t$ together with those in column $t$ form a partition of $V\backslash W$;
  \item for each $l\in S\backslash T$, the pairs in row $l$ and column $l$ form a partition of $V\backslash \{x\}$ for some $x\in V$.
\end{enumerate}

The following result is simple but useful  in our constructions.

\begin{prop}
Suppose that there exist both an HSAS$(s,v;t,w)$ and an SAS$(t,w)$. Then an SAS$(s,v)$ exists.
\end{prop}

In the following, we show how to construct SASs from SFSs.

\begin{cons}\label{BFC}[Basic Frame Construction] Suppose that there exists an SFS of type $\{(s_i,h_i): 1\leq i \leq n\}$.
Let $s=\sum_{i=1}^n s_i$ and $v=\sum_{i=1}^n h_i$. If for each $1\leq i \leq n-1$ there exists an HSAS$(s_i+e, h_i+w; e,w)$, furthermore,
\begin{enumerate}
  \item[(1)] if there exits an HSAS$(s_n+e, h_n+w;e,w)$, then an HSAS$(s + e, v + w;e,w)$ exists;
  \item[(2)] if there exits an SAS$(s_n+e, h_n+w)$, then an SAS$(s + e, v + w)$ exists;
  \item[(3)] if there exits an SAS$^*(s_n+e, h_n+w)$, then an SAS$^*(s + e, v + w)$ exists.
\end{enumerate}
\end{cons}

\begin{proof}  Let $A$ be an SFS of type $\{(s_i,h_i): 1\leq i \leq n\}$ on $V=\cup_{i=1}^s H_i$ with rows and columns indexed by $S$.
Let $W$ be a set of size $w$, disjoint from $V$, and take our new point set to be $V\cup W$. Now,
add $e$ new rows and columns. For each $1\leq i \leq n-1$, fill the $s_i \times s_i$ subsquare together with the $e$ new rows and columns with a copy of the HSAS$(s_i+e, h_i+w; e,w)$ on $H_i \cup W$, such that the intersection of the new rows and columns forms a hole. Then, fill the $s_n \times s_n$ subsquare together with the $e$ new rows and columns with a copy of the HSAS$(s_n+e, h_n+w;e,w)$ \big(SAS$(s_n+e, h_n+w; e,w)$, SAS$^*(s_n+e, h_n+w; e,w)$\big). It is routine to check that the resultant square is an HSAS$(s + e, v + w;e,w)$ \big(SAS$(s + e, v + w)$, SAS$^*(s + e, v + w)$\big).
\end{proof}

\subsection{Determining the Value of T$(2,n_1;2,n_2;6)$}

Suppose $\vu\in\bbZ_2^X$ is a codeword of an
MCWC$(w_1,n_1;w_2,n_2;d)$.
We can represent $\vu$ equivalently as a $4$-tuple $\langle
a_1, a_2, a_3,a_4\rangle\in X^4$, where
$\vu_{a_1}=\vu_{a_2}=\vu_{a_3}=\vu_{a_{4}} =1$.
Throughout this section, we shall often represent codewords of
MCWCs in this form.

\begin{lem}\label{MCWCsmall}
Let $n_1\in\{3,5,7,9,11,13,15,17,19,21,25,29,33,37\}$, $n_1\leq n_2\leq 2n_1-1$ and $n_2$ be odd. Then
\begin{enumerate}
  \item T$(2,n_1;2,n_2;6)=\lfloor \frac{n_2(n_1-1)}{4}\rfloor$, except for $(n_1,n_2)=(5,7)$; furthermore
  \item T$(2,5;2,7;6) = 6$.
\end{enumerate}
\end{lem}

\begin{proof} The upper bound T$(2,5;2,7;6) \leq 6$ can be found in \cite{table}. Codes achieving the upper bounds are constructed as follows.

For $3\leq n_1 \leq 9$, let $X=\{0,1,2,\ldots,n_1+n_2-1\}$. $X$ can be partitioned as $X=X_1\cup X_2$  with $X_1=\{0,1,\ldots,n_1-1\}$ and $X_2=\{n_1,n_1+1,\ldots,n_1+n_2-1\}$. The desired codes are constructed on $X$ and the codewords are listed in Table~\ref{tab.smallMCWC}.

For $n_1\in\{13,17,21,25,29,33,37\}$, the codes are constructed in the Appendix.

For $n_1\in\{11,15,19\}$ and $n_1\leq n_2\leq 2n_1-3$, take an HSAS$(n_2,n_1;3,3)$ from the Appendix and fill in the hole with an SAS$^*(3,3)$ (which is equivalent to an MCWC$(2,3;2,3;6)$ and has been constructed above) to obtain an SAS$^*(n_2,n_1)$. According to Proposition~\ref{MCWC2SAS*}, that is equivalent to an MCWC$(2,n_1;2,n_2;6)$ of size $\lfloor \frac{n_2(n_1-1)}{4}\rfloor$, as desired. For $n_1\in\{11,15,19\}$ and $n_2= 2n_1-1$, we proceed similarly; take an HSAS$(n_2,n_1;5,3)$ from the Appendix and fill in the hole with an SAS$^*(5,3)$ (which is equivalent to an MCWC$(2,5;2,3;6)$ and has been constructed above).
\end{proof}

\begin{table*}
\centering
\renewcommand{\arraystretch}{1}
\setlength\arraycolsep{3pt} \caption{Codewords of Small
MCWC$(2,n_1;2,n_2;6)$ for $3 \leq n_1 \leq 9$}
\label{tab.smallMCWC}
\begin{tabular}{c|l}
\hline $(n_1,n_2)$ & {\hfill Codewords \hfill} \\
\hline $(3,3)$ & $\begin{array}{lllllll}
\langle 0, 1, 3, 4 \rangle\\
\end{array}$ \\
\hline $(3,5)$ & $\begin{array}{lllllll}
\langle 0, 1, 3, 4 \rangle & \langle 1, 2, 5, 6 \rangle \\
\end{array}$ \\
\hline $(5,5)$ & $\begin{array}{lllllll}
\langle 0, 1, 5, 6 \rangle & \langle 0, 2, 7, 8 \rangle & \langle 1, 3, 7, 9 \rangle & \langle 2, 4, 5, 9 \rangle & \langle 3, 4, 6, 8 \rangle \\
\end{array}$ \\
\hline $(5,7)$ & $\begin{array}{lllllll}
\langle 0, 1, 5, 6 \rangle & \langle 0, 2, 7, 8 \rangle & \langle 0, 3, 9, 10 \rangle & \langle 1, 2, 9, 11 \rangle & \langle 1, 4, 7, 10 \rangle & \langle 3, 4, 5, 8 \rangle\\
\end{array}$ \\
\hline $(5,9)$ & $\begin{array}{lllllllll}
\langle 0, 3, 10, 9 \rangle & \langle 2, 3, 5, 13 \rangle & \langle 0, 2, 8, 7 \rangle & \langle 0, 4, 11, 12 \rangle & \langle 1, 2, 9, 11 \rangle & \langle 1, 3, 7, 12 \rangle & \langle 0, 1, 6, 5 \rangle & \langle 1, 4, 8, 13 \rangle & \langle 2, 4, 6, 10 \rangle \\
\end{array}$ \\
\hline $(7,7)$ & $\begin{array}{llllllllll}
\langle 0, 1, 7, 8 \rangle &
\langle 0, 2, 9, 10 \rangle &
\langle 0, 3, 11, 12 \rangle &
\langle 1, 2, 11, 13 \rangle &
\langle 1, 4, 9, 12 \rangle &
\langle 2, 5, 7, 12 \rangle &
\langle 3, 4, 7, 10 \rangle &
\langle 3, 5, 8, 9 \rangle &
\langle 4, 6, 8, 11 \rangle &
\langle 5, 6, 10, 13 \rangle \\
\end{array}$ \\
\hline $(7,9)$ & $\begin{array}{llllllllll}
\langle 0, 1, 7, 8 \rangle &
\langle 0, 2, 9, 10 \rangle &
\langle 0, 3, 11, 12 \rangle &
\langle 0, 4, 13, 14 \rangle &
\langle 1, 2, 11, 13 \rangle &
\langle 1, 3, 9, 14 \rangle &
\langle 1, 4, 10, 12 \rangle &
\langle 2, 3, 7, 15 \rangle &
\langle 2, 5, 8, 12 \rangle \\
\langle 3, 5, 10, 13 \rangle &
\langle 4, 5, 7, 9 \rangle &
\langle 4, 6, 8, 11 \rangle &
\langle 5, 6, 14, 15 \rangle &
\end{array}$ \\
\hline $(7,11)$ & $\begin{array}{llllllllll}
\langle 0, 1, 7, 8 \rangle &
\langle 0, 2, 9, 10 \rangle &
\langle 1, 5, 11, 13 \rangle &
\langle 0, 4, 13, 14 \rangle &
\langle 0, 5, 15, 16 \rangle &
\langle 3, 6, 16, 13 \rangle &
\langle 1, 3, 9, 14 \rangle &
\langle 0, 3, 11, 17 \rangle &
\langle 1, 6, 15, 17 \rangle \\
\langle 2, 3, 7, 15 \rangle &
\langle 2, 4, 8, 16 \rangle &
\langle 2, 5, 12, 14 \rangle &
\langle 3, 5, 8, 10 \rangle &
\langle 1, 4, 12, 10 \rangle &
\langle 4, 5, 7, 17 \rangle &
\langle 4, 6, 9, 11 \rangle &
\end{array}$ \\
\hline $(7,13)$ & $\begin{array}{llllllllll}
\langle 0, 1, 7, 8 \rangle &
\langle 0, 2, 9, 10 \rangle &
\langle 0, 3, 11, 12 \rangle &
\langle 1, 4, 12, 10 \rangle &
\langle 5, 4, 7, 9 \rangle &
\langle 0, 6, 17, 18 \rangle &
\langle 1, 2, 11, 13 \rangle &
\langle 1, 3, 9, 14 \rangle &
\langle 3, 5, 10, 8 \rangle \\
\langle 1, 5, 17, 19 \rangle &
\langle 3, 4, 19, 18 \rangle &
\langle 2, 4, 8, 16 \rangle &
\langle 5, 0, 16, 15 \rangle &
\langle 0, 4, 14, 13 \rangle &
\langle 2, 3, 7, 17 \rangle &
\langle 3, 6, 13, 16 \rangle &
\langle 4, 6, 11, 15 \rangle &
\langle 2, 5, 12, 18 \rangle \\
\langle 2, 6, 14, 19 \rangle &
\end{array}$ \\
\hline $(9,9)$ & $\begin{array}{llllllllll}
\langle 7, 3, 15, 12 \rangle &
\langle 2, 1, 16, 11 \rangle &
\langle 4, 8, 9, 15 \rangle &
\langle 0, 3, 11, 10 \rangle &
\langle 2, 8, 10, 13 \rangle &
\langle 2, 6, 9, 12 \rangle &
\langle 4, 5, 11, 12 \rangle &
\langle 1, 0, 12, 17 \rangle &
\langle 6, 7, 13, 17 \rangle \\
\langle 6, 0, 14, 15 \rangle &
\langle 6, 4, 10, 16 \rangle &
\langle 1, 4, 14, 13 \rangle &
\langle 7, 1, 10, 9 \rangle &
\langle 7, 8, 14, 11 \rangle &
\langle 3, 8, 17, 16 \rangle &
\langle 5, 2, 15, 17 \rangle &
\langle 3, 5, 14, 9 \rangle &
\langle 0, 5, 13, 16 \rangle \\
\end{array}$ \\
\hline $(9,11)$ & $\begin{array}{llllllllll}
\langle 4, 8, 13, 11 \rangle &
\langle 3, 0, 14, 10 \rangle &
\langle 6, 5, 11, 19 \rangle &
\langle 3, 1, 16, 11 \rangle &
\langle 0, 8, 16, 15 \rangle &
\langle 8, 2, 9, 17 \rangle &
\langle 6, 2, 14, 13 \rangle &
\langle 6, 8, 12, 10 \rangle &
\langle 4, 3, 17, 15 \rangle \\
\langle 6, 3, 9, 18 \rangle &
\langle 4, 1, 12, 18 \rangle &
\langle 3, 5, 13, 12 \rangle &
\langle 8, 1, 19, 14 \rangle &
\langle 7, 0, 11, 12 \rangle &
\langle 1, 7, 9, 13 \rangle &
\langle 0, 5, 17, 18 \rangle &
\langle 6, 7, 17, 16 \rangle &
\langle 2, 7, 19, 18 \rangle \\
\langle 7, 5, 14, 15 \rangle &
\langle 0, 4, 9, 19 \rangle &
\langle 1, 2, 10, 15 \rangle &
\langle 4, 5, 10, 16 \rangle &
\end{array}$ \\
\hline $(9,13)$ & $\begin{array}{llllllllll}
\langle 6, 4, 13, 17 \rangle &
\langle 3, 2, 17, 9 \rangle &
\langle 0, 1, 9, 10 \rangle &
\langle 6, 3, 18, 15 \rangle &
\langle 5, 0, 16, 17 \rangle &
\langle 2, 5, 18, 13 \rangle &
\langle 7, 1, 18, 20 \rangle &
\langle 2, 1, 12, 15 \rangle &
\langle 2, 4, 16, 14 \rangle \\
\langle 1, 5, 19, 21 \rangle &
\langle 6, 7, 9, 12 \rangle &
\langle 8, 7, 13, 16 \rangle &
\langle 8, 2, 20, 19 \rangle &
\langle 0, 3, 19, 13 \rangle &
\langle 4, 5, 20, 9 \rangle &
\langle 8, 0, 18, 12 \rangle &
\langle 7, 3, 14, 21 \rangle &
\langle 7, 5, 15, 10 \rangle \\
\langle 7, 4, 19, 11 \rangle &
\langle 1, 3, 16, 11 \rangle &
\langle 6, 8, 10, 11 \rangle &
\langle 6, 0, 20, 14 \rangle &
\langle 3, 4, 12, 10 \rangle &
\langle 1, 8, 14, 17 \rangle &
\langle 0, 2, 11, 21 \rangle &
\langle 4, 8, 15, 21 \rangle &
\end{array}$ \\
\hline $(9,15)$ & $\begin{array}{llllllllll}
\langle 0, 1, 9, 10 \rangle &
\langle 0, 2, 11, 12 \rangle &
\langle 1, 2, 15, 13 \rangle &
\langle 6, 2, 22, 23 \rangle &
\langle 0, 5, 17, 18 \rangle &
\langle 0, 6, 19, 20 \rangle &
\langle 0, 4, 16, 15 \rangle &
\langle 2, 3, 9, 21 \rangle &
\langle 8, 1, 17, 23 \rangle \\
\langle 4, 7, 17, 13 \rangle &
\langle 3, 4, 23, 19 \rangle &
\langle 5, 8, 16, 22 \rangle &
\langle 2, 5, 14, 20 \rangle &
\langle 2, 7, 16, 19 \rangle &
\langle 2, 4, 10, 18 \rangle &
\langle 3, 6, 17, 15 \rangle &
\langle 4, 8, 20, 21 \rangle &
\langle 1, 3, 11, 22 \rangle \\
\langle 0, 3, 14, 13 \rangle &
\langle 3, 5, 10, 12 \rangle &
\langle 1, 5, 19, 21 \rangle &
\langle 1, 4, 14, 12 \rangle &
\langle 4, 5, 9, 11 \rangle &
\langle 7, 8, 14, 11 \rangle &
\langle 7, 0, 21, 22 \rangle &
\langle 6, 7, 9, 12 \rangle &
\langle 6, 8, 10, 13 \rangle \\
\langle 5, 7, 23, 15 \rangle &
\langle 3, 7, 18, 20 \rangle &
\langle 1, 6, 16, 18 \rangle &
\end{array}$ \\
\hline $(9,17)$ & $\begin{array}{llllllllll}
\langle 8, 7, 15, 17 \rangle &
\langle 2, 6, 18, 9 \rangle &
\langle 0, 5, 14, 17 \rangle &
\langle 4, 0, 16, 19 \rangle &
\langle 1, 8, 11, 25 \rangle &
\langle 1, 0, 18, 24 \rangle &
\langle 1, 5, 16, 9 \rangle &
\langle 2, 3, 16, 24 \rangle &
\langle 7, 0, 21, 10 \rangle \\
\langle 4, 8, 14, 9 \rangle &
\langle 5, 6, 19, 10 \rangle &
\langle 8, 5, 12, 18 \rangle &
\langle 3, 4, 15, 13 \rangle &
\langle 3, 7, 19, 25 \rangle &
\langle 0, 6, 25, 15 \rangle &
\langle 3, 1, 17, 20 \rangle &
\langle 5, 7, 24, 13 \rangle &
\langle 8, 0, 13, 20 \rangle \\
\langle 3, 6, 11, 14 \rangle &
\langle 2, 0, 11, 12 \rangle &
\langle 4, 6, 22, 17 \rangle &
\langle 8, 2, 10, 22 \rangle &
\langle 8, 6, 24, 23 \rangle &
\langle 1, 2, 15, 19 \rangle &
\langle 4, 1, 12, 10 \rangle &
\langle 7, 6, 12, 16 \rangle &
\langle 3, 0, 23, 9 \rangle \\
\langle 7, 4, 18, 11 \rangle &
\langle 3, 5, 22, 21 \rangle &
\langle 4, 5, 20, 25 \rangle &
\langle 6, 1, 21, 13 \rangle &
\langle 2, 4, 23, 21 \rangle &
\langle 7, 1, 22, 23 \rangle &
\langle 2, 7, 14, 20 \rangle &
\end{array}$ \\
\hline
\end{tabular}
\end{table*}

\begin{lem}\label{PBD2MCWC}
Let $t$ be a positive integer with  $2t+1\geq 21$ and $2t+1\not \in\{23,27,29,33,39,43,51,59,75,83,87,95,99,107,$ $139,179\}$. Let $n_1=4t+1$ or $4t+3$, $n_1\leq n_2\leq 2n_1-1$ and $n_2$ be odd. Then T$(2,n_1;2,n_2;6)=\lfloor \frac{n_2(n_1-1)}{4}\rfloor$.
\end{lem}

\begin{proof}  According to Propositions~\ref{MCWC2SAS} and~\ref{MCWC2SAS*}, we only need to construct the corresponding SAS$(n_2,n_1)$ when $n_1\equiv 1\pmod{4}$ or SAS$^*(n_2,n_1)$ when $n_1\equiv 3\pmod{4}$.

For each given $t$ and $2t+1\not \in\{71,111,113,115,119\}$, take a $(2t+1,\{5,7,9\},1)$-PBD from Theorem~\ref{PBD579}, and remove one point to obtain a $\{5,7,9\}$-GDD of type $4^i 6^j  8^k$ with $4i+6j+8k=2t$. Assign each point with weights $(4,2)$ or $(2,2)$ and apply Construction~\ref{WFC}; the input SFSs of type $(4,2)^a(2,2)^b$ with $a+b\in\{5,7,9\}$ are constructed in the Appendix. Then we can get an SFS of type
$$(8,8)^{i_8}(10,8)^{i_{10}}\cdots(16,8)^{i_{16}}(12,12)^{j_{12}}\cdots(24,12)^{j_{24}}(16,16)^{k_{16}}\cdots(32,16)^{k_{32}},$$
for any nonnegative integers $i_8, i_{10},\ldots,i_{16},j_{12},\ldots,k_{32}$ with
\begin{align*}
i_8+ i_{10}+\cdots+i_{16}&=i\\
j_{12}+ j_{14}+\cdots+j_{24}&=j\\
k_{16}+ k_{18}+\cdots+k_{32}&=k.
\end{align*}
Now, we can fill the holes of the SFS in three ways:
\begin{enumerate}
  \item Add a new row and a new column and apply Construction~\ref{BFC}~(2) with `$e=1$' and `$w=1$'; the input HSAS$(r,v;1,1)$ (i.e. SAS$(r,v)$) with $v\in\{9,13,17\}$ and $v\leq r\leq 2v-1$ come from Lemma~\ref{MCWCsmall}. Then we get an  SAS$(s,4t+1;1,1)$ with $4t+1\leq s \leq 8t+1$, as desired;
  \item Add three new rows and three new columns and apply Construction~\ref{BFC}~(1) with `$e=3$' and `$w=3$'; the input HSAS$(r,v;3,3)$ with $v\in\{11,15,19\}$ and $v\leq r\leq 2v-3$ are constructed in the Appendix. We get an HSAS$(s,4t+3;3,3)$  with $4t+3\leq s \leq 8t+3$. Then fill in the hole with an SAS$(3,3)$ constructed in Lemma~\ref{MCWCsmall} to obtain the desired SAS$^*(s,4t+3)$  with $4t+3\leq s \leq 8t+3$.
  \item When the SFS has type $(16,8)^i(24,12)^j(32,16)^k$, add five new rows and five new columns and apply Construction~\ref{BFC}~(1)  with `$e=5$' and `$w=3$'; the input HSAS$(r,v;5,3)$ with $(r,v)\in\{(21,11),$ $(29,15),(37,19)\}$ are constructed in the Appendix. We get an HSAS$(8t+5,4t+3;5,3)$. Then fill in the hole with an SAS$^*(5,3)$ constructed in Lemma~\ref{MCWCsmall} to obtain the desired SAS$^*(8t+5,4t+3)$.
\end{enumerate}

For $2t+1=71$, take a TD$(9,8)$ from Theorem~\ref{TD} and truncate one of its group  to six points to obtain an $\{8,9\}$-GDD of type $8^8 6^1$, noting that $8\times 8+6=70=2t$. Then proceed similarly as above, we can obtain the desired SAS$(s,4t+1)$ and SAS$^*(s,4t+3)$. Here the additional input SFSs of type $(4,2)^a(2,2)^{8-a}$ with $0\leq a \leq 8$ are constructed in the Appendix.

For $2t+1\in\{111,113,115,119\}$, take a $\{7,9\}$-GDD of type $8^{15}$ from \cite[Part~4, Corollary~2.44]{CD} and truncate the last two groups to obtain $\{5,6,7,8,9\}$-GDDs of types $8^{13} 6^1$, $8^{14}$, $8^{13} 6^1 4^1$ and $8^{14} 6^1$, respectively. Then proceed similarly as above, we can obtain the desired SASs and SAS$^*$s. Here the additional input SFSs of type $(4,2)^a(2,2)^{6-a}$ with $0\leq a \leq 6$ are constructed in the Appendix.
\end{proof}

\begin{rem} In the proof of Lemma~\ref{PBD2MCWC}, we have constructed HSAS$(s,4t+3;3,3)$  with $4t+3\leq s \leq 8t+3$ and HSAS$(8t+5,4t+3;5,3)$. These HSASs will be used in later constructions.
\end{rem}

\begin{lem}
Let $t$ be a positive integer with $2t+1 \in\{39,43,51,59,75,99\}$. Let $n_1=4t+1$ or $4t+3$, $n_1\leq n_2\leq 2n_1-1$ and $n_2$ be odd. Then T$(2,n_1;2,n_2;6)=\lfloor \frac{n_2(n_1-1)}{4}\rfloor$.
\end{lem}

\begin{proof}
For $2t+1=39$, take a $\{5,7\}$-GDD of type $6^6 2^1$ from \cite[Part~4, Example~2.51]{CD}, noting that $6\times 6+2=2t$.  Assign each point with weights $(4,2)$ or $(2,2)$ and apply Construction~\ref{WFC}. Then we can get an SFS of type $$(12,12)^{i_{12}}(14,12)^{i_{14}}\cdots(24,12)^{i_{24}}(4,4)^{j_4}(8,4)^{j_8},$$
for any nonnegative integers $i_{12}, i_{14},\ldots,i_{24},j_4,j_8$ with $i_{12}+ i_{14}+\cdots+i_{24}=6$ and $j_4+j_8=1$.
Now, we can fill the holes of the SFS in three ways:
\begin{enumerate}
  \item Add a new row and a new column and apply Construction~\ref{BFC}~(2) with `$e=1$' and `$w=1$'; the input HSAS$(r,13;1,1)$ (i.e. SAS$(r,13)$) with $13\leq r\leq 25$,  SAS$(5,5)$ and SAS$(9,5)$ come from Lemma~\ref{MCWCsmall}. Then we get an SAS$(s,77)$ with $77\leq s \leq 153$, as desired.
  \item Add three new rows and three new columns and apply Construction~\ref{BFC}~(3) with `$e=3$' and `$w=3$'; the input HSAS$(r,15;3,3)$ with $15\leq r\leq 27$ are constructed in the Appendix and the input SAS$^*(7,7)$ and SAS$^*(11,7)$ come from Lemma~\ref{MCWCsmall}. Then we get an SAS$^*(s,79)$  with $79 \leq s \leq 155$.
  \item When the SFS has type $(24,12)^6 (8,4)^1$, add five new rows and five new columns and apply Construction~\ref{BFC}~(3) with `$e=5$' and `$w=3$'; the input HSAS$(29,15;5,3)$ is constructed in the Appendix and the input SAS$^*(13,7)$ comes from Lemma~\ref{MCWCsmall}. Then we get an SAS$^*(157,79)$, as desired.
\end{enumerate}

For $2t+1\in\{43,51,59,75,99\}$, we start with $\{5,6,7,8,9\}$-GDDs of types $8^5 2^1$, $8^6 2^1$, $8^7 2^1$, $8^9 2^1$, and $8^{12} 2^1$, respectively, which will be constructed below. Proceed as above to obtain the desired SASs and SAS$^*$s; here we fill in the holes of the SFS with SAS$(r,17)$ (see Lemma~\ref{MCWCsmall}), HSAS$(r,19;3,3)$ (see the Appendix) and HSAS$(37,19;5,3)$ (see the Appendix). The $\{5,6,7,8,9\}$-GDDs are constructed as follows. For the types  $8^5 2^1$, $8^6 2^1$ and $8^7 2^1$, take a TD$(9,8)$ from Theorem~\ref{TD} and truncate the last four groups. For the type $8^9 2^1$, take a TD$(9,9)$ from Theorem~\ref{TD} and remove one point to redefine the groups to obtain a $\{9\}$-GDD of type $8^{10}$. Then truncate the last group. For the type $8^{12} 2^1$, take a $\{9\}$-GDD of type $8^{15} 16^1$ from \cite[Part~4, Corollary~2.44]{CD} and truncate the last four groups.
\end{proof}

\begin{lem}
Let $t$ be a positive integer. If $2t+1 \in\{107,139,179\}$. Let $n_1=4t+1$ or $4t+3$, $n_1\leq n_2\leq 2n_1-1$ and $n_2$ be odd. Then T$(2,n_1;2,n_2;6)=\lfloor \frac{n_2(n_1-1)}{4}\rfloor$.
\end{lem}

\begin{proof}
For $2t+1=107$, take a TD$(6,20)$ from Theorem~\ref{TD} and truncate the last group to six points to obtain a $\{5,6\}$-GDD of type $20^5 6^1$. Assign each point with weights $(4,2)$ or $(2,2)$ and apply Construction~\ref{WFC}. Then we can get an SFS of type $$(40,40)^{i_{40}}(42,40)^{i_{42}}\cdots(80,40)^{i_{80}}(12,12)^{j_{12}} (14,12)^{j_{14}}\cdots(24,12)^{j_{24}},$$
for any nonnegative integers $i_{40}, i_{42},\ldots,i_{80},j_{12},\ldots,j_{24}$ with $i_{40}+ i_{42}+\cdots+i_{80}=5$ and $j_{12}+j_{14}+\ldots+j_{24}=1$.
Now, we can fill the holes of the SFS in three ways:
\begin{enumerate}
  \item Add a new row and a new column and apply Construction~\ref{BFC}~(2) with `$e=1$' and `$w=1$'; the input HSASs and SASs come from Lemmas~\ref{MCWCsmall}--\ref{PBD2MCWC}. Then we get an SAS$(s,213)$ with $213\leq s \leq 425$, as desired.
  \item Add three new rows and three new columns and apply Construction~\ref{BFC}~(3) with `$e=3$' and `$w=3$'; the input HSAS$(r,43;3,3)$ with $43\leq r\leq 83$ are constructed in the proof of Lemma~\ref{PBD2MCWC} and the input SAS$^*(r,15)$ comes from Lemma~\ref{MCWCsmall}. Then we get an SAS$^*(s,215)$  with $215 \leq s \leq 427$.
  \item When the SFS has type $(80,40)^5 (24,12)^1$, add five new rows and five new columns and apply Construction~\ref{BFC}~(3) with `$e=5$' and `$w=3$'; the input HSAS$(85,43;5,3)$ and  SAS$^*(29,15)$ come from Lemmas~\ref{MCWCsmall}--\ref{PBD2MCWC}. Then we get an SAS$^*(429,215)$, as desired.
\end{enumerate}
For $2t+1=139$ or $179$, take a TD$(8,24)$ from Theorem~\ref{TD} and truncate the last three groups to obtain $\{5,6,7,8,9\}$-GDDs of types $24^5 6^3$ or $24^7 6^1 4^1$. Then proceed similarly as above to obtain the desired  SASs and SAS$^*$s; the input HSASs, SASs and SAS$^*$s all come from Lemma~\ref{MCWCsmall}.
\end{proof}

\begin{lem}
Let $t$ be a positive integer with $2t+1 \in\{83,87,95\}$. Let $n_1=4t+1$, $n_1\leq n_2\leq 2n_1-1$ and $n_2$ be odd. Then T$(2,n_1;2,n_2;6)=\frac{n_2(n_1-1)}{4}$.
\end{lem}

\begin{proof}
Take a TD$(6,16)$ from Theorem~\ref{TD} and truncate the last group to obtain $\{5,6\}$-GDDs of types $16^5 2^1$, $16^5 6^1$ or $16^5 14^1$, respectively. Then proceed similarly as above to obtain the desired SASs; the input SAS$(s,v)$ with $s\in\{5,13,29,33\}$ all come from Lemma~\ref{MCWCsmall}.
\end{proof}

Combining the above lemmas, we get the following result.

\begin{thm}Let $n_1,n_2$ be two odd integers with $0< n_1 \leq n_2 \leq 2n_1-1$. Then
T$(2,n_1;2,n_2;6) = \lfloor \frac{n_2(n_1-1)}{4}\rfloor$, except for $(n_1,n_2)=(5,7)$, and except possibly for $n_1\in\{23,27,31,35,39,45,47,53,55,57,59,65,67,165,175,191\}$.
\end{thm}

\section{Conclusions}
In this paper, we consider the  bounds and constructions of MCWCs. For the upper bound, we use three different approaches to improve the generalised Johnson bounds mentioned in \cite{zhang}. For the lower bound, we derive two asymptotic lower bounds, the first is from the technique of concatenation and the second is from the Gilbert-Varshamov type bound. A comparison between these two bounds is also given. For the constructions, by establishing the connections between some combinatorial structures and MCWCs, several new combinatorial constructions for MCWCs are given. We  obtain the asymptotic existence result of two classes of optimal MCWCs and construct a class of optimal MWCWs which are open in \cite{Cheeoptimal}. As consequences,  the Johnson-type bounds are shown to be asymptotically exact for MCWCs with distances $2\sum_{i=1}^mw_i-2$ or $2mw-w$. The maximum sizes of MCWCs with total weight less than or equal to four are  determined almost completely.

\bibliographystyle{plain}
\bibliography{REF}

\section{Appendix}

\subsection{Small MCWC$(2,n_1;2,n_2;6)$ for $n_1 \equiv 1\pmod{4}$ and $13\leq n_1 \leq 37$}

\begin{lem}
T$(2,13;2,n;6) =3n$ for each odd $n$ and $13 \leq n \leq 25$.
\end{lem}

\begin{proof} Let $X_1=(\bbZ_3\times \{0,1,2,3\})\cup \{\infty\}$. For $13 \leq n \leq 17$, let $X_2= (\bbZ_3\times \{4,5,6,7\})\cup (\{a\}\times \{1,\ldots,n-12\})$; for $19 \leq n \leq 23$, let $X_2= (\bbZ_3\times \{4,5,\ldots,9\})\cup (\{a\}\times \{1,\ldots,n-18\})$; for $n = 25$, let $X_2= (\bbZ_3\times \{4,5,\ldots,11\})\cup (\{a\}\times \{1\})$.
Denote $X=X_1\cup X_2$. The desired codes of size $3n$ are constructed on $\bbZ_2^X$. The codewords are obtained by developing the following base codewords under the action of the cyclic group $\bbZ_3$, where the points $\infty$ and $a_i$ are fixed.

\noindent $n=13$:
{\scriptsize
$$\begin{array}{llllllllll}
\langle\infty, 0_{0}, 0_{5}, 0_{4}\rangle &
\langle\infty, 0_{3}, 2_{6}, 2_{7}\rangle &
\langle1_{1}, 1_{0}, 2_{4}, a_{1}\rangle &
\langle2_{2}, 2_{3}, 2_{5}, a_{1}\rangle &
\langle1_{0}, 2_{2}, 0_{7}, 2_{6}\rangle &
\langle1_{1}, 2_{2}, 1_{5}, 1_{6}\rangle &
\langle0_{0}, 2_{1}, 0_{7}, 2_{4}\rangle &
\langle0_{0}, 2_{3}, 0_{6}, 1_{5}\rangle \\
\langle2_{1}, 2_{2}, 1_{4}, 0_{5}\rangle &
\langle1_{2}, 2_{3}, 2_{4}, 2_{6}\rangle &
\langle2_{0}, 0_{1}, 0_{7}, 1_{6}\rangle &
\langle1_{2}, 0_{3}, 0_{7}, 1_{4}\rangle &
\langle0_{1}, 1_{3}, 2_{5}, 2_{7}\rangle &
\end{array}$$
}

\noindent $n=15$:
{\scriptsize
$$\begin{array}{llllllllll}
\langle\infty, 2_{1}, 2_{7}, 2_{6}\rangle &
\langle\infty, 0_{3}, 0_{5}, 2_{4}\rangle &
\langle2_{3}, 2_{2}, 0_{5}, a_{1}\rangle &
\langle2_{1}, 2_{0}, 1_{7}, a_{1}\rangle &
\langle2_{2}, 0_{3}, 0_{7}, a_{2}\rangle &
\langle1_{0}, 2_{1}, 0_{4}, a_{2}\rangle &
\langle0_{0}, 2_{1}, 0_{5}, a_{3}\rangle &
\langle0_{3}, 1_{2}, 1_{4}, a_{3}\rangle \\
\langle0_{1}, 2_{3}, 2_{6}, 2_{4}\rangle &
\langle2_{1}, 0_{2}, 2_{4}, 0_{7}\rangle &
\langle1_{0}, 1_{3}, 0_{6}, 2_{7}\rangle &
\langle0_{1}, 0_{3}, 1_{6}, 2_{5}\rangle &
\langle1_{2}, 2_{0}, 2_{4}, 0_{6}\rangle &
\langle0_{2}, 0_{0}, 2_{5}, 0_{6}\rangle &
\langle0_{0}, 1_{2}, 1_{5}, 0_{7}\rangle &
\end{array}$$
}

\noindent $n=17$:
{\scriptsize
$$\begin{array}{llllllllll}
\langle\infty, 1_{1}, 1_{7}, 1_{6}\rangle &
\langle\infty, 0_{3}, 1_{4}, 2_{5}\rangle &
\langle1_{2}, 1_{3}, 0_{6}, a_{1}\rangle &
\langle0_{0}, 0_{1}, 1_{4}, a_{1}\rangle &
\langle1_{0}, 0_{1}, 1_{6}, a_{2}\rangle &
\langle1_{3}, 0_{2}, 2_{5}, a_{2}\rangle &
\langle0_{1}, 2_{0}, 2_{5}, a_{3}\rangle &
\langle2_{3}, 0_{2}, 1_{7}, a_{3}\rangle \\
\langle1_{1}, 0_{3}, 0_{4}, a_{4}\rangle &
\langle0_{0}, 1_{2}, 2_{6}, a_{4}\rangle &
\langle2_{1}, 0_{3}, 2_{4}, a_{5}\rangle &
\langle0_{2}, 1_{0}, 2_{7}, a_{5}\rangle &
\langle1_{3}, 0_{0}, 1_{5}, 2_{7}\rangle &
\langle1_{2}, 0_{1}, 1_{5}, 1_{7}\rangle &
\langle0_{0}, 0_{3}, 0_{7}, 1_{6}\rangle &
\langle1_{1}, 0_{2}, 1_{5}, 0_{6}\rangle \\
\langle0_{0}, 0_{2}, 0_{4}, 2_{4}\rangle &
\end{array}$$
}

\noindent $n=19$:
{\scriptsize
$$\begin{array}{llllllllll}
\langle\infty, 0_{0}, 0_{4}, 0_{5}\rangle &
\langle\infty, 0_{1}, 0_{6}, 0_{7}\rangle &
\langle\infty, 0_{2}, 0_{8}, 0_{9}\rangle &
\langle0_{0}, 0_{1}, 1_{4}, a_{1}\rangle &
\langle0_{2}, 0_{3}, 0_{5}, a_{1}\rangle &
\langle0_{0}, 1_{0}, 2_{5}, 0_{6}\rangle &
\langle0_{3}, 1_{3}, 0_{4}, 0_{8}\rangle &
\langle0_{0}, 1_{3}, 2_{7}, 2_{9}\rangle \\
\langle0_{0}, 0_{2}, 1_{7}, 1_{8}\rangle &
\langle0_{1}, 1_{1}, 0_{5}, 0_{8}\rangle &
\langle0_{1}, 0_{2}, 2_{4}, 1_{5}\rangle &
\langle0_{1}, 1_{2}, 1_{6}, 0_{9}\rangle &
\langle0_{2}, 1_{2}, 2_{6}, 1_{4}\rangle &
\langle0_{2}, 1_{3}, 2_{5}, 0_{7}\rangle &
\langle0_{2}, 2_{3}, 2_{7}, 1_{9}\rangle &
\langle0_{0}, 1_{1}, 0_{7}, 2_{8}\rangle \\
\langle0_{0}, 2_{1}, 2_{4}, 0_{9}\rangle &
\langle0_{0}, 2_{3}, 1_{6}, 0_{8}\rangle &
\langle0_{1}, 2_{3}, 2_{6}, 2_{9}\rangle &
\end{array}$$
}

\noindent $n=21$:
{\scriptsize
$$\begin{array}{llllllllll}
\langle\infty, 1_{1}, 1_{7}, 1_{6}\rangle &
\langle\infty, 0_{2}, 1_{8}, 0_{9}\rangle &
\langle\infty, 1_{0}, 1_{4}, 1_{5}\rangle &
\langle2_{1}, 2_{0}, 0_{7}, a_{1}\rangle &
\langle0_{2}, 0_{3}, 0_{5}, a_{1}\rangle &
\langle1_{3}, 0_{1}, 1_{4}, a_{2}\rangle &
\langle0_{0}, 0_{2}, 1_{5}, a_{2}\rangle &
\langle2_{0}, 1_{1}, 1_{4}, a_{3}\rangle \\
\langle2_{3}, 1_{2}, 0_{7}, a_{3}\rangle &
\langle0_{2}, 1_{2}, 2_{6}, 0_{4}\rangle &
\langle0_{0}, 1_{2}, 2_{9}, 2_{7}\rangle &
\langle1_{0}, 1_{3}, 1_{7}, 2_{9}\rangle &
\langle0_{1}, 0_{3}, 2_{7}, 1_{8}\rangle &
\langle0_{0}, 1_{0}, 2_{6}, 0_{8}\rangle &
\langle0_{3}, 2_{3}, 1_{4}, 2_{8}\rangle &
\langle1_{1}, 2_{2}, 1_{5}, 0_{4}\rangle \\
\langle2_{1}, 1_{3}, 1_{6}, 0_{6}\rangle &
\langle0_{0}, 1_{3}, 0_{9}, 2_{5}\rangle &
\langle1_{1}, 1_{2}, 1_{8}, 0_{8}\rangle &
\langle1_{1}, 2_{1}, 0_{9}, 0_{5}\rangle &
\langle0_{2}, 2_{3}, 0_{6}, 2_{9}\rangle &
\end{array}$$
}

\noindent $n=23$:
{\scriptsize
$$\begin{array}{llllllllll}
\langle\infty, 0_{0}, 0_{5}, 0_{4}\rangle &
\langle\infty, 1_{2}, 0_{8}, 0_{9}\rangle &
\langle\infty, 1_{1}, 2_{6}, 0_{7}\rangle &
\langle0_{3}, 0_{2}, 0_{5}, a_{1}\rangle &
\langle1_{0}, 1_{1}, 2_{4}, a_{1}\rangle &
\langle2_{1}, 1_{0}, 2_{5}, a_{2}\rangle &
\langle0_{3}, 2_{2}, 2_{4}, a_{2}\rangle &
\langle2_{0}, 1_{1}, 2_{7}, a_{3}\rangle \\
\langle2_{3}, 0_{2}, 0_{6}, a_{3}\rangle &
\langle0_{0}, 0_{2}, 2_{5}, a_{4}\rangle &
\langle0_{3}, 0_{1}, 0_{6}, a_{4}\rangle &
\langle1_{1}, 2_{3}, 0_{4}, a_{5}\rangle &
\langle0_{2}, 2_{0}, 2_{6}, a_{5}\rangle &
\langle0_{0}, 2_{2}, 0_{8}, 1_{7}\rangle &
\langle0_{0}, 2_{0}, 1_{6}, 1_{8}\rangle &
\langle0_{0}, 2_{3}, 0_{9}, 2_{4}\rangle \\
\langle0_{3}, 1_{3}, 2_{8}, 2_{5}\rangle &
\langle0_{1}, 2_{3}, 2_{9}, 0_{7}\rangle &
\langle0_{1}, 1_{1}, 2_{5}, 1_{8}\rangle &
\langle0_{1}, 1_{2}, 2_{6}, 1_{9}\rangle &
\langle0_{0}, 0_{3}, 2_{7}, 2_{9}\rangle &
\langle1_{2}, 2_{1}, 1_{8}, 2_{9}\rangle &
\langle0_{2}, 1_{2}, 2_{4}, 1_{7}\rangle &
\end{array}$$
}

\noindent $n=25$:
{\scriptsize
$$\begin{array}{llllllllll}
\langle\infty, 0_{3}, 0_{11}, 0_{7}\rangle &
\langle\infty, 2_{0}, 2_{5}, 1_{9}\rangle &
\langle\infty, 0_{2}, 0_{6}, 1_{8}\rangle &
\langle\infty, 0_{1}, 1_{10}, 1_{4}\rangle &
\langle0_{0}, 0_{2}, 0_{7}, a_{1}\rangle &
\langle0_{1}, 1_{3}, 0_{6}, a_{1}\rangle &
\langle2_{2}, 0_{1}, 0_{5}, 1_{7}\rangle &
\langle1_{3}, 1_{0}, 0_{8}, 2_{5}\rangle \\
\langle2_{0}, 0_{3}, 0_{8}, 1_{7}\rangle &
\langle0_{3}, 1_{2}, 1_{10}, 0_{10}\rangle &
\langle0_{0}, 2_{1}, 1_{6}, 0_{9}\rangle &
\langle0_{0}, 1_{2}, 2_{6}, 0_{6}\rangle &
\langle2_{0}, 1_{0}, 0_{4}, 2_{11}\rangle &
\langle0_{0}, 0_{1}, 0_{4}, 2_{10}\rangle &
\langle0_{2}, 1_{2}, 1_{9}, 2_{11}\rangle &
\langle2_{1}, 0_{2}, 0_{8}, 0_{11}\rangle \\
\langle1_{2}, 2_{0}, 0_{9}, 2_{10}\rangle &
\langle1_{0}, 0_{3}, 0_{5}, 2_{7}\rangle &
\langle0_{1}, 2_{1}, 2_{7}, 2_{9}\rangle &
\langle2_{3}, 0_{3}, 1_{4}, 1_{9}\rangle &
\langle2_{2}, 2_{1}, 1_{4}, 1_{5}\rangle &
\langle0_{0}, 1_{1}, 1_{10}, 0_{8}\rangle &
\langle1_{3}, 1_{1}, 2_{6}, 0_{11}\rangle &
\langle2_{1}, 1_{3}, 0_{5}, 2_{11}\rangle \\
\langle0_{2}, 1_{3}, 1_{4}, 2_{8}\rangle &
\end{array}$$
}
\end{proof}

\begin{lem}
T$(2,17;2,n;6) =4n$ for each odd $n$ and $17 \leq n \leq 33$.
\end{lem}

\begin{proof} Let $X_1=(\bbZ_4\times \{0,1,2,3\})\cup \{\infty\}$. For $17 \leq n \leq 23$, let $X_2= (\bbZ_4\times \{4,5,6,7\})\cup (\{a\}\times \{1,\ldots,n-16\})$; for $25 \leq n \leq 31$, let $X_2= (\bbZ_4\times \{4,5,\ldots,9\})\cup (\{a\}\times \{1,\ldots,n-24\})$; for $n = 33$, let $X_2= (\bbZ_4\times \{4,5,\ldots,11\})\cup (\{a\}\times \{1\})$.
Denote $X=X_1\cup X_2$. The codes of size $4n$ are constructed on $\bbZ_2^X$ and the base codewords are listed as follows.

\noindent $n=17$:
{\scriptsize
$$\begin{array}{llllllllll}
\langle\infty, 1_{0}, 1_{4}, 1_{5}\rangle &
\langle\infty, 2_{1}, 1_{7}, 2_{6}\rangle &
\langle3_{0}, 3_{1}, 0_{4}, a_{1}\rangle &
\langle3_{2}, 2_{3}, 3_{5}, a_{1}\rangle &
\langle2_{0}, 2_{3}, 1_{7}, 1_{5}\rangle &
\langle0_{2}, 2_{3}, 0_{6}, 2_{5}\rangle &
\langle0_{1}, 1_{1}, 1_{5}, 0_{4}\rangle &
\langle0_{2}, 1_{2}, 2_{4}, 1_{7}\rangle \\
\langle1_{1}, 1_{3}, 3_{5}, 0_{6}\rangle &
\langle1_{1}, 0_{2}, 3_{6}, 3_{4}\rangle &
\langle1_{0}, 3_{1}, 1_{7}, 0_{6}\rangle &
\langle0_{3}, 1_{3}, 1_{6}, 2_{4}\rangle &
\langle0_{0}, 1_{0}, 2_{5}, 3_{4}\rangle &
\langle2_{0}, 3_{3}, 0_{7}, 3_{7}\rangle &
\langle0_{1}, 2_{2}, 1_{7}, 3_{5}\rangle &
\langle0_{0}, 3_{2}, 0_{6}, 1_{6}\rangle \\
\langle0_{2}, 0_{3}, 0_{4}, 2_{7}\rangle &
\end{array}$$
}

\noindent $n=19$:
{\scriptsize
$$\begin{array}{llllllllll}
\langle\infty, 3_{0}, 1_{4}, 2_{5}\rangle &
\langle\infty, 3_{1}, 2_{6}, 0_{7}\rangle &
\langle1_{0}, 2_{1}, 1_{7}, a_{1}\rangle &
\langle0_{3}, 1_{2}, 0_{6}, a_{1}\rangle &
\langle2_{0}, 3_{3}, 3_{5}, a_{2}\rangle &
\langle0_{1}, 1_{2}, 0_{4}, a_{2}\rangle &
\langle0_{2}, 0_{3}, 1_{4}, a_{3}\rangle &
\langle1_{1}, 2_{0}, 0_{5}, a_{3}\rangle \\
\langle0_{0}, 2_{3}, 0_{5}, 3_{7}\rangle &
\langle3_{3}, 2_{2}, 3_{7}, 2_{5}\rangle &
\langle0_{0}, 0_{3}, 1_{6}, 0_{4}\rangle &
\langle2_{0}, 3_{2}, 3_{4}, 1_{6}\rangle &
\langle3_{2}, 3_{1}, 1_{4}, 3_{7}\rangle &
\langle1_{1}, 0_{2}, 3_{7}, 3_{5}\rangle &
\langle0_{1}, 3_{1}, 0_{6}, 0_{5}\rangle &
\langle2_{3}, 3_{3}, 1_{4}, 1_{6}\rangle \\
\langle2_{1}, 2_{0}, 0_{6}, 1_{4}\rangle &
\langle1_{2}, 0_{2}, 2_{5}, 1_{6}\rangle &
\langle0_{0}, 3_{3}, 1_{7}, 2_{7}\rangle &
\end{array}$$
}

\noindent $n=21$:
{\scriptsize
$$\begin{array}{llllllllll}
\langle\infty, 2_{1}, 2_{6}, 2_{7}\rangle &
\langle\infty, 2_{0}, 0_{5}, 2_{4}\rangle &
\langle3_{2}, 3_{3}, 3_{5}, a_{1}\rangle &
\langle0_{0}, 0_{1}, 1_{6}, a_{1}\rangle &
\langle0_{3}, 3_{2}, 1_{5}, a_{2}\rangle &
\langle3_{1}, 1_{0}, 3_{4}, a_{2}\rangle &
\langle3_{2}, 1_{3}, 3_{6}, a_{3}\rangle &
\langle0_{0}, 1_{1}, 3_{4}, a_{3}\rangle \\
\langle2_{3}, 3_{2}, 0_{4}, a_{4}\rangle &
\langle1_{0}, 0_{1}, 2_{5}, a_{4}\rangle &
\langle2_{3}, 0_{1}, 3_{7}, a_{5}\rangle &
\langle1_{0}, 1_{2}, 0_{5}, a_{5}\rangle &
\langle2_{1}, 1_{3}, 3_{7}, 3_{5}\rangle &
\langle0_{2}, 0_{1}, 3_{6}, 3_{4}\rangle &
\langle0_{1}, 1_{3}, 2_{6}, 0_{5}\rangle &
\langle1_{2}, 0_{1}, 2_{7}, 1_{4}\rangle \\
\langle1_{0}, 3_{2}, 1_{6}, 3_{7}\rangle &
\langle3_{0}, 2_{2}, 3_{5}, 0_{4}\rangle &
\langle2_{0}, 3_{2}, 1_{7}, 0_{6}\rangle &
\langle1_{3}, 2_{3}, 1_{6}, 2_{4}\rangle &
\langle0_{0}, 1_{3}, 0_{7}, 1_{7}\rangle &
\end{array}$$
}

\noindent $n=23$:
{\scriptsize
$$\begin{array}{llllllllll}
\langle\infty, 3_{0}, 3_{5}, 3_{4}\rangle &
\langle\infty, 3_{1}, 1_{7}, 0_{6}\rangle &
\langle1_{2}, 1_{3}, 2_{6}, a_{1}\rangle &
\langle2_{0}, 2_{1}, 3_{4}, a_{1}\rangle &
\langle2_{1}, 1_{0}, 0_{4}, a_{2}\rangle &
\langle3_{3}, 2_{2}, 2_{5}, a_{2}\rangle &
\langle0_{0}, 1_{3}, 2_{5}, a_{3}\rangle &
\langle0_{1}, 2_{2}, 0_{7}, a_{3}\rangle \\
\langle1_{0}, 0_{1}, 2_{5}, a_{4}\rangle &
\langle2_{3}, 3_{2}, 3_{7}, a_{4}\rangle &
\langle2_{3}, 2_{1}, 1_{4}, a_{5}\rangle &
\langle2_{2}, 1_{0}, 1_{6}, a_{5}\rangle &
\langle0_{1}, 1_{3}, 3_{5}, a_{6}\rangle &
\langle0_{0}, 2_{2}, 2_{6}, a_{6}\rangle &
\langle3_{0}, 3_{2}, 2_{5}, a_{7}\rangle &
\langle2_{3}, 0_{1}, 1_{7}, a_{7}\rangle \\
\langle0_{2}, 0_{1}, 3_{7}, 0_{4}\rangle &
\langle0_{3}, 3_{3}, 1_{4}, 2_{6}\rangle &
\langle3_{0}, 2_{2}, 0_{6}, 3_{7}\rangle &
\langle1_{1}, 0_{1}, 3_{6}, 1_{5}\rangle &
\langle1_{3}, 2_{0}, 3_{7}, 1_{6}\rangle &
\langle2_{3}, 0_{0}, 2_{7}, 2_{4}\rangle &
\langle0_{2}, 1_{2}, 3_{4}, 2_{5}\rangle &
\end{array}$$
}

\noindent $n=25$:
{\scriptsize
$$\begin{array}{llllllllll}
\langle\infty, 3_{1}, 3_{7}, 1_{6}\rangle &
\langle\infty, 2_{0}, 2_{4}, 2_{5}\rangle &
\langle\infty, 0_{2}, 0_{9}, 0_{8}\rangle &
\langle3_{3}, 0_{2}, 3_{6}, a_{1}\rangle &
\langle0_{0}, 0_{1}, 1_{4}, a_{1}\rangle &
\langle0_{3}, 1_{3}, 2_{4}, 2_{9}\rangle &
\langle2_{2}, 0_{1}, 0_{6}, 3_{9}\rangle &
\langle0_{1}, 1_{1}, 1_{5}, 0_{4}\rangle \\
\langle1_{2}, 0_{0}, 2_{8}, 2_{6}\rangle &
\langle2_{0}, 1_{0}, 3_{5}, 0_{4}\rangle &
\langle0_{3}, 1_{0}, 1_{7}, 0_{8}\rangle &
\langle1_{2}, 1_{3}, 0_{4}, 3_{8}\rangle &
\langle2_{1}, 2_{2}, 0_{4}, 1_{7}\rangle &
\langle1_{1}, 0_{3}, 1_{8}, 0_{5}\rangle &
\langle3_{1}, 0_{3}, 1_{5}, 2_{6}\rangle &
\langle2_{0}, 0_{2}, 3_{9}, 0_{7}\rangle \\
\langle2_{0}, 0_{1}, 1_{6}, 3_{8}\rangle &
\langle0_{0}, 3_{1}, 0_{8}, 3_{9}\rangle &
\langle2_{1}, 0_{3}, 0_{9}, 0_{7}\rangle &
\langle0_{0}, 1_{3}, 3_{7}, 0_{9}\rangle &
\langle0_{0}, 0_{3}, 3_{5}, 1_{6}\rangle &
\langle0_{1}, 3_{2}, 1_{7}, 2_{8}\rangle &
\langle2_{2}, 1_{2}, 2_{4}, 0_{5}\rangle &
\langle2_{0}, 2_{2}, 2_{6}, 0_{9}\rangle \\
\langle0_{2}, 2_{3}, 0_{5}, 1_{7}\rangle &
\end{array}$$
}

\noindent $n=27$:
{\scriptsize
$$\begin{array}{llllllllll}
\langle\infty, 2_{0}, 2_{4}, 2_{5}\rangle &
\langle\infty, 2_{2}, 0_{8}, 2_{9}\rangle &
\langle\infty, 0_{1}, 1_{6}, 1_{7}\rangle &
\langle2_{3}, 2_{2}, 0_{7}, a_{1}\rangle &
\langle2_{1}, 2_{0}, 3_{4}, a_{1}\rangle &
\langle1_{2}, 2_{3}, 3_{5}, a_{2}\rangle &
\langle2_{1}, 1_{0}, 0_{6}, a_{2}\rangle &
\langle3_{0}, 1_{1}, 1_{4}, a_{3}\rangle \\
\langle2_{3}, 0_{2}, 1_{5}, a_{3}\rangle &
\langle1_{1}, 3_{2}, 1_{6}, 2_{9}\rangle &
\langle1_{1}, 2_{1}, 2_{8}, 3_{5}\rangle &
\langle0_{0}, 0_{2}, 1_{6}, 0_{7}\rangle &
\langle3_{0}, 2_{3}, 2_{5}, 1_{6}\rangle &
\langle0_{0}, 2_{2}, 1_{8}, 1_{7}\rangle &
\langle1_{0}, 3_{3}, 1_{6}, 0_{9}\rangle &
\langle0_{2}, 3_{3}, 2_{9}, 2_{4}\rangle \\
\langle0_{2}, 1_{2}, 0_{6}, 1_{4}\rangle &
\langle1_{1}, 2_{0}, 0_{7}, 3_{9}\rangle &
\langle0_{1}, 3_{3}, 3_{6}, 3_{4}\rangle &
\langle1_{0}, 0_{0}, 2_{5}, 0_{8}\rangle &
\langle0_{3}, 1_{3}, 0_{8}, 2_{4}\rangle &
\langle1_{0}, 1_{3}, 1_{9}, 0_{7}\rangle &
\langle2_{1}, 2_{3}, 0_{8}, 2_{7}\rangle &
\langle2_{0}, 3_{2}, 0_{9}, 0_{8}\rangle \\
\langle1_{1}, 3_{3}, 1_{9}, 0_{8}\rangle &
\langle0_{1}, 1_{2}, 2_{7}, 0_{5}\rangle &
\langle0_{1}, 3_{2}, 2_{4}, 3_{5}\rangle &
\end{array}$$
}

\noindent $n=29$:
{\scriptsize
$$\begin{array}{llllllllll}
\langle\infty, 0_{0}, 0_{4}, 0_{5}\rangle &
\langle\infty, 0_{1}, 0_{7}, 0_{6}\rangle &
\langle\infty, 2_{2}, 2_{9}, 2_{8}\rangle &
\langle2_{3}, 2_{2}, 2_{5}, a_{1}\rangle &
\langle0_{1}, 0_{0}, 1_{4}, a_{1}\rangle &
\langle2_{1}, 1_{0}, 0_{4}, a_{2}\rangle &
\langle2_{3}, 1_{2}, 3_{5}, a_{2}\rangle &
\langle2_{3}, 0_{2}, 2_{6}, a_{3}\rangle \\
\langle3_{1}, 1_{0}, 3_{4}, a_{3}\rangle &
\langle3_{3}, 0_{2}, 3_{4}, a_{4}\rangle &
\langle1_{1}, 2_{0}, 3_{5}, a_{4}\rangle &
\langle0_{2}, 0_{0}, 3_{5}, a_{5}\rangle &
\langle1_{1}, 1_{3}, 0_{4}, a_{5}\rangle &
\langle0_{0}, 3_{0}, 2_{8}, 3_{9}\rangle &
\langle0_{0}, 2_{3}, 1_{9}, 3_{7}\rangle &
\langle0_{3}, 3_{0}, 3_{7}, 1_{6}\rangle \\
\langle0_{2}, 0_{1}, 1_{8}, 3_{6}\rangle &
\langle3_{1}, 0_{3}, 2_{7}, 2_{9}\rangle &
\langle3_{1}, 2_{3}, 1_{8}, 0_{5}\rangle &
\langle3_{1}, 0_{2}, 1_{6}, 0_{6}\rangle &
\langle3_{1}, 1_{3}, 2_{8}, 1_{7}\rangle &
\langle1_{3}, 0_{3}, 1_{9}, 2_{4}\rangle &
\langle0_{1}, 3_{1}, 1_{9}, 3_{5}\rangle &
\langle0_{2}, 1_{0}, 2_{7}, 3_{9}\rangle \\
\langle1_{2}, 3_{1}, 0_{7}, 3_{8}\rangle &
\langle2_{0}, 1_{3}, 3_{8}, 3_{6}\rangle &
\langle0_{2}, 3_{2}, 1_{9}, 1_{4}\rangle &
\langle3_{0}, 3_{3}, 2_{6}, 3_{8}\rangle &
\langle0_{0}, 1_{2}, 2_{5}, 2_{7}\rangle &
\end{array}$$
}

\noindent $n=31$:
{\scriptsize
$$\begin{array}{llllllllll}
\langle\infty, 2_{3}, 1_{9}, 1_{5}\rangle &
\langle\infty, 2_{0}, 0_{6}, 1_{7}\rangle &
\langle\infty, 3_{2}, 0_{8}, 3_{4}\rangle &
\langle3_{3}, 3_{1}, 0_{4}, a_{1}\rangle &
\langle3_{0}, 0_{2}, 3_{5}, a_{1}\rangle &
\langle3_{1}, 2_{0}, 2_{4}, a_{2}\rangle &
\langle1_{3}, 3_{2}, 1_{6}, a_{2}\rangle &
\langle3_{1}, 1_{0}, 3_{4}, a_{3}\rangle \\
\langle1_{3}, 0_{2}, 0_{6}, a_{3}\rangle &
\langle1_{0}, 0_{1}, 2_{6}, a_{4}\rangle &
\langle0_{2}, 0_{3}, 0_{7}, a_{4}\rangle &
\langle3_{0}, 0_{3}, 0_{4}, a_{5}\rangle &
\langle2_{1}, 0_{2}, 0_{5}, a_{5}\rangle &
\langle1_{2}, 1_{1}, 2_{5}, a_{6}\rangle &
\langle3_{0}, 2_{3}, 3_{6}, a_{6}\rangle &
\langle1_{1}, 3_{3}, 1_{5}, a_{7}\rangle \\
\langle1_{2}, 3_{0}, 3_{7}, a_{7}\rangle &
\langle3_{0}, 0_{0}, 0_{8}, 2_{5}\rangle &
\langle1_{1}, 0_{1}, 3_{7}, 3_{8}\rangle &
\langle3_{3}, 0_{2}, 3_{9}, 1_{7}\rangle &
\langle1_{0}, 1_{2}, 1_{9}, 0_{6}\rangle &
\langle0_{1}, 3_{2}, 0_{9}, 0_{6}\rangle &
\langle0_{3}, 0_{0}, 2_{9}, 1_{5}\rangle &
\langle2_{0}, 1_{2}, 0_{7}, 1_{8}\rangle \\
\langle3_{3}, 2_{3}, 0_{8}, 1_{4}\rangle &
\langle2_{3}, 1_{1}, 0_{6}, 1_{8}\rangle &
\langle1_{2}, 2_{2}, 0_{8}, 0_{4}\rangle &
\langle0_{3}, 2_{0}, 0_{8}, 1_{9}\rangle &
\langle3_{0}, 3_{1}, 0_{7}, 0_{9}\rangle &
\langle2_{2}, 1_{1}, 3_{4}, 0_{9}\rangle &
\langle0_{1}, 3_{3}, 3_{5}, 0_{7}\rangle &
\end{array}$$
}

\noindent $n=33$:
{\scriptsize
$$\begin{array}{llllllllll}
\langle\infty, 2_{0}, 0_{8}, 0_{9}\rangle &
\langle\infty, 2_{1}, 3_{4}, 3_{6}\rangle &
\langle\infty, 3_{3}, 2_{11}, 3_{5}\rangle &
\langle\infty, 3_{2}, 1_{10}, 3_{7}\rangle &
\langle2_{2}, 3_{3}, 0_{5}, a_{1}\rangle &
\langle3_{1}, 1_{0}, 0_{7}, a_{1}\rangle &
\langle2_{1}, 1_{3}, 1_{11}, 3_{8}\rangle &
\langle2_{2}, 1_{2}, 1_{11}, 1_{5}\rangle \\
\langle3_{1}, 0_{1}, 3_{9}, 2_{7}\rangle &
\langle1_{2}, 1_{1}, 0_{4}, 2_{10}\rangle &
\langle0_{1}, 1_{2}, 2_{4}, 3_{8}\rangle &
\langle0_{3}, 0_{0}, 3_{9}, 1_{7}\rangle &
\langle1_{3}, 3_{2}, 2_{6}, 2_{9}\rangle &
\langle3_{0}, 0_{0}, 3_{6}, 1_{11}\rangle &
\langle2_{2}, 1_{0}, 0_{4}, 1_{10}\rangle &
\langle3_{0}, 1_{2}, 3_{11}, 0_{8}\rangle \\
\langle1_{0}, 0_{2}, 1_{9}, 1_{5}\rangle &
\langle1_{0}, 1_{1}, 2_{5}, 3_{6}\rangle &
\langle1_{0}, 2_{1}, 0_{8}, 0_{11}\rangle &
\langle0_{1}, 2_{2}, 3_{6}, 2_{10}\rangle &
\langle0_{2}, 0_{3}, 2_{9}, 3_{7}\rangle &
\langle3_{0}, 1_{3}, 2_{10}, 0_{6}\rangle &
\langle2_{0}, 2_{2}, 2_{8}, 0_{7}\rangle &
\langle3_{1}, 3_{3}, 3_{6}, 0_{11}\rangle \\
\langle0_{3}, 1_{3}, 1_{8}, 3_{4}\rangle &
\langle0_{0}, 3_{1}, 1_{9}, 2_{10}\rangle &
\langle0_{0}, 3_{3}, 1_{10}, 2_{5}\rangle &
\langle3_{2}, 0_{1}, 0_{7}, 0_{8}\rangle &
\langle2_{2}, 1_{3}, 3_{11}, 2_{4}\rangle &
\langle0_{1}, 1_{3}, 0_{10}, 1_{9}\rangle &
\langle0_{0}, 1_{3}, 1_{4}, 3_{5}\rangle &
\langle0_0, 2_0, 0_4, 2_4\rangle^s \\
\langle0_1, 2_1, 0_5, 2_5\rangle^s &
\langle0_2, 2_2, 0_6, 2_6\rangle^s &
\langle0_3, 2_3, 0_7, 2_7\rangle^s &
\end{array}$$
}
Note that each of the codewords marked $s$ only generates two codewords.
\end{proof}

\begin{lem}
T$(2,21;2,n;6) =5n$ for each odd $n$ and $21 \leq n \leq 41$.
\end{lem}

\begin{proof} Let $X_1=(\bbZ_5\times \{0,1,2,3\})\cup \{\infty\}$. For $21 \leq n \leq 29$, let $X_2= (\bbZ_5\times \{4,5,6,7\})\cup (\{a\}\times \{1,\ldots,n-20\})$; for $31 \leq n \leq 39$, let $X_2= (\bbZ_5\times \{4,5,\ldots,9\})\cup (\{a\}\times \{1,\ldots,n-30\})$; for $n = 41$, let $X_2= (\bbZ_5\times \{4,5,\ldots,11\})\cup (\{a\}\times \{1\})$.
Denote $X=X_1\cup X_2$. The desired codes of size $5n$ are constructed on $\bbZ_2^X$ and the base codewords are listed as follows.

\noindent $n=21$:
{\scriptsize
$$\begin{array}{llllllllll}
\langle\infty, 2_0, 4_5, 4_4\rangle &
\langle\infty, 0_1, 0_7, 3_6\rangle &
\langle1_1, 2_0, 2_4, a_1\rangle &
\langle2_2, 0_3, 3_6, a_1\rangle &
\langle2_2, 2_3, 3_5, 3_7\rangle &
\langle4_1, 4_2, 1_4, 3_5\rangle &
\langle1_2, 4_2, 4_6, 1_5\rangle &
\langle1_2, 4_0, 4_5, 0_4\rangle \\
\langle1_3, 4_3, 4_5, 3_6\rangle &
\langle1_2, 0_3, 0_6, 1_4\rangle &
\langle0_3, 4_2, 2_7, 1_6\rangle &
\langle2_0, 0_1, 3_7, 2_7\rangle &
\langle0_1, 4_1, 4_6, 4_4\rangle &
\langle0_3, 4_3, 2_4, 4_4\rangle &
\langle1_0, 2_0, 0_6, 2_6\rangle &
\langle1_0, 3_1, 4_5, 0_5\rangle \\
\langle0_0, 4_3, 1_5, 3_7\rangle &
\langle3_0, 4_1, 2_4, 0_6\rangle &
\langle0_1, 3_1, 4_7, 3_5\rangle &
\langle2_0, 2_2, 0_4, 1_7\rangle &
\langle0_2, 2_3, 0_7, 2_7\rangle &
\end{array}$$
}

\noindent $n=23$:
{\scriptsize
$$\begin{array}{llllllllll}
\langle\infty, 3_0, 3_4, 3_5\rangle &
\langle\infty, 2_1, 2_6, 2_7\rangle &
\langle2_0, 2_1, 3_4, a_1\rangle &
\langle1_2, 3_3, 3_5, a_1\rangle &
\langle2_0, 3_1, 0_4, a_2\rangle &
\langle1_2, 2_3, 4_6, a_2\rangle &
\langle0_2, 4_3, 3_5, a_3\rangle &
\langle3_0, 0_1, 0_4, a_3\rangle \\
\langle1_2, 1_3, 1_6, 3_7\rangle &
\langle0_0, 4_1, 3_7, 2_7\rangle &
\langle0_1, 1_2, 1_7, 2_6\rangle &
\langle0_3, 1_3, 0_4, 0_7\rangle &
\langle1_2, 4_3, 0_5, 0_7\rangle &
\langle1_2, 0_2, 2_4, 1_5\rangle &
\langle0_1, 2_1, 3_5, 3_6\rangle &
\langle0_0, 4_0, 1_5, 3_5\rangle \\
\langle0_1, 1_1, 0_5, 4_4\rangle &
\langle0_0, 2_0, 1_6, 0_6\rangle &
\langle0_2, 3_2, 3_4, 2_6\rangle &
\langle0_3, 2_3, 3_4, 3_6\rangle &
\langle0_0, 3_2, 4_7, 1_7\rangle &
\langle1_0, 3_3, 0_4, 1_7\rangle &
\langle0_1, 0_3, 2_5, 4_6\rangle &
\end{array}$$
}

\noindent $n=25$:
{\scriptsize
$$\begin{array}{llllllllll}
\langle\infty, 2_1, 2_7, 0_6\rangle &
\langle\infty, 3_3, 2_5, 4_4\rangle &
\langle3_3, 2_2, 1_4, a_1\rangle &
\langle3_1, 0_0, 4_7, a_1\rangle &
\langle0_0, 0_1, 3_4, a_2\rangle &
\langle3_2, 0_3, 0_5, a_2\rangle &
\langle1_2, 1_0, 1_7, a_3\rangle &
\langle3_3, 1_1, 2_4, a_3\rangle \\
\langle3_0, 2_1, 3_5, a_4\rangle &
\langle3_3, 4_2, 0_4, a_4\rangle &
\langle1_1, 4_3, 0_7, a_5\rangle &
\langle1_2, 2_0, 1_6, a_5\rangle &
\langle2_1, 3_2, 1_5, 1_4\rangle &
\langle3_3, 3_2, 2_6, 2_7\rangle &
\langle1_1, 3_2, 3_4, 4_5\rangle &
\langle1_0, 4_0, 0_5, 3_4\rangle \\
\langle3_0, 3_3, 3_6, 3_4\rangle &
\langle4_2, 3_2, 1_6, 0_7\rangle &
\langle1_3, 1_1, 2_6, 3_5\rangle &
\langle4_0, 1_1, 0_6, 1_6\rangle &
\langle1_1, 0_0, 3_6, 1_4\rangle &
\langle1_2, 1_1, 1_5, 4_7\rangle &
\langle0_3, 2_2, 1_5, 3_6\rangle &
\langle2_0, 1_3, 4_5, 4_7\rangle \\
\langle0_0, 1_3, 1_7, 3_7\rangle &
\end{array}$$
}

\noindent $n=27$:
{\scriptsize
$$\begin{array}{llllllllll}
\langle\infty, 0_2, 0_5, 0_4\rangle &
\langle\infty, 0_1, 3_7, 0_6\rangle &
\langle3_3, 3_2, 4_6, a_1\rangle &
\langle2_1, 2_0, 0_4, a_1\rangle &
\langle0_3, 2_2, 3_7, a_2\rangle &
\langle2_0, 0_1, 4_4, a_2\rangle &
\langle3_1, 1_0, 4_5, a_3\rangle &
\langle0_2, 2_3, 2_6, a_3\rangle \\
\langle4_0, 3_1, 2_6, a_4\rangle &
\langle0_2, 4_3, 4_7, a_4\rangle &
\langle4_3, 3_0, 2_4, a_5\rangle &
\langle1_1, 3_2, 3_6, a_5\rangle &
\langle1_3, 2_1, 2_7, a_6\rangle &
\langle2_2, 0_0, 0_4, a_6\rangle &
\langle1_0, 4_2, 2_5, a_7\rangle &
\langle0_1, 0_3, 2_7, a_7\rangle \\
\langle1_3, 2_0, 2_5, 4_5\rangle &
\langle0_1, 1_3, 3_6, 1_4\rangle &
\langle4_1, 1_1, 1_5, 0_7\rangle &
\langle4_2, 0_3, 3_6, 4_7\rangle &
\langle3_0, 4_0, 0_7, 3_6\rangle &
\langle3_2, 0_2, 4_5, 2_4\rangle &
\langle0_0, 1_1, 2_6, 1_4\rangle &
\langle2_3, 1_3, 1_5, 3_4\rangle \\
\langle2_2, 1_1, 3_4, 4_5\rangle &
\langle0_2, 4_0, 2_7, 3_7\rangle &
\langle0_0, 2_3, 4_5, 1_6\rangle &
\end{array}$$
}

\noindent $n=29$:
{\scriptsize
$$\begin{array}{llllllllll}
\langle\infty, 0_0, 2_5, 1_4\rangle &
\langle\infty, 0_1, 4_7, 0_6\rangle &
\langle4_1, 0_2, 1_7, a_1\rangle &
\langle3_0, 0_3, 0_4, a_1\rangle &
\langle2_3, 2_2, 1_5, a_2\rangle &
\langle4_0, 0_1, 3_4, a_2\rangle &
\langle0_3, 4_2, 2_5, a_3\rangle &
\langle4_1, 2_0, 0_6, a_3\rangle \\
\langle1_1, 3_0, 4_5, a_4\rangle &
\langle1_3, 3_2, 0_4, a_4\rangle &
\langle0_0, 4_1, 3_5, a_5\rangle &
\langle1_3, 2_2, 0_6, a_5\rangle &
\langle0_0, 0_2, 2_7, a_6\rangle &
\langle4_1, 3_3, 4_4, a_6\rangle &
\langle1_3, 4_1, 3_6, a_7\rangle &
\langle1_2, 0_0, 0_7, a_7\rangle \\
\langle2_0, 0_3, 2_4, a_8\rangle &
\langle3_2, 4_1, 4_5, a_8\rangle &
\langle2_1, 2_3, 2_7, a_9\rangle &
\langle2_0, 1_2, 1_5, a_9\rangle &
\langle3_0, 0_0, 4_6, 4_7\rangle &
\langle2_3, 3_3, 4_7, 3_5\rangle &
\langle1_2, 1_1, 0_4, 3_5\rangle &
\langle0_3, 0_0, 3_4, 3_7\rangle \\
\langle2_1, 0_3, 0_6, 3_5\rangle &
\langle3_3, 4_0, 4_6, 1_6\rangle &
\langle1_2, 4_2, 2_4, 0_6\rangle &
\langle2_1, 0_2, 3_7, 0_7\rangle &
\langle0_1, 2_2, 2_4, 2_6\rangle &
\end{array}$$
}

\noindent $n=31$:
{\scriptsize
$$\begin{array}{llllllllll}
\langle\infty, 3_1, 3_6, 3_7\rangle &
\langle\infty, 0_0, 2_4, 4_5\rangle &
\langle\infty, 1_2, 0_8, 1_9\rangle &
\langle4_0, 4_1, 0_4, a_1\rangle &
\langle1_2, 0_3, 0_5, a_1\rangle &
\langle1_0, 1_3, 4_8, 2_9\rangle &
\langle1_2, 2_3, 3_6, 3_5\rangle &
\langle0_0, 1_0, 4_9, 2_8\rangle \\
\langle0_2, 3_3, 3_8, 3_9\rangle &
\langle0_1, 3_1, 2_5, 3_8\rangle &
\langle0_2, 2_2, 0_6, 3_4\rangle &
\langle0_0, 2_3, 0_5, 4_7\rangle &
\langle0_1, 4_1, 4_4, 3_6\rangle &
\langle1_1, 1_2, 0_9, 3_4\rangle &
\langle0_1, 3_2, 3_5, 0_9\rangle &
\langle2_0, 3_1, 3_5, 4_9\rangle \\
\langle1_0, 0_1, 4_7, 1_8\rangle &
\langle0_2, 1_2, 0_4, 1_7\rangle &
\langle1_1, 2_3, 2_7, 3_8\rangle &
\langle1_1, 3_3, 2_6, 0_8\rangle &
\langle0_0, 4_2, 2_5, 2_7\rangle &
\langle0_3, 1_3, 3_6, 2_4\rangle &
\langle0_3, 3_3, 1_7, 2_9\rangle &
\langle1_2, 3_3, 2_5, 2_8\rangle \\
\langle1_1, 1_3, 4_9, 4_4\rangle &
\langle1_0, 4_3, 4_4, 4_6\rangle &
\langle0_0, 2_2, 1_7, 4_8\rangle &
\langle0_0, 3_0, 0_6, 4_6\rangle &
\langle1_0, 2_3, 1_4, 4_5\rangle &
\langle0_1, 1_2, 3_7, 2_6\rangle &
\langle0_0, 3_1, 0_7, 0_9\rangle &
\end{array}$$
}

\noindent $n=33$:
{\scriptsize
$$\begin{array}{llllllllll}
\langle\infty, 4_2, 1_8, 4_9\rangle &
\langle\infty, 3_0, 3_4, 3_5\rangle &
\langle\infty, 1_1, 1_6, 3_7\rangle &
\langle3_1, 3_0, 4_4, a_1\rangle &
\langle3_3, 3_2, 3_5, a_1\rangle &
\langle1_3, 0_2, 2_5, a_2\rangle &
\langle3_0, 4_1, 1_4, a_2\rangle &
\langle3_2, 2_3, 1_5, a_3\rangle \\
\langle3_0, 0_1, 4_7, a_3\rangle &
\langle1_3, 2_3, 3_7, 1_4\rangle &
\langle2_1, 4_3, 2_7, 3_6\rangle &
\langle0_0, 1_3, 0_9, 3_9\rangle &
\langle0_2, 2_2, 3_4, 2_4\rangle &
\langle1_0, 2_2, 1_7, 0_8\rangle &
\langle1_1, 3_1, 2_9, 0_8\rangle &
\langle0_3, 3_3, 4_8, 1_4\rangle \\
\langle1_0, 4_2, 0_9, 3_5\rangle &
\langle2_1, 0_2, 0_6, 4_9\rangle &
\langle1_1, 2_3, 4_8, 4_5\rangle &
\langle1_1, 1_2, 2_8, 3_6\rangle &
\langle1_2, 3_3, 1_8, 4_6\rangle &
\langle4_1, 0_2, 0_7, 1_5\rangle &
\langle0_0, 2_2, 4_7, 3_7\rangle &
\langle2_0, 0_0, 1_6, 0_6\rangle \\
\langle1_1, 0_3, 1_9, 0_6\rangle &
\langle1_0, 4_3, 2_5, 2_9\rangle &
\langle0_1, 4_1, 0_5, 3_4\rangle &
\langle0_0, 0_3, 2_6, 0_8\rangle &
\langle4_0, 3_2, 2_8, 1_7\rangle &
\langle2_2, 0_3, 0_9, 3_6\rangle &
\langle1_0, 0_0, 4_5, 2_8\rangle &
\langle0_0, 0_2, 2_9, 4_4\rangle \\
\langle0_1, 3_3, 0_4, 3_7\rangle &
\end{array}$$
}

\noindent $n=35$:
{\scriptsize
$$\begin{array}{llllllllll}
\langle\infty, 0_3, 3_7, 2_9\rangle &
\langle\infty, 0_0, 0_5, 1_4\rangle &
\langle\infty, 0_2, 0_8, 4_6\rangle &
\langle4_1, 4_0, 0_7, a_1\rangle &
\langle3_3, 0_2, 3_5, a_1\rangle &
\langle2_0, 0_1, 2_4, a_2\rangle &
\langle4_3, 3_2, 4_6, a_2\rangle &
\langle2_3, 0_2, 4_7, a_3\rangle \\
\langle1_1, 2_0, 1_5, a_3\rangle &
\langle0_3, 3_0, 4_6, a_4\rangle &
\langle1_1, 1_2, 4_7, a_4\rangle &
\langle3_1, 1_0, 0_7, a_5\rangle &
\langle1_2, 0_3, 1_4, a_5\rangle &
\langle3_2, 2_2, 4_9, 3_5\rangle &
\langle4_3, 4_2, 1_6, 2_9\rangle &
\langle1_3, 0_1, 1_9, 2_9\rangle \\
\langle1_1, 4_3, 1_4, 4_4\rangle &
\langle0_1, 4_3, 0_7, 0_6\rangle &
\langle0_1, 2_2, 1_4, 4_5\rangle &
\langle2_0, 0_3, 4_9, 4_7\rangle &
\langle3_2, 2_0, 0_7, 3_9\rangle &
\langle1_3, 1_1, 0_4, 0_8\rangle &
\langle4_1, 2_2, 3_7, 4_8\rangle &
\langle2_0, 4_0, 2_9, 0_5\rangle \\
\langle4_2, 2_0, 1_4, 0_4\rangle &
\langle4_1, 1_1, 0_6, 2_8\rangle &
\langle0_0, 1_1, 3_6, 4_6\rangle &
\langle2_2, 4_0, 3_8, 1_8\rangle &
\langle3_1, 4_1, 0_5, 3_9\rangle &
\langle2_2, 2_0, 1_9, 2_6\rangle &
\langle4_2, 0_0, 2_4, 2_6\rangle &
\langle0_0, 0_3, 0_7, 3_8\rangle \\
\langle2_3, 4_3, 4_8, 3_5\rangle &
\langle0_0, 4_3, 2_5, 0_8\rangle &
\langle0_1, 4_2, 3_5, 2_8\rangle &
\end{array}$$
}

\noindent $n=37$:
{\scriptsize
$$\begin{array}{llllllllll}
\langle\infty, 3_0, 2_5, 2_4\rangle &
\langle\infty, 3_3, 0_8, 3_9\rangle &
\langle\infty, 0_2, 2_6, 1_7\rangle &
\langle3_1, 4_2, 4_7, a_1\rangle &
\langle2_3, 2_0, 2_5, a_1\rangle &
\langle2_1, 0_2, 4_4, a_2\rangle &
\langle3_3, 4_0, 1_6, a_2\rangle &
\langle0_2, 3_3, 2_4, a_3\rangle \\
\langle0_0, 3_1, 0_6, a_3\rangle &
\langle1_0, 4_3, 4_4, a_4\rangle &
\langle4_1, 4_2, 4_5, a_4\rangle &
\langle1_3, 1_2, 0_6, a_5\rangle &
\langle0_1, 0_0, 3_5, a_5\rangle &
\langle0_2, 2_0, 3_4, a_6\rangle &
\langle0_1, 0_3, 4_5, a_6\rangle &
\langle3_2, 1_0, 0_7, a_7\rangle \\
\langle4_1, 3_3, 0_4, a_7\rangle &
\langle3_1, 1_1, 0_9, 1_4\rangle &
\langle4_3, 3_2, 4_6, 1_7\rangle &
\langle0_2, 1_1, 0_4, 2_8\rangle &
\langle0_0, 4_2, 3_8, 4_6\rangle &
\langle2_3, 1_3, 0_7, 4_5\rangle &
\langle0_0, 2_0, 4_9, 0_7\rangle &
\langle2_0, 4_3, 3_9, 4_8\rangle \\
\langle3_0, 4_3, 1_9, 0_7\rangle &
\langle3_1, 0_3, 1_6, 1_9\rangle &
\langle0_0, 4_1, 4_8, 1_5\rangle &
\langle1_0, 2_1, 2_7, 1_8\rangle &
\langle2_1, 4_2, 3_5, 0_8\rangle &
\langle0_1, 3_0, 1_6, 4_6\rangle &
\langle0_0, 1_2, 2_4, 1_8\rangle &
\langle1_1, 2_1, 0_7, 2_9\rangle \\
\langle4_2, 3_3, 3_7, 1_9\rangle &
\langle2_1, 0_3, 4_8, 2_6\rangle &
\langle0_2, 0_0, 2_5, 0_9\rangle &
\langle3_3, 0_3, 1_4, 1_8\rangle &
\langle0_2, 2_2, 3_5, 3_9\rangle &
\end{array}$$
}

\noindent $n=39$:
{\scriptsize
$$\begin{array}{llllllllll}
\langle\infty, 3_1, 4_9, 4_8\rangle &
\langle\infty, 0_2, 1_4, 1_5\rangle &
\langle\infty, 2_3, 3_6, 3_7\rangle &
\langle2_2, 1_3, 1_5, a_1\rangle &
\langle2_0, 3_1, 1_7, a_1\rangle &
\langle3_3, 2_0, 3_6, a_2\rangle &
\langle3_1, 2_2, 1_4, a_2\rangle &
\langle0_0, 2_1, 2_4, a_3\rangle \\
\langle2_2, 4_3, 0_5, a_3\rangle &
\langle3_0, 3_1, 3_5, a_4\rangle &
\langle2_3, 4_2, 1_4, a_4\rangle &
\langle1_3, 4_0, 0_5, a_5\rangle &
\langle2_2, 4_1, 3_6, a_5\rangle &
\langle1_2, 2_0, 3_7, a_6\rangle &
\langle4_3, 1_1, 0_4, a_6\rangle &
\langle2_0, 1_3, 3_4, a_7\rangle \\
\langle2_1, 4_2, 3_7, a_7\rangle &
\langle2_2, 4_0, 2_6, a_8\rangle &
\langle4_3, 4_1, 2_5, a_8\rangle &
\langle3_0, 0_2, 3_6, a_9\rangle &
\langle2_1, 3_3, 2_7, a_9\rangle &
\langle4_1, 2_1, 1_5, 4_8\rangle &
\langle0_1, 4_1, 0_6, 2_9\rangle &
\langle4_0, 0_0, 4_4, 0_9\rangle \\
\langle4_3, 2_3, 4_8, 4_7\rangle &
\langle2_1, 1_3, 1_9, 4_7\rangle &
\langle1_3, 0_2, 0_9, 2_8\rangle &
\langle0_3, 2_0, 3_8, 0_4\rangle &
\langle0_0, 3_1, 4_5, 2_7\rangle &
\langle4_0, 4_3, 2_9, 1_5\rangle &
\langle3_2, 4_2, 4_8, 0_9\rangle &
\langle1_0, 3_0, 0_9, 1_8\rangle \\
\langle0_0, 4_1, 2_6, 2_8\rangle &
\langle0_3, 1_3, 2_9, 3_6\rangle &
\langle1_1, 1_2, 0_8, 3_6\rangle &
\langle4_0, 4_2, 2_7, 4_7\rangle &
\langle2_0, 3_2, 0_5, 1_8\rangle &
\langle4_2, 3_1, 4_4, 3_9\rangle &
\langle0_2, 0_3, 3_4, 4_6\rangle &
\end{array}$$
}

\noindent $n=41$:
{\scriptsize
$$\begin{array}{llllllllll}
\langle\infty, 1_{0}, 3_{10}, 3_{6}\rangle &
\langle\infty, 0_{3}, 1_{11}, 3_{5}\rangle &
\langle\infty, 0_{2}, 3_{4}, 4_{7}\rangle &
\langle\infty, 3_{1}, 0_{8}, 3_{9}\rangle &
\langle0_{1}, 3_{0}, 0_{7}, a_1\rangle &
\langle3_{2}, 3_{3}, 3_{6}, a_1\rangle &
\langle2_{1}, 4_{0}, 3_{6}, 4_{5}\rangle &
\langle0_{1}, 1_{1}, 4_{8}, 2_{10}\rangle \\
\langle2_{1}, 3_{2}, 4_{9}, 1_{9}\rangle &
\langle1_{1}, 3_{1}, 0_{7}, 3_{11}\rangle &
\langle2_{2}, 0_{3}, 1_{6}, 1_{8}\rangle &
\langle0_{2}, 2_{3}, 0_{8}, 4_{10}\rangle &
\langle3_{2}, 3_{1}, 4_{4}, 4_{8}\rangle &
\langle4_{0}, 1_{2}, 4_{11}, 3_{4}\rangle &
\langle0_{1}, 1_{3}, 0_{6}, 3_{9}\rangle &
\langle1_{3}, 2_{0}, 3_{7}, 0_{5}\rangle \\
\langle0_{0}, 1_{1}, 1_{10}, 2_{5}\rangle &
\langle2_{3}, 3_{3}, 1_{9}, 2_{8}\rangle &
\langle0_{0}, 3_{3}, 0_{6}, 4_{9}\rangle &
\langle2_{1}, 0_{2}, 4_{4}, 0_{10}\rangle &
\langle4_{0}, 1_{3}, 4_{7}, 2_{7}\rangle &
\langle2_{0}, 0_{2}, 1_{10}, 3_{6}\rangle &
\langle0_{3}, 2_{1}, 1_{10}, 0_{5}\rangle &
\langle0_{0}, 0_{3}, 0_{9}, 3_{6}\rangle \\
\langle0_{1}, 2_{3}, 1_{11}, 4_{4}\rangle &
\langle2_{0}, 4_{0}, 4_{8}, 0_{4}\rangle &
\langle2_{3}, 2_{1}, 2_{4}, 0_{11}\rangle &
\langle0_{0}, 1_{0}, 3_{11}, 3_{9}\rangle &
\langle1_{3}, 2_{2}, 2_{5}, 3_{11}\rangle &
\langle3_{1}, 4_{0}, 1_{4}, 3_{8}\rangle &
\langle1_{1}, 0_{2}, 2_{7}, 1_{5}\rangle &
\langle2_{0}, 2_{2}, 1_{11}, 0_{10}\rangle \\
\langle0_{2}, 1_{2}, 4_{5}, 0_{9}\rangle &
\langle0_{1}, 2_{2}, 4_{6}, 3_{6}\rangle &
\langle1_{1}, 1_{0}, 0_{5}, 2_{9}\rangle &
\langle0_{2}, 3_{2}, 0_{11}, 3_{7}\rangle &
\langle3_{1}, 2_{3}, 1_{7}, 2_{11}\rangle &
\langle0_{2}, 1_{0}, 2_{8}, 2_{5}\rangle &
\langle1_{2}, 2_{3}, 2_{7}, 4_{8}\rangle &
\langle2_{3}, 0_{3}, 3_{4}, 0_{10}\rangle \\
\langle0_{0}, 1_{3}, 0_{4}, 0_{10}\rangle &
\end{array}$$
}
\end{proof}

\begin{lem}
T$(2,25;2,n;6) =6n$ for each odd $n$ and $25 \leq n \leq 49$.
\end{lem}

\begin{proof} Let $X_1=(\bbZ_6\times \{0,1,2,3\})\cup \{\infty\}$. For $25 \leq n \leq 35$, let $X_2= (\bbZ_6\times \{4,5,6,7\})\cup (\{a\}\times \{1,\ldots,n-24\})$; for $37 \leq n \leq 47$, let $X_2= (\bbZ_6\times \{4,5,\ldots,9\})\cup (\{a\}\times \{1,\ldots,n-36\})$; for $n = 49$, let $X_2= (\bbZ_6\times \{4,5,\ldots,11\})\cup (\{a\}\times \{1\})$.
Denote $X=X_1\cup X_2$. The desired codes of size $6n$ are constructed on $\bbZ_2^X$ and the base codewords are listed as follows.

\noindent $n=25$:
{\scriptsize
$$
\begin{array}{llllllllll}
\langle\infty, 5_{0}, 0_{5}, 5_{4}\rangle &
\langle\infty, 5_{1}, 5_{7}, 5_{6}\rangle &
\langle1_{2}, 3_{3}, 1_{5}, a_{1}\rangle &
\langle1_{0}, 1_{1}, 2_{4}, a_{1}\rangle &
\langle5_{0}, 3_{3}, 4_{6}, 5_{7}\rangle &
\langle0_{0}, 5_{0}, 4_{4}, 2_{4}\rangle &
\langle0_{3}, 3_{0}, 3_{5}, 3_{6}\rangle &
\langle5_{0}, 4_{2}, 4_{7}, 2_{5}\rangle \\
\langle0_{2}, 0_{0}, 2_{7}, 5_{5}\rangle &
\langle3_{1}, 4_{2}, 1_{4}, 0_{6}\rangle &
\langle4_{0}, 0_{1}, 5_{6}, 1_{6}\rangle &
\langle5_{2}, 5_{3}, 4_{4}, 0_{5}\rangle &
\langle4_{0}, 0_{3}, 0_{5}, 2_{6}\rangle &
\langle4_{1}, 5_{3}, 5_{7}, 3_{4}\rangle &
\langle4_{0}, 2_{1}, 2_{5}, 0_{6}\rangle &
\langle0_{1}, 1_{0}, 4_{7}, 2_{7}\rangle \\
\langle2_{1}, 3_{1}, 0_{5}, 4_{5}\rangle &
\langle1_{3}, 0_{2}, 1_{4}, 1_{6}\rangle &
\langle3_{3}, 4_{3}, 5_{4}, 2_{7}\rangle &
\langle3_{2}, 5_{2}, 3_{6}, 5_{4}\rangle &
\langle0_{3}, 2_{1}, 4_{6}, 1_{7}\rangle &
\langle3_{2}, 5_{1}, 1_{4}, 2_{7}\rangle &
\langle4_{3}, 1_{1}, 1_{4}, 0_{5}\rangle &
\langle5_{2}, 3_{3}, 2_{6}, 0_{7}\rangle \\
\langle0_{2}, 1_{2}, 3_{5}, 4_{7}\rangle &
\end{array}$$
}

\noindent $n=27$:
{\scriptsize
$$\begin{array}{llllllllll}
\langle\infty, 1_{1}, 1_{7}, 1_{6}\rangle &
\langle\infty, 2_{0}, 2_{4}, 2_{5}\rangle &
\langle1_{2}, 1_{3}, 1_{5}, a_{1}\rangle &
\langle2_{0}, 2_{1}, 3_{4}, a_{1}\rangle &
\langle0_{0}, 1_{1}, 3_{4}, a_{2}\rangle &
\langle2_{2}, 1_{3}, 0_{5}, a_{2}\rangle &
\langle5_{0}, 1_{1}, 1_{4}, a_{3}\rangle &
\langle1_{2}, 2_{3}, 2_{6}, a_{3}\rangle \\
\langle0_{0}, 2_{3}, 4_{6}, 3_{5}\rangle &
\langle0_{0}, 4_{3}, 2_{6}, 0_{7}\rangle &
\langle0_{1}, 5_{1}, 3_{4}, 5_{5}\rangle &
\langle0_{0}, 0_{2}, 2_{7}, 1_{7}\rangle &
\langle0_{2}, 4_{3}, 5_{7}, 1_{5}\rangle &
\langle0_{3}, 1_{3}, 0_{4}, 5_{7}\rangle &
\langle1_{1}, 1_{2}, 4_{5}, 0_{4}\rangle &
\langle1_{0}, 5_{1}, 0_{6}, 4_{6}\rangle \\
\langle0_{1}, 2_{1}, 4_{5}, 4_{7}\rangle &
\langle0_{3}, 2_{3}, 4_{4}, 5_{6}\rangle &
\langle0_{1}, 2_{2}, 2_{6}, 5_{7}\rangle &
\langle1_{0}, 0_{2}, 4_{7}, 0_{7}\rangle &
\langle0_{2}, 2_{2}, 5_{6}, 4_{6}\rangle &
\langle0_{2}, 3_{3}, 4_{4}, 5_{5}\rangle &
\langle0_{2}, 5_{2}, 0_{4}, 2_{4}\rangle &
\langle0_{0}, 1_{0}, 2_{5}, 5_{4}\rangle \\
\langle0_{0}, 1_{3}, 5_{5}, 4_{7}\rangle &
\langle0_{0}, 3_{1}, 0_{6}, 4_{5}\rangle &
\langle0_{1}, 3_{3}, 4_{6}, 3_{7}\rangle &
\end{array}$$
}

\noindent $n=29$:
{\scriptsize
$$\begin{array}{llllllllll}
\langle\infty, 4_{1}, 4_{7}, 0_{6}\rangle &
\langle\infty, 1_{0}, 0_{5}, 1_{4}\rangle &
\langle4_{1}, 4_{0}, 5_{4}, a_{1}\rangle &
\langle1_{3}, 0_{2}, 5_{6}, a_{1}\rangle &
\langle3_{3}, 5_{2}, 2_{4}, a_{2}\rangle &
\langle0_{0}, 2_{1}, 5_{7}, a_{2}\rangle &
\langle4_{1}, 1_{0}, 0_{4}, a_{3}\rangle &
\langle2_{3}, 3_{2}, 4_{5}, a_{3}\rangle \\
\langle2_{1}, 3_{0}, 3_{6}, a_{4}\rangle &
\langle0_{2}, 3_{3}, 5_{4}, a_{4}\rangle &
\langle4_{3}, 2_{2}, 4_{7}, a_{5}\rangle &
\langle1_{0}, 5_{1}, 5_{6}, a_{5}\rangle &
\langle3_{0}, 0_{3}, 0_{4}, 3_{7}\rangle &
\langle2_{3}, 3_{3}, 3_{6}, 3_{5}\rangle &
\langle4_{1}, 0_{2}, 0_{7}, 2_{4}\rangle &
\langle4_{0}, 3_{2}, 5_{5}, 1_{5}\rangle \\
\langle2_{3}, 4_{3}, 3_{7}, 1_{6}\rangle &
\langle1_{2}, 2_{2}, 4_{6}, 1_{5}\rangle &
\langle4_{0}, 0_{2}, 5_{7}, 1_{7}\rangle &
\langle3_{1}, 0_{2}, 0_{4}, 0_{6}\rangle &
\langle2_{1}, 3_{2}, 1_{7}, 0_{7}\rangle &
\langle5_{1}, 2_{3}, 0_{7}, 1_{5}\rangle &
\langle2_{1}, 0_{2}, 3_{5}, 1_{4}\rangle &
\langle5_{0}, 5_{3}, 1_{7}, 1_{6}\rangle \\
\langle5_{1}, 0_{1}, 5_{5}, 4_{6}\rangle &
\langle4_{0}, 0_{0}, 1_{6}, 0_{5}\rangle &
\langle2_{2}, 4_{0}, 0_{4}, 3_{6}\rangle &
\langle4_{1}, 3_{3}, 4_{4}, 1_{5}\rangle &
\langle0_{0}, 1_{3}, 4_{4}, 4_{5}\rangle &
\end{array}$$
}

\noindent $n=31$:
{\scriptsize
$$\begin{array}{llllllllll}
\langle\infty, 2_{1}, 2_{7}, 0_{6}\rangle &
\langle\infty, 5_{0}, 5_{4}, 3_{5}\rangle &
\langle1_{1}, 1_{0}, 5_{7}, a_{1}\rangle &
\langle1_{2}, 1_{3}, 2_{6}, a_{1}\rangle &
\langle0_{2}, 5_{3}, 3_{6}, a_{2}\rangle &
\langle5_{1}, 4_{0}, 1_{4}, a_{2}\rangle &
\langle1_{1}, 5_{0}, 1_{4}, a_{3}\rangle &
\langle0_{2}, 3_{3}, 1_{5}, a_{3}\rangle \\
\langle2_{1}, 5_{0}, 0_{5}, a_{4}\rangle &
\langle5_{3}, 3_{2}, 5_{6}, a_{4}\rangle &
\langle3_{0}, 1_{1}, 1_{6}, a_{5}\rangle &
\langle4_{3}, 0_{2}, 0_{4}, a_{5}\rangle &
\langle0_{2}, 1_{3}, 4_{6}, a_{6}\rangle &
\langle3_{0}, 2_{1}, 1_{4}, a_{6}\rangle &
\langle2_{3}, 4_{1}, 1_{6}, a_{7}\rangle &
\langle5_{0}, 1_{2}, 4_{7}, a_{7}\rangle \\
\langle1_{3}, 1_{1}, 0_{5}, 1_{5}\rangle &
\langle0_{2}, 4_{1}, 5_{5}, 0_{6}\rangle &
\langle2_{3}, 1_{3}, 2_{4}, 3_{5}\rangle &
\langle3_{1}, 2_{2}, 4_{7}, 5_{5}\rangle &
\langle1_{1}, 0_{3}, 0_{7}, 4_{7}\rangle &
\langle1_{2}, 0_{2}, 5_{4}, 1_{7}\rangle &
\langle1_{3}, 2_{0}, 2_{7}, 3_{7}\rangle &
\langle3_{1}, 0_{2}, 5_{7}, 1_{4}\rangle \\
\langle0_{0}, 3_{2}, 2_{6}, 3_{5}\rangle &
\langle2_{3}, 0_{3}, 5_{4}, 5_{7}\rangle &
\langle4_{0}, 2_{2}, 0_{7}, 0_{5}\rangle &
\langle3_{0}, 4_{0}, 3_{5}, 3_{6}\rangle &
\langle2_{1}, 3_{2}, 5_{5}, 5_{4}\rangle &
\langle0_{1}, 3_{3}, 1_{4}, 5_{6}\rangle &
\langle0_{0}, 2_{0}, 1_{4}, 3_{6}\rangle &
\end{array}$$
}

\noindent $n=33$:
{\scriptsize
$$\begin{array}{llllllllll}
\langle\infty, 1_{1}, 4_{6}, 1_{7}\rangle &
\langle\infty, 4_{0}, 4_{4}, 4_{5}\rangle &
\langle2_{1}, 2_{0}, 3_{4}, a_{1}\rangle &
\langle3_{3}, 3_{2}, 5_{5}, a_{1}\rangle &
\langle0_{3}, 5_{2}, 3_{5}, a_{2}\rangle &
\langle2_{1}, 1_{0}, 4_{4}, a_{2}\rangle &
\langle5_{2}, 1_{3}, 0_{5}, a_{3}\rangle &
\langle1_{1}, 5_{0}, 1_{4}, a_{3}\rangle \\
\langle2_{1}, 5_{0}, 0_{5}, a_{4}\rangle &
\langle5_{2}, 2_{3}, 5_{4}, a_{4}\rangle &
\langle0_{1}, 2_{0}, 4_{6}, a_{5}\rangle &
\langle0_{2}, 4_{3}, 4_{7}, a_{5}\rangle &
\langle5_{2}, 4_{3}, 2_{5}, a_{6}\rangle &
\langle4_{1}, 5_{0}, 3_{4}, a_{6}\rangle &
\langle1_{0}, 4_{2}, 4_{5}, a_{7}\rangle &
\langle0_{3}, 1_{1}, 3_{7}, a_{7}\rangle \\
\langle4_{0}, 4_{2}, 3_{5}, a_{8}\rangle &
\langle5_{1}, 1_{3}, 0_{7}, a_{8}\rangle &
\langle1_{1}, 5_{3}, 0_{5}, a_{9}\rangle &
\langle5_{0}, 4_{2}, 0_{7}, a_{9}\rangle &
\langle0_{3}, 1_{3}, 5_{4}, 1_{4}\rangle &
\langle2_{0}, 0_{2}, 5_{7}, 1_{7}\rangle &
\langle0_{1}, 4_{2}, 5_{6}, 4_{7}\rangle &
\langle0_{2}, 4_{2}, 3_{6}, 3_{4}\rangle \\
\langle3_{0}, 4_{2}, 4_{6}, 0_{6}\rangle &
\langle2_{1}, 5_{3}, 1_{7}, 5_{5}\rangle &
\langle0_{3}, 2_{3}, 5_{6}, 4_{6}\rangle &
\langle1_{0}, 5_{0}, 3_{7}, 3_{5}\rangle &
\langle2_{0}, 5_{3}, 1_{4}, 0_{6}\rangle &
\langle2_{0}, 1_{3}, 1_{6}, 2_{7}\rangle &
\langle3_{1}, 3_{2}, 0_{7}, 1_{4}\rangle &
\langle1_{1}, 0_{1}, 2_{6}, 1_{5}\rangle \\
\langle0_{1}, 2_{2}, 3_{4}, 0_{6}\rangle &
\end{array}$$
}

\noindent $n=35$:
{\scriptsize
$$\begin{array}{llllllllll}
\langle\infty, 3_{3}, 3_{6}, 4_{7}\rangle &
\langle\infty, 1_{0}, 5_{5}, 2_{4}\rangle &
\langle1_{3}, 0_{2}, 1_{7}, a_{1}\rangle &
\langle3_{1}, 5_{0}, 5_{5}, a_{1}\rangle &
\langle1_{1}, 0_{0}, 4_{4}, a_{2}\rangle &
\langle5_{3}, 2_{2}, 2_{5}, a_{2}\rangle &
\langle1_{2}, 3_{3}, 5_{4}, a_{3}\rangle &
\langle4_{0}, 0_{1}, 1_{6}, a_{3}\rangle \\
\langle2_{2}, 1_{1}, 5_{6}, a_{4}\rangle &
\langle5_{3}, 2_{0}, 2_{7}, a_{4}\rangle &
\langle3_{0}, 0_{1}, 5_{5}, a_{5}\rangle &
\langle1_{3}, 3_{2}, 5_{6}, a_{5}\rangle &
\langle0_{2}, 5_{3}, 1_{5}, a_{6}\rangle &
\langle4_{0}, 3_{1}, 0_{6}, a_{6}\rangle &
\langle0_{0}, 0_{2}, 4_{6}, a_{7}\rangle &
\langle2_{1}, 4_{3}, 2_{5}, a_{7}\rangle \\
\langle3_{3}, 3_{1}, 2_{4}, a_{8}\rangle &
\langle3_{2}, 4_{0}, 0_{7}, a_{8}\rangle &
\langle2_{3}, 5_{1}, 1_{7}, a_{9}\rangle &
\langle3_{2}, 2_{0}, 1_{5}, a_{9}\rangle &
\langle5_{2}, 1_{0}, 4_{7}, a_{10}\rangle &
\langle2_{3}, 3_{1}, 3_{6}, a_{10}\rangle &
\langle2_{2}, 2_{3}, 1_{6}, a_{11}\rangle &
\langle5_{0}, 5_{1}, 0_{7}, a_{11}\rangle \\
\langle4_{2}, 2_{0}, 5_{4}, 1_{4}\rangle &
\langle2_{2}, 5_{1}, 1_{4}, 4_{7}\rangle &
\langle0_{0}, 2_{0}, 1_{6}, 2_{4}\rangle &
\langle3_{2}, 4_{1}, 1_{7}, 2_{5}\rangle &
\langle3_{2}, 1_{1}, 3_{6}, 5_{4}\rangle &
\langle1_{3}, 0_{0}, 5_{7}, 1_{5}\rangle &
\langle5_{3}, 4_{3}, 1_{6}, 5_{4}\rangle &
\langle4_{2}, 4_{1}, 4_{4}, 4_{7}\rangle \\
\langle3_{0}, 5_{3}, 1_{7}, 0_{5}\rangle &
\langle5_{1}, 3_{2}, 4_{6}, 0_{5}\rangle &
\langle0_{1}, 4_{3}, 1_{4}, 3_{5}\rangle &
\end{array}$$
}

\noindent $n=37$:
{\scriptsize
$$\begin{array}{llllllllll}
\langle\infty, 0_{2}, 0_{9}, 1_{8}\rangle &
\langle\infty, 3_{1}, 4_{7}, 3_{6}\rangle &
\langle\infty, 4_{0}, 4_{5}, 2_{4}\rangle &
\langle0_{1}, 0_{0}, 1_{4}, a_{1}\rangle &
\langle0_{3}, 4_{2}, 3_{6}, a_{1}\rangle &
\langle3_{2}, 5_{2}, 4_{5}, 0_{4}\rangle &
\langle1_{2}, 4_{1}, 0_{4}, 0_{9}\rangle &
\langle0_{0}, 1_{0}, 0_{4}, 3_{9}\rangle \\
\langle1_{3}, 0_{1}, 1_{9}, 3_{5}\rangle &
\langle0_{1}, 5_{0}, 5_{7}, 5_{8}\rangle &
\langle0_{3}, 4_{0}, 2_{9}, 1_{8}\rangle &
\langle5_{2}, 5_{3}, 5_{8}, 3_{4}\rangle &
\langle2_{3}, 3_{1}, 0_{7}, 1_{7}\rangle &
\langle1_{1}, 3_{0}, 5_{6}, 1_{5}\rangle &
\langle0_{3}, 4_{3}, 2_{8}, 3_{4}\rangle &
\langle0_{2}, 5_{3}, 1_{6}, 2_{9}\rangle \\
\langle0_{0}, 3_{1}, 0_{6}, 5_{6}\rangle &
\langle2_{1}, 3_{1}, 0_{8}, 0_{4}\rangle &
\langle5_{1}, 3_{3}, 0_{8}, 5_{4}\rangle &
\langle0_{1}, 5_{2}, 0_{9}, 5_{4}\rangle &
\langle4_{2}, 3_{2}, 0_{8}, 1_{6}\rangle &
\langle5_{1}, 0_{2}, 1_{7}, 5_{8}\rangle &
\langle1_{1}, 4_{3}, 2_{6}, 3_{8}\rangle &
\langle4_{3}, 0_{2}, 4_{7}, 0_{7}\rangle \\
\langle1_{1}, 5_{0}, 0_{9}, 5_{9}\rangle &
\langle5_{3}, 4_{3}, 3_{9}, 5_{4}\rangle &
\langle4_{2}, 1_{0}, 2_{8}, 3_{7}\rangle &
\langle2_{3}, 5_{0}, 2_{5}, 2_{6}\rangle &
\langle2_{3}, 1_{0}, 5_{7}, 0_{5}\rangle &
\langle0_{3}, 5_{2}, 5_{6}, 1_{6}\rangle &
\langle3_{0}, 1_{0}, 5_{8}, 4_{7}\rangle &
\langle3_{0}, 5_{2}, 2_{9}, 2_{7}\rangle \\
\langle5_{0}, 5_{2}, 1_{4}, 1_{5}\rangle &
\langle0_{1}, 4_{2}, 0_{7}, 4_{5}\rangle &
\langle4_{0}, 3_{1}, 2_{6}, 5_{5}\rangle &
\langle4_{3}, 1_{2}, 5_{5}, 5_{9}\rangle &
\langle0_{1}, 2_{3}, 1_{5}, 5_{5}\rangle &
\end{array}$$
}

\noindent $n=39$:
{\scriptsize
$$\begin{array}{llllllllll}
\langle\infty, 1_{0}, 3_{9}, 0_{8}\rangle &
\langle\infty, 1_{1}, 3_{7}, 3_{6}\rangle &
\langle\infty, 3_{3}, 0_{4}, 4_{5}\rangle &
\langle0_{2}, 0_{3}, 0_{4}, a_{1}\rangle &
\langle4_{0}, 1_{1}, 5_{7}, a_{1}\rangle &
\langle4_{2}, 0_{3}, 5_{6}, a_{2}\rangle &
\langle4_{0}, 0_{1}, 1_{4}, a_{2}\rangle &
\langle1_{0}, 5_{1}, 3_{4}, a_{3}\rangle \\
\langle5_{2}, 3_{3}, 2_{7}, a_{3}\rangle &
\langle5_{3}, 3_{1}, 4_{8}, 0_{7}\rangle &
\langle5_{1}, 0_{1}, 2_{4}, 2_{8}\rangle &
\langle4_{3}, 3_{1}, 4_{9}, 0_{9}\rangle &
\langle0_{0}, 3_{2}, 0_{5}, 4_{5}\rangle &
\langle0_{1}, 2_{2}, 2_{5}, 4_{9}\rangle &
\langle4_{3}, 4_{0}, 5_{6}, 0_{6}\rangle &
\langle1_{1}, 4_{2}, 0_{8}, 0_{6}\rangle \\
\langle4_{1}, 5_{2}, 2_{6}, 3_{4}\rangle &
\langle1_{0}, 0_{3}, 4_{5}, 0_{6}\rangle &
\langle1_{0}, 2_{3}, 1_{9}, 1_{4}\rangle &
\langle5_{1}, 4_{3}, 0_{7}, 3_{5}\rangle &
\langle3_{2}, 0_{3}, 0_{8}, 4_{7}\rangle &
\langle0_{3}, 1_{3}, 2_{8}, 4_{9}\rangle &
\langle4_{2}, 4_{0}, 3_{4}, 5_{4}\rangle &
\langle3_{3}, 1_{0}, 0_{7}, 1_{6}\rangle \\
\langle0_{1}, 4_{1}, 0_{9}, 1_{5}\rangle &
\langle5_{0}, 4_{0}, 2_{6}, 2_{9}\rangle &
\langle5_{3}, 2_{1}, 2_{6}, 2_{5}\rangle &
\langle4_{2}, 5_{1}, 2_{6}, 5_{8}\rangle &
\langle4_{2}, 5_{2}, 4_{6}, 2_{9}\rangle &
\langle0_{0}, 5_{2}, 1_{5}, 3_{7}\rangle &
\langle5_{3}, 0_{2}, 0_{9}, 5_{5}\rangle &
\langle1_{1}, 1_{2}, 1_{7}, 0_{7}\rangle \\
\langle5_{0}, 1_{2}, 0_{9}, 3_{7}\rangle &
\langle1_{0}, 5_{0}, 1_{7}, 1_{8}\rangle &
\langle3_{0}, 3_{1}, 2_{5}, 1_{8}\rangle &
\langle2_{2}, 4_{1}, 4_{4}, 3_{9}\rangle &
\langle5_{0}, 3_{2}, 1_{5}, 2_{8}\rangle &
\langle3_{0}, 0_{3}, 1_{4}, 4_{8}\rangle &
\langle0_{2}, 1_{3}, 3_{4}, 4_{8}\rangle &
\end{array}$$
}

\noindent $n=41$:
{\scriptsize
$$\begin{array}{llllllllll}
\langle\infty, 3_{0}, 5_{5}, 1_{4}\rangle &
\langle\infty, 2_{3}, 5_{9}, 0_{8}\rangle &
\langle\infty, 1_{1}, 0_{6}, 3_{7}\rangle &
\langle4_{1}, 5_{0}, 4_{4}, a_{1}\rangle &
\langle5_{2}, 5_{3}, 2_{5}, a_{1}\rangle &
\langle5_{2}, 0_{3}, 2_{6}, a_{2}\rangle &
\langle4_{1}, 4_{0}, 5_{4}, a_{2}\rangle &
\langle1_{2}, 3_{3}, 2_{5}, a_{3}\rangle \\
\langle3_{1}, 1_{0}, 1_{7}, a_{3}\rangle &
\langle1_{2}, 4_{3}, 1_{4}, a_{4}\rangle &
\langle2_{1}, 5_{0}, 3_{5}, a_{4}\rangle &
\langle2_{1}, 4_{0}, 1_{5}, a_{5}\rangle &
\langle4_{2}, 2_{3}, 0_{4}, a_{5}\rangle &
\langle2_{1}, 3_{1}, 0_{8}, 5_{5}\rangle &
\langle1_{0}, 2_{1}, 4_{4}, 5_{9}\rangle &
\langle2_{1}, 1_{3}, 4_{8}, 2_{7}\rangle \\
\langle2_{2}, 2_{1}, 1_{4}, 0_{9}\rangle &
\langle0_{2}, 1_{2}, 0_{7}, 0_{6}\rangle &
\langle0_{2}, 5_{1}, 2_{7}, 5_{8}\rangle &
\langle0_{1}, 0_{3}, 5_{7}, 5_{9}\rangle &
\langle4_{2}, 3_{3}, 5_{7}, 4_{9}\rangle &
\langle2_{2}, 0_{0}, 1_{5}, 5_{8}\rangle &
\langle1_{1}, 3_{1}, 4_{6}, 1_{5}\rangle &
\langle5_{0}, 5_{3}, 5_{8}, 2_{6}\rangle \\
\langle4_{2}, 2_{2}, 2_{8}, 0_{5}\rangle &
\langle5_{0}, 4_{2}, 1_{7}, 2_{7}\rangle &
\langle2_{3}, 4_{1}, 5_{7}, 0_{9}\rangle &
\langle4_{0}, 2_{3}, 0_{6}, 2_{6}\rangle &
\langle3_{3}, 5_{3}, 5_{4}, 5_{5}\rangle &
\langle4_{0}, 4_{2}, 5_{9}, 4_{5}\rangle &
\langle4_{0}, 5_{0}, 2_{8}, 3_{7}\rangle &
\langle0_{2}, 5_{0}, 4_{6}, 1_{4}\rangle \\
\langle1_{3}, 0_{1}, 2_{6}, 1_{9}\rangle &
\langle0_{1}, 3_{2}, 0_{9}, 4_{6}\rangle &
\langle1_{3}, 4_{0}, 5_{7}, 3_{9}\rangle &
\langle1_{3}, 0_{0}, 2_{8}, 5_{5}\rangle &
\langle1_{3}, 4_{1}, 2_{4}, 3_{8}\rangle &
\langle4_{2}, 5_{1}, 2_{4}, 0_{8}\rangle &
\langle5_{0}, 1_{3}, 0_{6}, 0_{8}\rangle &
\langle1_{2}, 3_{0}, 3_{6}, 3_{9}\rangle \\
\langle0_{0}, 3_{2}, 0_{4}, 2_{9}\rangle &
\end{array}$$
}

\noindent $n=43$:
{\scriptsize
$$\begin{array}{llllllllll}
\langle\infty, 4_{2}, 3_{4}, 2_{8}\rangle &
\langle\infty, 1_{1}, 4_{9}, 3_{6}\rangle &
\langle\infty, 5_{3}, 0_{5}, 3_{7}\rangle &
\langle2_{2}, 3_{3}, 4_{6}, a_{1}\rangle &
\langle2_{0}, 3_{1}, 5_{5}, a_{1}\rangle &
\langle2_{0}, 5_{3}, 4_{4}, a_{2}\rangle &
\langle4_{1}, 1_{2}, 2_{5}, a_{2}\rangle &
\langle0_{3}, 5_{0}, 5_{7}, a_{3}\rangle \\
\langle4_{1}, 3_{2}, 1_{5}, a_{3}\rangle &
\langle5_{3}, 5_{2}, 0_{7}, a_{4}\rangle &
\langle2_{0}, 4_{1}, 2_{6}, a_{4}\rangle &
\langle3_{0}, 4_{2}, 4_{4}, a_{5}\rangle &
\langle3_{1}, 3_{3}, 0_{6}, a_{5}\rangle &
\langle4_{0}, 1_{1}, 0_{6}, a_{6}\rangle &
\langle4_{3}, 2_{2}, 0_{7}, a_{6}\rangle &
\langle0_{3}, 1_{0}, 4_{4}, a_{7}\rangle \\
\langle2_{1}, 2_{2}, 1_{7}, a_{7}\rangle &
\langle1_{3}, 3_{1}, 4_{4}, 2_{9}\rangle &
\langle3_{1}, 0_{3}, 0_{8}, 2_{4}\rangle &
\langle5_{2}, 0_{0}, 0_{4}, 5_{6}\rangle &
\langle3_{2}, 3_{0}, 0_{9}, 0_{8}\rangle &
\langle2_{1}, 3_{0}, 3_{8}, 2_{4}\rangle &
\langle3_{1}, 2_{3}, 2_{5}, 5_{8}\rangle &
\langle0_{1}, 5_{1}, 1_{7}, 1_{9}\rangle \\
\langle5_{3}, 0_{2}, 2_{5}, 1_{6}\rangle &
\langle1_{0}, 0_{0}, 4_{6}, 2_{7}\rangle &
\langle3_{2}, 1_{2}, 5_{4}, 2_{9}\rangle &
\langle0_{2}, 3_{3}, 3_{7}, 3_{4}\rangle &
\langle2_{0}, 4_{3}, 3_{6}, 0_{8}\rangle &
\langle2_{2}, 4_{0}, 5_{5}, 4_{9}\rangle &
\langle3_{2}, 4_{2}, 1_{6}, 5_{8}\rangle &
\langle2_{2}, 0_{3}, 2_{9}, 1_{8}\rangle \\
\langle0_{1}, 2_{3}, 1_{5}, 3_{4}\rangle &
\langle2_{2}, 5_{0}, 4_{7}, 0_{9}\rangle &
\langle0_{1}, 2_{1}, 0_{8}, 4_{4}\rangle &
\langle2_{0}, 0_{1}, 4_{9}, 0_{7}\rangle &
\langle0_{0}, 2_{2}, 2_{5}, 2_{8}\rangle &
\langle5_{3}, 4_{1}, 2_{7}, 5_{6}\rangle &
\langle1_{3}, 0_{3}, 4_{9}, 5_{8}\rangle &
\langle3_{1}, 3_{0}, 0_{7}, 2_{8}\rangle \\
\langle4_{0}, 2_{3}, 2_{9}, 4_{5}\rangle &
\langle0_{0}, 0_{3}, 4_{5}, 5_{9}\rangle &
\langle0_{1}, 1_{2}, 0_{5}, 0_{6}\rangle &
\end{array}$$
}

\noindent $n=45$:
{\scriptsize
$$\begin{array}{llllllllll}
\langle\infty, 0_{3}, 3_{7}, 2_{6}\rangle &
\langle\infty, 4_{0}, 5_{4}, 3_{5}\rangle &
\langle\infty, 0_{2}, 0_{9}, 3_{8}\rangle &
\langle0_{3}, 0_{2}, 2_{7}, a_{1}\rangle &
\langle5_{0}, 2_{1}, 1_{6}, a_{1}\rangle &
\langle4_{1}, 3_{0}, 0_{7}, a_{2}\rangle &
\langle3_{2}, 4_{3}, 4_{6}, a_{2}\rangle &
\langle2_{0}, 1_{1}, 1_{6}, a_{3}\rangle \\
\langle3_{3}, 0_{2}, 5_{4}, a_{3}\rangle &
\langle0_{2}, 5_{3}, 2_{5}, a_{4}\rangle &
\langle0_{1}, 2_{0}, 4_{4}, a_{4}\rangle &
\langle0_{3}, 4_{2}, 5_{7}, a_{5}\rangle &
\langle4_{0}, 0_{1}, 2_{5}, a_{5}\rangle &
\langle4_{0}, 2_{3}, 1_{6}, a_{6}\rangle &
\langle5_{2}, 4_{1}, 5_{4}, a_{6}\rangle &
\langle0_{2}, 3_{0}, 1_{4}, a_{7}\rangle \\
\langle3_{1}, 2_{3}, 1_{5}, a_{7}\rangle &
\langle5_{3}, 1_{1}, 1_{5}, a_{8}\rangle &
\langle0_{0}, 5_{2}, 5_{7}, a_{8}\rangle &
\langle0_{3}, 3_{1}, 4_{5}, a_{9}\rangle &
\langle3_{2}, 1_{0}, 0_{4}, a_{9}\rangle &
\langle2_{0}, 4_{3}, 4_{9}, 5_{5}\rangle &
\langle4_{1}, 5_{3}, 1_{8}, 0_{4}\rangle &
\langle2_{0}, 5_{3}, 0_{9}, 3_{7}\rangle \\
\langle2_{2}, 0_{1}, 5_{5}, 4_{8}\rangle &
\langle0_{2}, 4_{3}, 4_{5}, 4_{7}\rangle &
\langle1_{1}, 1_{2}, 0_{7}, 5_{9}\rangle &
\langle3_{1}, 2_{1}, 5_{6}, 3_{9}\rangle &
\langle3_{0}, 5_{0}, 1_{7}, 2_{9}\rangle &
\langle3_{0}, 3_{3}, 3_{4}, 0_{8}\rangle &
\langle0_{1}, 0_{0}, 1_{8}, 3_{4}\rangle &
\langle4_{0}, 5_{2}, 2_{6}, 4_{8}\rangle \\
\langle0_{2}, 5_{2}, 1_{9}, 4_{6}\rangle &
\langle4_{1}, 3_{2}, 5_{6}, 3_{8}\rangle &
\langle0_{1}, 4_{2}, 4_{6}, 3_{9}\rangle &
\langle1_{0}, 0_{3}, 5_{8}, 1_{7}\rangle &
\langle4_{3}, 4_{1}, 3_{4}, 0_{9}\rangle &
\langle4_{3}, 2_{3}, 2_{8}, 5_{6}\rangle &
\langle0_{0}, 1_{0}, 1_{6}, 1_{5}\rangle &
\langle2_{0}, 3_{3}, 4_{8}, 2_{9}\rangle \\
\langle5_{2}, 3_{2}, 1_{4}, 4_{5}\rangle &
\langle2_{3}, 0_{1}, 0_{4}, 5_{9}\rangle &
\langle2_{0}, 0_{2}, 1_{8}, 3_{9}\rangle &
\langle2_{1}, 4_{1}, 5_{7}, 4_{8}\rangle &
\langle0_{1}, 3_{2}, 3_{5}, 0_{7}\rangle &
\end{array}$$
}

\noindent $n=47$:
{\scriptsize
$$\begin{array}{llllllllll}
\langle\infty, 4_{3}, 5_{5}, 4_{4}\rangle &
\langle\infty, 3_{1}, 4_{7}, 4_{6}\rangle &
\langle\infty, 4_{0}, 3_{9}, 0_{8}\rangle &
\langle4_{0}, 4_{1}, 2_{7}, a_{1}\rangle &
\langle0_{2}, 3_{3}, 0_{5}, a_{1}\rangle &
\langle2_{3}, 1_{2}, 4_{6}, a_{2}\rangle &
\langle2_{0}, 3_{1}, 3_{4}, a_{2}\rangle &
\langle2_{3}, 0_{2}, 5_{7}, a_{3}\rangle \\
\langle4_{0}, 3_{1}, 5_{6}, a_{3}\rangle &
\langle4_{1}, 0_{2}, 2_{4}, a_{4}\rangle &
\langle5_{3}, 2_{0}, 5_{7}, a_{4}\rangle &
\langle2_{0}, 4_{1}, 5_{5}, a_{5}\rangle &
\langle5_{2}, 4_{3}, 0_{7}, a_{5}\rangle &
\langle4_{2}, 3_{0}, 2_{4}, a_{6}\rangle &
\langle5_{1}, 3_{3}, 5_{5}, a_{6}\rangle &
\langle0_{2}, 4_{0}, 1_{6}, a_{7}\rangle \\
\langle3_{3}, 0_{1}, 2_{4}, a_{7}\rangle &
\langle2_{0}, 2_{3}, 0_{4}, a_{8}\rangle &
\langle1_{1}, 0_{2}, 0_{7}, a_{8}\rangle &
\langle5_{0}, 4_{3}, 4_{5}, a_{9}\rangle &
\langle5_{1}, 3_{2}, 1_{7}, a_{9}\rangle &
\langle4_{3}, 3_{1}, 0_{4}, a_{10}\rangle &
\langle4_{2}, 4_{0}, 2_{5}, a_{10}\rangle &
\langle4_{0}, 3_{2}, 0_{7}, a_{11}\rangle \\
\langle1_{3}, 1_{1}, 2_{4}, a_{11}\rangle &
\langle5_{1}, 0_{1}, 5_{8}, 3_{6}\rangle &
\langle0_{0}, 4_{2}, 0_{9}, 5_{8}\rangle &
\langle5_{1}, 5_{2}, 4_{4}, 1_{8}\rangle &
\langle0_{2}, 4_{3}, 5_{8}, 2_{6}\rangle &
\langle4_{0}, 1_{2}, 5_{9}, 1_{4}\rangle &
\langle4_{1}, 0_{1}, 1_{9}, 3_{5}\rangle &
\langle1_{3}, 2_{1}, 2_{6}, 2_{9}\rangle \\
\langle1_{2}, 0_{2}, 0_{6}, 2_{5}\rangle &
\langle1_{0}, 4_{1}, 2_{5}, 5_{8}\rangle &
\langle1_{1}, 4_{2}, 5_{9}, 3_{9}\rangle &
\langle3_{0}, 1_{3}, 0_{9}, 5_{9}\rangle &
\langle5_{1}, 0_{2}, 4_{6}, 3_{8}\rangle &
\langle4_{0}, 0_{3}, 2_{9}, 0_{6}\rangle &
\langle2_{3}, 0_{1}, 5_{9}, 0_{7}\rangle &
\langle1_{0}, 5_{1}, 1_{5}, 2_{7}\rangle \\
\end{array}$$
}

\noindent $n=49$:
{\scriptsize
$$\begin{array}{llllllllll}
\langle\infty, 5_{2}, 1_{6}, 5_{7}\rangle &
\langle\infty, 3_{1}, 4_{10}, 5_{5}\rangle &
\langle\infty, 0_{0}, 3_{11}, 1_{8}\rangle &
\langle\infty, 4_{3}, 2_{9}, 2_{4}\rangle &
\langle5_{0}, 4_{1}, 5_{5}, a_{1}\rangle &
\langle1_{2}, 2_{3}, 2_{4}, a_{1}\rangle &
\langle3_{1}, 1_{3}, 2_{11}, 5_{8}\rangle &
\langle2_{0}, 3_{3}, 5_{10}, 4_{6}\rangle \\
\langle4_{0}, 1_{1}, 5_{9}, 0_{7}\rangle &
\langle0_{0}, 2_{3}, 1_{4}, 5_{10}\rangle &
\langle1_{2}, 2_{2}, 1_{9}, 3_{5}\rangle &
\langle1_{0}, 2_{1}, 2_{6}, 4_{9}\rangle &
\langle0_{3}, 2_{3}, 1_{10}, 5_{9}\rangle &
\langle2_{1}, 5_{3}, 0_{7}, 1_{4}\rangle &
\langle4_{3}, 3_{1}, 2_{6}, 0_{8}\rangle &
\langle3_{2}, 4_{1}, 4_{9}, 1_{6}\rangle \\
\langle0_{1}, 2_{2}, 1_{4}, 2_{11}\rangle &
\langle1_{1}, 4_{2}, 1_{7}, 5_{8}\rangle &
\langle3_{0}, 2_{3}, 1_{6}, 5_{5}\rangle &
\langle3_{1}, 3_{3}, 4_{8}, 2_{8}\rangle &
\langle0_{1}, 5_{1}, 2_{7}, 5_{10}\rangle &
\langle4_{0}, 5_{0}, 3_{4}, 4_{9}\rangle &
\langle1_{3}, 4_{0}, 3_{6}, 4_{8}\rangle &
\langle0_{1}, 5_{3}, 5_{9}, 5_{5}\rangle \\
\langle3_{0}, 0_{2}, 2_{8}, 2_{7}\rangle &
\langle3_{0}, 1_{1}, 5_{10}, 1_{8}\rangle &
\langle5_{2}, 3_{0}, 4_{5}, 5_{8}\rangle &
\langle5_{0}, 3_{3}, 5_{7}, 3_{11}\rangle &
\langle0_{0}, 0_{2}, 1_{7}, 3_{8}\rangle &
\langle2_{0}, 2_{1}, 3_{11}, 0_{5}\rangle &
\langle1_{2}, 0_{0}, 4_{10}, 2_{11}\rangle &
\langle0_{2}, 5_{3}, 2_{4}, 5_{8}\rangle \\
\langle1_{3}, 0_{3}, 5_{7}, 2_{9}\rangle &
\langle2_{2}, 4_{3}, 1_{11}, 5_{5}\rangle &
\langle1_{2}, 1_{3}, 5_{5}, 5_{11}\rangle &
\langle0_{0}, 2_{0}, 5_{5}, 4_{9}\rangle &
\langle0_{0}, 0_{3}, 5_{11}, 0_{10}\rangle &
\langle0_{1}, 2_{3}, 4_{11}, 2_{6}\rangle &
\langle2_{2}, 0_{3}, 2_{5}, 3_{6}\rangle &
\langle2_{2}, 4_{0}, 4_{11}, 1_{6}\rangle \\
\end{array}$$
}
Note that each of the codewords marked $s$ only generates three codewords.
\end{proof}

\begin{lem}
T$(2,29;2,n;6) =7n$ for each odd $n$ and $29 \leq n \leq 57$.
\end{lem}

\begin{proof} Let $X_1=(\bbZ_7\times \{0,1,2,3\})\cup \{\infty\}$. For $29 \leq n \leq 41$, let $X_2= (\bbZ_7\times \{4,5,6,7\})\cup (\{a\}\times \{1,\ldots,n-28\})$; for $43 \leq n \leq 55$, let $X_2= (\bbZ_7\times \{4,5,\ldots,9\})\cup (\{a\}\times \{1,\ldots,n-42\})$; for $n = 57$, let $X_2= (\bbZ_7\times \{4,5,\ldots,11\})\cup (\{a\}\times \{1\})$.
Denote $X=X_1\cup X_2$. The desired codes of size $7n$ are constructed on $\bbZ_2^X$ and the base codewords are listed as follows.

\noindent $n=29$:
{\scriptsize
$$
\begin{array}{llllllllll}
\langle\infty, 2_{0}, 2_{5}, 5_{4}\rangle &
\langle\infty, 0_{1}, 0_{7}, 0_{6}\rangle &
\langle6_{2}, 5_{3}, 3_{5}, a_{1}\rangle &
\langle0_{1}, 6_{0}, 1_{4}, a_{1}\rangle &
\langle2_{2}, 5_{3}, 2_{7}, 1_{6}\rangle &
\langle2_{2}, 2_{3}, 2_{5}, 2_{6}\rangle &
\langle1_{3}, 2_{1}, 0_{5}, 0_{4}\rangle &
\langle1_{0}, 0_{3}, 6_{7}, 5_{4}\rangle \\
\langle0_{0}, 6_{1}, 1_{6}, 4_{6}\rangle &
\langle1_{2}, 2_{2}, 1_{4}, 3_{5}\rangle &
\langle3_{2}, 1_{3}, 6_{7}, 2_{5}\rangle &
\langle1_{0}, 5_{1}, 2_{7}, 1_{7}\rangle &
\langle0_{2}, 4_{2}, 2_{4}, 5_{6}\rangle &
\langle0_{1}, 4_{1}, 6_{7}, 3_{6}\rangle &
\langle0_{1}, 5_{1}, 0_{5}, 6_{5}\rangle &
\langle1_{0}, 1_{2}, 5_{7}, 3_{7}\rangle \\
\langle2_{1}, 2_{2}, 6_{6}, 6_{4}\rangle &
\langle1_{0}, 0_{0}, 2_{5}, 6_{4}\rangle &
\langle3_{0}, 1_{1}, 2_{7}, 6_{7}\rangle &
\langle2_{2}, 3_{3}, 5_{5}, 0_{5}\rangle &
\langle1_{0}, 1_{3}, 4_{5}, 3_{6}\rangle &
\langle1_{3}, 5_{3}, 2_{6}, 1_{7}\rangle &
\langle1_{2}, 3_{2}, 2_{7}, 4_{4}\rangle &
\langle1_{1}, 0_{1}, 3_{4}, 0_{4}\rangle \\
\langle2_{0}, 4_{1}, 5_{6}, 0_{5}\rangle &
\langle0_{0}, 2_{0}, 0_{6}, 6_{5}\rangle &
\langle2_{2}, 6_{3}, 5_{6}, 0_{7}\rangle &
\langle1_{3}, 6_{3}, 2_{4}, 3_{4}\rangle &
\langle0_{0}, 1_{3}, 1_{4}, 6_{6}\rangle &
\end{array}$$
}

\noindent $n=31$:
{\scriptsize
$$\begin{array}{llllllllll}
\langle\infty, 5_{0}, 4_{5}, 5_{4}\rangle &
\langle\infty, 3_{1}, 0_{7}, 2_{6}\rangle &
\langle1_{3}, 1_{2}, 1_{7}, a_{1}\rangle &
\langle4_{1}, 2_{0}, 3_{4}, a_{1}\rangle &
\langle2_{0}, 3_{1}, 5_{4}, a_{2}\rangle &
\langle5_{3}, 4_{2}, 6_{6}, a_{2}\rangle &
\langle6_{3}, 4_{2}, 3_{7}, a_{3}\rangle &
\langle6_{1}, 2_{0}, 4_{4}, a_{3}\rangle \\
\langle5_{0}, 1_{1}, 6_{7}, 2_{7}\rangle &
\langle6_{0}, 4_{1}, 6_{7}, 2_{6}\rangle &
\langle5_{2}, 2_{3}, 0_{5}, 0_{7}\rangle &
\langle2_{2}, 2_{0}, 0_{7}, 5_{7}\rangle &
\langle4_{2}, 2_{3}, 5_{7}, 0_{4}\rangle &
\langle4_{0}, 1_{0}, 1_{6}, 2_{6}\rangle &
\langle4_{1}, 5_{0}, 6_{5}, 2_{5}\rangle &
\langle6_{1}, 3_{3}, 1_{6}, 3_{5}\rangle \\
\langle6_{0}, 4_{3}, 5_{4}, 1_{6}\rangle &
\langle1_{2}, 3_{2}, 2_{4}, 3_{4}\rangle &
\langle3_{1}, 6_{2}, 3_{7}, 3_{4}\rangle &
\langle4_{2}, 0_{3}, 2_{4}, 4_{5}\rangle &
\langle2_{3}, 1_{3}, 1_{4}, 5_{4}\rangle &
\langle2_{2}, 6_{2}, 3_{5}, 5_{6}\rangle &
\langle6_{0}, 0_{3}, 6_{5}, 1_{7}\rangle &
\langle3_{3}, 0_{3}, 2_{7}, 6_{6}\rangle \\
\langle3_{1}, 5_{1}, 6_{5}, 6_{6}\rangle &
\langle3_{3}, 5_{3}, 6_{5}, 5_{6}\rangle &
\langle3_{0}, 4_{0}, 1_{4}, 6_{5}\rangle &
\langle5_{0}, 0_{2}, 3_{5}, 4_{6}\rangle &
\langle1_{2}, 2_{1}, 2_{6}, 6_{6}\rangle &
\langle2_{1}, 5_{1}, 6_{4}, 1_{7}\rangle &
\langle0_{1}, 1_{2}, 0_{5}, 6_{5}\rangle &
\end{array}$$
}

\noindent $n=33$:
{\scriptsize
$$\begin{array}{llllllllll}
\langle\infty, 3_{0}, 3_{4}, 3_{5}\rangle &
\langle\infty, 1_{1}, 6_{6}, 4_{7}\rangle &
\langle5_{2}, 4_{3}, 4_{5}, a_{1}\rangle &
\langle3_{0}, 3_{1}, 4_{4}, a_{1}\rangle &
\langle6_{2}, 0_{3}, 0_{7}, a_{2}\rangle &
\langle6_{0}, 0_{1}, 2_{4}, a_{2}\rangle &
\langle2_{3}, 0_{2}, 5_{6}, a_{3}\rangle &
\langle5_{0}, 3_{1}, 3_{7}, a_{3}\rangle \\
\langle4_{2}, 0_{3}, 3_{7}, a_{4}\rangle &
\langle4_{0}, 3_{1}, 6_{4}, a_{4}\rangle &
\langle5_{0}, 1_{1}, 5_{7}, a_{5}\rangle &
\langle4_{2}, 1_{3}, 4_{4}, a_{5}\rangle &
\langle1_{1}, 2_{3}, 6_{7}, 0_{5}\rangle &
\langle1_{1}, 0_{1}, 2_{5}, 5_{4}\rangle &
\langle2_{1}, 0_{1}, 5_{5}, 1_{6}\rangle &
\langle0_{3}, 6_{3}, 1_{4}, 4_{4}\rangle \\
\langle1_{0}, 1_{3}, 0_{5}, 5_{5}\rangle &
\langle1_{1}, 3_{2}, 5_{5}, 4_{6}\rangle &
\langle1_{0}, 0_{3}, 2_{7}, 2_{5}\rangle &
\langle4_{0}, 1_{2}, 2_{5}, 0_{6}\rangle &
\langle1_{1}, 4_{3}, 1_{6}, 5_{6}\rangle &
\langle2_{1}, 2_{2}, 1_{4}, 4_{6}\rangle &
\langle3_{0}, 6_{2}, 2_{7}, 6_{7}\rangle &
\langle3_{2}, 2_{0}, 1_{4}, 0_{6}\rangle \\
\langle1_{0}, 2_{0}, 6_{4}, 4_{5}\rangle &
\langle0_{2}, 4_{2}, 0_{6}, 0_{5}\rangle &
\langle2_{0}, 2_{2}, 4_{7}, 6_{7}\rangle &
\langle1_{2}, 2_{2}, 5_{4}, 3_{4}\rangle &
\langle2_{2}, 6_{1}, 6_{5}, 0_{7}\rangle &
\langle0_{0}, 3_{0}, 4_{6}, 2_{6}\rangle &
\langle0_{3}, 4_{3}, 5_{7}, 6_{6}\rangle &
\langle1_{3}, 6_{3}, 2_{5}, 6_{6}\rangle \\
\langle0_{1}, 0_{3}, 0_{4}, 6_{7}\rangle &
\end{array}$$
}

\noindent $n=35$:
{\scriptsize
$$\begin{array}{llllllllll}
\langle\infty, 1_{2}, 0_{5}, 1_{4}\rangle &
\langle\infty, 1_{1}, 1_{6}, 1_{7}\rangle &
\langle6_{3}, 2_{2}, 6_{7}, a_{1}\rangle &
\langle6_{0}, 6_{1}, 0_{4}, a_{1}\rangle &
\langle5_{2}, 4_{3}, 6_{5}, a_{2}\rangle &
\langle5_{0}, 6_{1}, 1_{4}, a_{2}\rangle &
\langle5_{1}, 1_{3}, 5_{4}, a_{3}\rangle &
\langle0_{0}, 0_{2}, 2_{6}, a_{3}\rangle \\
\langle3_{3}, 3_{2}, 2_{4}, a_{4}\rangle &
\langle6_{0}, 2_{1}, 4_{5}, a_{4}\rangle &
\langle5_{2}, 1_{3}, 4_{7}, a_{5}\rangle &
\langle3_{0}, 5_{1}, 4_{5}, a_{5}\rangle &
\langle5_{0}, 3_{1}, 2_{4}, a_{6}\rangle &
\langle5_{3}, 4_{2}, 5_{6}, a_{6}\rangle &
\langle1_{0}, 0_{1}, 3_{7}, a_{7}\rangle &
\langle6_{2}, 4_{3}, 6_{6}, a_{7}\rangle \\
\langle1_{2}, 3_{3}, 6_{6}, 4_{5}\rangle &
\langle2_{1}, 5_{1}, 3_{7}, 3_{5}\rangle &
\langle3_{3}, 4_{1}, 0_{5}, 0_{6}\rangle &
\langle0_{3}, 2_{0}, 1_{7}, 6_{5}\rangle &
\langle3_{0}, 6_{3}, 3_{7}, 1_{7}\rangle &
\langle3_{2}, 5_{2}, 1_{7}, 6_{4}\rangle &
\langle0_{1}, 0_{3}, 6_{7}, 3_{4}\rangle &
\langle5_{3}, 1_{1}, 6_{6}, 3_{7}\rangle \\
\langle4_{3}, 5_{0}, 4_{5}, 2_{6}\rangle &
\langle4_{1}, 3_{1}, 1_{4}, 5_{6}\rangle &
\langle0_{2}, 5_{0}, 0_{5}, 3_{6}\rangle &
\langle2_{2}, 6_{0}, 6_{4}, 6_{5}\rangle &
\langle2_{3}, 5_{0}, 1_{6}, 3_{4}\rangle &
\langle4_{1}, 6_{2}, 1_{7}, 4_{5}\rangle &
\langle2_{3}, 4_{3}, 0_{5}, 4_{4}\rangle &
\langle2_{0}, 3_{0}, 3_{6}, 6_{7}\rangle \\
\langle5_{1}, 5_{2}, 2_{6}, 4_{6}\rangle &
\langle3_{2}, 2_{0}, 5_{5}, 3_{7}\rangle &
\langle0_{0}, 4_{2}, 2_{4}, 6_{4}\rangle &
\end{array}$$
}

\noindent $n=37$:
{\scriptsize
$$\begin{array}{llllllllll}
\langle\infty, 2_{1}, 2_{7}, 2_{6}\rangle &
\langle\infty, 1_{0}, 1_{5}, 1_{4}\rangle &
\langle5_{2}, 5_{3}, 2_{5}, a_{1}\rangle &
\langle5_{0}, 5_{1}, 6_{4}, a_{1}\rangle &
\langle6_{0}, 0_{1}, 2_{4}, a_{2}\rangle &
\langle4_{2}, 5_{3}, 5_{5}, a_{2}\rangle &
\langle1_{1}, 6_{0}, 1_{4}, a_{3}\rangle &
\langle5_{2}, 0_{3}, 1_{5}, a_{3}\rangle \\
\langle0_{1}, 4_{0}, 3_{4}, a_{4}\rangle &
\langle4_{2}, 0_{3}, 3_{5}, a_{4}\rangle &
\langle2_{2}, 6_{3}, 4_{6}, a_{5}\rangle &
\langle2_{0}, 6_{1}, 3_{5}, a_{5}\rangle &
\langle6_{2}, 4_{3}, 2_{7}, a_{6}\rangle &
\langle2_{1}, 4_{0}, 1_{4}, a_{6}\rangle &
\langle4_{0}, 3_{1}, 6_{5}, a_{7}\rangle &
\langle5_{2}, 4_{3}, 2_{4}, a_{7}\rangle \\
\langle2_{0}, 4_{2}, 0_{4}, a_{8}\rangle &
\langle2_{1}, 2_{3}, 4_{5}, a_{8}\rangle &
\langle1_{1}, 4_{3}, 4_{7}, a_{9}\rangle &
\langle3_{0}, 3_{2}, 3_{6}, a_{9}\rangle &
\langle0_{1}, 4_{2}, 2_{6}, 1_{7}\rangle &
\langle1_{0}, 0_{0}, 4_{5}, 2_{6}\rangle &
\langle2_{0}, 3_{2}, 4_{7}, 1_{5}\rangle &
\langle0_{2}, 6_{2}, 0_{4}, 5_{4}\rangle \\
\langle2_{1}, 4_{2}, 0_{6}, 6_{4}\rangle &
\langle1_{3}, 0_{3}, 6_{5}, 0_{4}\rangle &
\langle1_{0}, 5_{0}, 5_{7}, 6_{7}\rangle &
\langle0_{2}, 3_{2}, 4_{6}, 2_{7}\rangle &
\langle3_{1}, 4_{3}, 0_{7}, 1_{4}\rangle &
\langle1_{1}, 3_{3}, 5_{6}, 2_{6}\rangle &
\langle0_{1}, 2_{1}, 0_{5}, 1_{5}\rangle &
\langle0_{1}, 4_{1}, 6_{7}, 3_{6}\rangle \\
\langle0_{3}, 2_{3}, 3_{6}, 3_{4}\rangle &
\langle0_{0}, 2_{0}, 6_{6}, 5_{6}\rangle &
\langle0_{3}, 4_{3}, 4_{6}, 1_{7}\rangle &
\langle2_{0}, 6_{3}, 1_{7}, 5_{7}\rangle &
\langle0_{2}, 2_{2}, 2_{5}, 0_{7}\rangle &
\end{array}$$
}

\noindent $n=39$:
{\scriptsize
$$\begin{array}{llllllllll}
\langle\infty, 6_{1}, 6_{7}, 6_{6}\rangle &
\langle\infty, 1_{0}, 1_{5}, 2_{4}\rangle &
\langle4_{1}, 4_{0}, 5_{5}, a_{1}\rangle &
\langle4_{2}, 4_{3}, 1_{6}, a_{1}\rangle &
\langle5_{0}, 6_{1}, 1_{4}, a_{2}\rangle &
\langle3_{2}, 4_{3}, 0_{5}, a_{2}\rangle &
\langle5_{2}, 4_{3}, 6_{6}, a_{3}\rangle &
\langle1_{0}, 0_{1}, 3_{5}, a_{3}\rangle \\
\langle2_{2}, 5_{3}, 0_{7}, a_{4}\rangle &
\langle4_{1}, 1_{0}, 0_{4}, a_{4}\rangle &
\langle2_{3}, 5_{2}, 5_{4}, a_{5}\rangle &
\langle4_{1}, 0_{0}, 5_{6}, a_{5}\rangle &
\langle2_{3}, 4_{2}, 4_{5}, a_{6}\rangle &
\langle1_{0}, 6_{1}, 5_{4}, a_{6}\rangle &
\langle4_{1}, 4_{2}, 5_{4}, a_{7}\rangle &
\langle0_{0}, 2_{3}, 3_{5}, a_{7}\rangle \\
\langle2_{3}, 2_{1}, 6_{4}, a_{8}\rangle &
\langle6_{0}, 3_{2}, 1_{6}, a_{8}\rangle &
\langle6_{0}, 0_{2}, 6_{6}, a_{9}\rangle &
\langle6_{1}, 4_{3}, 4_{4}, a_{9}\rangle &
\langle0_{2}, 2_{0}, 1_{5}, a_{10}\rangle &
\langle6_{1}, 0_{3}, 5_{7}, a_{10}\rangle &
\langle4_{3}, 1_{1}, 3_{5}, a_{11}\rangle &
\langle2_{2}, 2_{0}, 0_{4}, a_{11}\rangle \\
\langle4_{0}, 2_{0}, 3_{6}, 2_{7}\rangle &
\langle2_{0}, 2_{3}, 5_{7}, 0_{5}\rangle &
\langle2_{3}, 4_{3}, 3_{7}, 2_{6}\rangle &
\langle2_{3}, 6_{3}, 6_{5}, 6_{7}\rangle &
\langle2_{3}, 1_{0}, 1_{4}, 5_{6}\rangle &
\langle3_{0}, 0_{0}, 4_{7}, 2_{7}\rangle &
\langle1_{1}, 0_{1}, 5_{5}, 0_{5}\rangle &
\langle2_{1}, 0_{1}, 4_{6}, 5_{6}\rangle \\
\langle2_{2}, 1_{1}, 1_{4}, 3_{7}\rangle &
\langle1_{1}, 6_{2}, 5_{7}, 6_{7}\rangle &
\langle4_{1}, 3_{2}, 0_{7}, 3_{6}\rangle &
\langle1_{2}, 4_{2}, 6_{5}, 0_{5}\rangle &
\langle0_{0}, 4_{3}, 2_{4}, 3_{6}\rangle &
\langle1_{2}, 2_{2}, 4_{4}, 4_{7}\rangle &
\langle0_{2}, 2_{3}, 4_{4}, 3_{6}\rangle &
\end{array}$$
}

\noindent $n=41$:
{\scriptsize
$$\begin{array}{llllllllll}
\langle\infty, 2_{1}, 1_{5}, 1_{4}\rangle &
\langle\infty, 4_{2}, 1_{7}, 4_{6}\rangle &
\langle2_{1}, 2_{0}, 6_{4}, a_{1}\rangle &
\langle2_{2}, 2_{3}, 5_{6}, a_{1}\rangle &
\langle1_{2}, 2_{3}, 6_{6}, a_{2}\rangle &
\langle2_{1}, 1_{0}, 4_{4}, a_{2}\rangle &
\langle0_{3}, 5_{2}, 4_{5}, a_{3}\rangle &
\langle2_{0}, 4_{1}, 4_{4}, a_{3}\rangle \\
\langle0_{3}, 1_{2}, 1_{4}, a_{4}\rangle &
\langle1_{0}, 4_{1}, 1_{6}, a_{4}\rangle &
\langle4_{2}, 0_{3}, 0_{7}, a_{5}\rangle &
\langle1_{1}, 4_{0}, 4_{5}, a_{5}\rangle &
\langle3_{3}, 6_{2}, 1_{5}, a_{6}\rangle &
\langle0_{1}, 2_{0}, 3_{4}, a_{6}\rangle &
\langle6_{0}, 5_{1}, 6_{4}, a_{7}\rangle &
\langle1_{3}, 3_{2}, 3_{7}, a_{7}\rangle \\
\langle0_{1}, 5_{2}, 5_{5}, a_{8}\rangle &
\langle1_{3}, 3_{0}, 2_{6}, a_{8}\rangle &
\langle4_{3}, 2_{1}, 0_{4}, a_{9}\rangle &
\langle1_{2}, 4_{0}, 2_{5}, a_{9}\rangle &
\langle2_{3}, 3_{1}, 4_{6}, a_{10}\rangle &
\langle2_{2}, 2_{0}, 0_{4}, a_{10}\rangle &
\langle0_{3}, 2_{1}, 1_{7}, a_{11}\rangle &
\langle5_{0}, 4_{2}, 1_{6}, a_{11}\rangle \\
\langle5_{3}, 5_{1}, 0_{5}, a_{12}\rangle &
\langle4_{2}, 2_{0}, 3_{6}, a_{12}\rangle &
\langle6_{2}, 1_{0}, 0_{4}, a_{13}\rangle &
\langle3_{3}, 2_{1}, 0_{7}, a_{13}\rangle &
\langle3_{0}, 6_{0}, 5_{7}, 3_{7}\rangle &
\langle3_{0}, 5_{0}, 2_{5}, 1_{7}\rangle &
\langle2_{3}, 4_{3}, 1_{4}, 3_{5}\rangle &
\langle6_{1}, 3_{1}, 3_{5}, 6_{6}\rangle \\
\langle6_{1}, 3_{3}, 6_{7}, 1_{6}\rangle &
\langle6_{3}, 3_{3}, 1_{4}, 6_{5}\rangle &
\langle0_{0}, 1_{0}, 5_{6}, 3_{5}\rangle &
\langle3_{2}, 5_{2}, 2_{4}, 4_{7}\rangle &
\langle3_{3}, 2_{3}, 1_{7}, 2_{6}\rangle &
\langle2_{1}, 0_{1}, 4_{7}, 3_{7}\rangle &
\langle5_{2}, 2_{0}, 3_{5}, 3_{7}\rangle &
\langle3_{2}, 6_{1}, 5_{6}, 4_{6}\rangle \\
\langle0_{2}, 1_{2}, 3_{4}, 4_{5}\rangle &
\end{array}$$
}

\noindent $n=43$:
{\scriptsize
$$\begin{array}{llllllllll}
\langle\infty, 2_{3}, 5_{8}, 5_{9}\rangle &
\langle\infty, 0_{2}, 3_{6}, 2_{4}\rangle &
\langle\infty, 4_{0}, 2_{5}, 5_{7}\rangle &
\langle5_{3}, 4_{2}, 4_{5}, a_{1}\rangle &
\langle5_{0}, 5_{1}, 4_{6}, a_{1}\rangle &
\langle3_{1}, 5_{1}, 0_{6}, 5_{5}\rangle &
\langle2_{0}, 6_{0}, 4_{9}, 5_{9}\rangle &
\langle0_{0}, 4_{3}, 2_{6}, 4_{8}\rangle \\
\langle4_{2}, 4_{0}, 5_{6}, 4_{6}\rangle &
\langle5_{0}, 3_{2}, 3_{4}, 4_{7}\rangle &
\langle2_{2}, 5_{1}, 2_{8}, 4_{7}\rangle &
\langle6_{1}, 6_{3}, 4_{9}, 3_{7}\rangle &
\langle2_{3}, 0_{0}, 3_{4}, 2_{7}\rangle &
\langle4_{2}, 1_{1}, 5_{9}, 2_{4}\rangle &
\langle1_{3}, 5_{3}, 5_{9}, 2_{6}\rangle &
\langle0_{3}, 4_{0}, 3_{4}, 4_{4}\rangle \\
\langle3_{2}, 5_{2}, 4_{5}, 1_{8}\rangle &
\langle3_{3}, 0_{2}, 4_{7}, 4_{9}\rangle &
\langle2_{3}, 6_{1}, 2_{6}, 4_{6}\rangle &
\langle3_{1}, 0_{1}, 1_{7}, 6_{9}\rangle &
\langle3_{2}, 4_{2}, 6_{9}, 1_{6}\rangle &
\langle0_{0}, 1_{1}, 4_{7}, 1_{8}\rangle &
\langle5_{0}, 4_{2}, 3_{7}, 5_{8}\rangle &
\langle6_{1}, 0_{3}, 6_{6}, 0_{5}\rangle \\
\langle6_{0}, 0_{2}, 3_{4}, 0_{9}\rangle &
\langle0_{1}, 4_{3}, 2_{7}, 2_{8}\rangle &
\langle6_{0}, 5_{3}, 6_{5}, 0_{4}\rangle &
\langle2_{0}, 0_{1}, 2_{9}, 6_{5}\rangle &
\langle3_{0}, 5_{1}, 2_{5}, 5_{4}\rangle &
\langle0_{2}, 4_{3}, 2_{5}, 4_{4}\rangle &
\langle1_{1}, 2_{1}, 0_{4}, 5_{4}\rangle &
\langle2_{2}, 2_{3}, 4_{8}, 1_{4}\rangle \\
\langle2_{0}, 0_{0}, 5_{8}, 3_{5}\rangle &
\langle3_{1}, 4_{2}, 3_{9}, 5_{4}\rangle &
\langle0_{0}, 0_{3}, 3_{6}, 6_{8}\rangle &
\langle2_{2}, 2_{1}, 0_{5}, 5_{5}\rangle &
\langle5_{0}, 2_{1}, 0_{8}, 3_{6}\rangle &
\langle3_{3}, 2_{0}, 2_{7}, 5_{7}\rangle &
\langle6_{1}, 5_{3}, 2_{8}, 0_{9}\rangle &
\langle1_{3}, 3_{0}, 0_{9}, 5_{5}\rangle \\
\langle6_{2}, 2_{2}, 1_{6}, 2_{7}\rangle &
\langle6_{2}, 4_{1}, 3_{8}, 5_{8}\rangle &
\langle0_{2}, 2_{3}, 4_{5}, 5_{7}\rangle &
\end{array}$$
}

\noindent $n=45$:
{\scriptsize
$$\begin{array}{llllllllll}
\langle\infty, 2_{2}, 2_{9}, 2_{8}\rangle &
\langle\infty, 2_{0}, 2_{4}, 2_{5}\rangle &
\langle\infty, 2_{1}, 4_{6}, 3_{7}\rangle &
\langle5_{0}, 5_{1}, 6_{4}, a_{1}\rangle &
\langle1_{2}, 1_{3}, 1_{5}, a_{1}\rangle &
\langle3_{1}, 2_{0}, 5_{4}, a_{2}\rangle &
\langle4_{2}, 5_{3}, 2_{6}, a_{2}\rangle &
\langle4_{0}, 6_{1}, 6_{4}, a_{3}\rangle \\
\langle3_{2}, 5_{3}, 4_{5}, a_{3}\rangle &
\langle0_{0}, 5_{0}, 2_{9}, 3_{5}\rangle &
\langle0_{1}, 4_{1}, 5_{9}, 2_{5}\rangle &
\langle2_{1}, 5_{2}, 2_{9}, 2_{5}\rangle &
\langle3_{1}, 4_{3}, 0_{7}, 0_{9}\rangle &
\langle1_{0}, 2_{2}, 2_{6}, 5_{8}\rangle &
\langle1_{2}, 4_{3}, 5_{4}, 6_{7}\rangle &
\langle0_{1}, 3_{3}, 4_{5}, 6_{8}\rangle \\
\langle1_{0}, 4_{1}, 3_{6}, 0_{4}\rangle &
\langle2_{1}, 6_{3}, 6_{6}, 6_{8}\rangle &
\langle1_{0}, 4_{2}, 3_{5}, 3_{7}\rangle &
\langle0_{3}, 5_{3}, 6_{7}, 1_{6}\rangle &
\langle1_{0}, 0_{3}, 2_{9}, 5_{7}\rangle &
\langle3_{0}, 5_{2}, 2_{7}, 3_{8}\rangle &
\langle2_{0}, 0_{1}, 5_{6}, 2_{9}\rangle &
\langle1_{3}, 4_{3}, 2_{8}, 3_{6}\rangle \\
\langle2_{1}, 2_{3}, 2_{7}, 1_{9}\rangle &
\langle3_{0}, 3_{2}, 1_{9}, 6_{7}\rangle &
\langle0_{1}, 5_{2}, 5_{7}, 1_{6}\rangle &
\langle2_{1}, 1_{1}, 6_{4}, 4_{8}\rangle &
\langle1_{2}, 6_{2}, 4_{5}, 0_{8}\rangle &
\langle3_{1}, 0_{2}, 4_{8}, 2_{7}\rangle &
\langle1_{0}, 0_{1}, 2_{7}, 0_{6}\rangle &
\langle0_{2}, 6_{2}, 1_{9}, 1_{4}\rangle \\
\langle0_{0}, 1_{3}, 6_{9}, 4_{5}\rangle &
\langle0_{3}, 6_{3}, 2_{4}, 5_{4}\rangle &
\langle0_{2}, 4_{2}, 1_{6}, 6_{6}\rangle &
\langle2_{0}, 2_{3}, 4_{8}, 0_{6}\rangle &
\langle1_{1}, 1_{2}, 4_{9}, 0_{4}\rangle &
\langle3_{0}, 0_{1}, 3_{6}, 3_{7}\rangle &
\langle1_{2}, 6_{3}, 3_{8}, 6_{4}\rangle &
\langle0_{0}, 6_{0}, 4_{4}, 5_{8}\rangle \\
\langle4_{2}, 3_{3}, 3_{9}, 0_{4}\rangle &
\langle0_{1}, 5_{1}, 6_{5}, 5_{8}\rangle &
\langle3_{2}, 0_{3}, 5_{5}, 4_{7}\rangle &
\langle2_{0}, 4_{3}, 3_{8}, 5_{9}\rangle &
\langle0_{0}, 4_{3}, 1_{5}, 6_{5}\rangle &
\end{array}$$
}

\noindent $n=47$:
{\scriptsize
$$\begin{array}{llllllllll}
\langle\infty, 3_{2}, 3_{8}, 6_{9}\rangle &
\langle\infty, 3_{1}, 5_{7}, 1_{6}\rangle &
\langle\infty, 3_{0}, 3_{5}, 0_{4}\rangle &
\langle6_{1}, 2_{0}, 6_{5}, a_{1}\rangle &
\langle1_{2}, 4_{3}, 3_{7}, a_{1}\rangle &
\langle5_{3}, 0_{2}, 2_{5}, a_{2}\rangle &
\langle5_{1}, 4_{0}, 0_{4}, a_{2}\rangle &
\langle2_{3}, 0_{2}, 1_{5}, a_{3}\rangle \\
\langle5_{1}, 3_{0}, 6_{7}, a_{3}\rangle &
\langle2_{2}, 3_{3}, 5_{4}, a_{4}\rangle &
\langle6_{1}, 3_{0}, 0_{6}, a_{4}\rangle &
\langle0_{1}, 2_{0}, 3_{5}, a_{5}\rangle &
\langle0_{3}, 0_{2}, 1_{6}, a_{5}\rangle &
\langle2_{1}, 0_{1}, 2_{8}, 6_{5}\rangle &
\langle1_{2}, 0_{0}, 6_{7}, 0_{9}\rangle &
\langle4_{1}, 4_{3}, 2_{7}, 1_{9}\rangle \\
\langle6_{0}, 2_{0}, 0_{6}, 1_{9}\rangle &
\langle4_{1}, 2_{2}, 2_{4}, 3_{4}\rangle &
\langle4_{0}, 3_{3}, 5_{7}, 0_{8}\rangle &
\langle0_{2}, 1_{2}, 1_{7}, 1_{9}\rangle &
\langle3_{3}, 1_{3}, 6_{9}, 4_{7}\rangle &
\langle1_{2}, 3_{0}, 0_{7}, 4_{8}\rangle &
\langle4_{0}, 5_{0}, 3_{4}, 4_{8}\rangle &
\langle3_{1}, 3_{2}, 0_{4}, 6_{7}\rangle \\
\langle4_{2}, 1_{3}, 6_{6}, 2_{5}\rangle &
\langle6_{1}, 3_{1}, 1_{9}, 4_{8}\rangle &
\langle0_{3}, 4_{3}, 6_{6}, 6_{8}\rangle &
\langle2_{2}, 5_{0}, 0_{4}, 6_{9}\rangle &
\langle1_{2}, 2_{0}, 1_{5}, 0_{8}\rangle &
\langle4_{1}, 5_{2}, 3_{9}, 1_{6}\rangle &
\langle2_{3}, 4_{0}, 6_{6}, 2_{7}\rangle &
\langle1_{1}, 5_{3}, 4_{9}, 5_{8}\rangle \\
\langle1_{3}, 2_{1}, 5_{4}, 4_{5}\rangle &
\langle4_{3}, 6_{1}, 4_{5}, 6_{9}\rangle &
\langle5_{1}, 1_{3}, 5_{7}, 4_{8}\rangle &
\langle1_{1}, 2_{0}, 1_{6}, 4_{8}\rangle &
\langle1_{3}, 4_{0}, 1_{9}, 4_{4}\rangle &
\langle6_{2}, 2_{2}, 0_{8}, 5_{5}\rangle &
\langle2_{0}, 4_{2}, 6_{8}, 2_{6}\rangle &
\langle4_{3}, 5_{2}, 4_{4}, 2_{8}\rangle \\
\end{array}$$
}

\noindent $n=49$:
{\scriptsize
$$\begin{array}{llllllllll}
\langle\infty, 1_{1}, 0_{4}, 1_{7}\rangle &
\langle\infty, 3_{0}, 2_{5}, 5_{8}\rangle &
\langle\infty, 3_{3}, 2_{6}, 4_{9}\rangle &
\langle0_{0}, 4_{2}, 0_{4}, a_{1}\rangle &
\langle4_{1}, 5_{3}, 5_{5}, a_{1}\rangle &
\langle1_{3}, 4_{2}, 0_{7}, a_{2}\rangle &
\langle3_{0}, 5_{1}, 4_{6}, a_{2}\rangle &
\langle4_{2}, 6_{0}, 1_{7}, a_{3}\rangle \\
\langle3_{1}, 0_{3}, 5_{4}, a_{3}\rangle &
\langle4_{1}, 3_{0}, 3_{5}, a_{4}\rangle &
\langle0_{3}, 0_{2}, 4_{6}, a_{4}\rangle &
\langle4_{0}, 2_{3}, 5_{4}, a_{5}\rangle &
\langle1_{1}, 6_{2}, 6_{7}, a_{5}\rangle &
\langle5_{3}, 5_{0}, 5_{7}, a_{6}\rangle &
\langle0_{1}, 1_{2}, 3_{5}, a_{6}\rangle &
\langle1_{3}, 6_{2}, 3_{4}, a_{7}\rangle \\
\langle0_{1}, 2_{0}, 5_{5}, a_{7}\rangle &
\langle5_{2}, 3_{3}, 5_{5}, 5_{9}\rangle &
\langle5_{0}, 4_{2}, 2_{7}, 2_{9}\rangle &
\langle6_{0}, 0_{2}, 5_{6}, 2_{8}\rangle &
\langle5_{0}, 6_{3}, 1_{6}, 4_{7}\rangle &
\langle1_{2}, 4_{3}, 0_{6}, 4_{6}\rangle &
\langle2_{1}, 2_{3}, 4_{8}, 5_{9}\rangle &
\langle6_{1}, 2_{1}, 1_{7}, 6_{8}\rangle \\
\langle0_{1}, 3_{2}, 4_{7}, 0_{5}\rangle &
\langle3_{0}, 6_{2}, 3_{8}, 1_{6}\rangle &
\langle0_{3}, 1_{0}, 6_{5}, 6_{8}\rangle &
\langle5_{1}, 3_{3}, 0_{5}, 6_{7}\rangle &
\langle1_{1}, 2_{1}, 6_{4}, 0_{9}\rangle &
\langle2_{2}, 3_{3}, 3_{9}, 0_{4}\rangle &
\langle4_{0}, 0_{3}, 1_{5}, 6_{4}\rangle &
\langle5_{2}, 5_{0}, 1_{9}, 6_{5}\rangle \\
\langle2_{2}, 5_{2}, 4_{9}, 1_{5}\rangle &
\langle1_{0}, 6_{0}, 0_{9}, 0_{8}\rangle &
\langle2_{2}, 0_{2}, 2_{4}, 5_{8}\rangle &
\langle1_{3}, 3_{3}, 5_{7}, 6_{8}\rangle &
\langle5_{1}, 1_{0}, 3_{6}, 6_{4}\rangle &
\langle2_{1}, 5_{3}, 2_{9}, 3_{6}\rangle &
\langle1_{1}, 0_{2}, 1_{4}, 5_{5}\rangle &
\langle5_{1}, 4_{3}, 5_{6}, 1_{8}\rangle \\
\end{array}$$
}

\noindent $n=51$:
{\scriptsize
$$\begin{array}{llllllllll}
\langle\infty, 5_{2}, 6_{5}, 0_{8}\rangle &
\langle\infty, 2_{0}, 6_{6}, 6_{4}\rangle &
\langle\infty, 2_{3}, 4_{7}, 1_{9}\rangle &
\langle5_{0}, 5_{1}, 4_{7}, a_{1}\rangle &
\langle2_{3}, 0_{2}, 4_{5}, a_{1}\rangle &
\langle2_{0}, 5_{3}, 2_{4}, a_{2}\rangle &
\langle1_{2}, 6_{1}, 6_{6}, a_{2}\rangle &
\langle2_{0}, 3_{2}, 0_{7}, a_{3}\rangle \\
\langle4_{1}, 4_{3}, 0_{6}, a_{3}\rangle &
\langle3_{1}, 6_{0}, 4_{4}, a_{4}\rangle &
\langle3_{2}, 6_{3}, 6_{7}, a_{4}\rangle &
\langle0_{2}, 6_{1}, 1_{6}, a_{5}\rangle &
\langle2_{0}, 2_{3}, 4_{4}, a_{5}\rangle &
\langle5_{0}, 4_{2}, 0_{6}, a_{6}\rangle &
\langle3_{1}, 6_{3}, 4_{5}, a_{6}\rangle &
\langle2_{3}, 0_{0}, 6_{5}, a_{7}\rangle \\
\langle0_{1}, 4_{2}, 3_{4}, a_{7}\rangle &
\langle5_{1}, 0_{0}, 0_{7}, a_{8}\rangle &
\langle6_{3}, 2_{2}, 0_{5}, a_{8}\rangle &
\langle2_{3}, 1_{1}, 5_{7}, a_{9}\rangle &
\langle4_{0}, 1_{2}, 5_{6}, a_{9}\rangle &
\langle2_{0}, 4_{1}, 0_{9}, 4_{7}\rangle &
\langle2_{2}, 2_{3}, 1_{7}, 3_{9}\rangle &
\langle1_{3}, 2_{3}, 6_{7}, 4_{9}\rangle \\
\langle1_{2}, 1_{1}, 6_{9}, 4_{5}\rangle &
\langle4_{2}, 3_{2}, 2_{8}, 3_{5}\rangle &
\langle0_{0}, 1_{3}, 1_{9}, 3_{6}\rangle &
\langle3_{2}, 0_{2}, 5_{4}, 1_{7}\rangle &
\langle3_{3}, 6_{1}, 3_{6}, 3_{5}\rangle &
\langle0_{2}, 0_{0}, 6_{6}, 3_{8}\rangle &
\langle0_{0}, 4_{0}, 6_{8}, 2_{5}\rangle &
\langle3_{2}, 1_{0}, 4_{4}, 0_{4}\rangle \\
\langle1_{2}, 6_{3}, 3_{9}, 4_{4}\rangle &
\langle5_{3}, 0_{3}, 5_{8}, 1_{8}\rangle &
\langle6_{2}, 1_{2}, 5_{9}, 1_{6}\rangle &
\langle1_{0}, 3_{0}, 1_{6}, 3_{9}\rangle &
\langle1_{3}, 6_{1}, 4_{5}, 5_{6}\rangle &
\langle1_{0}, 2_{0}, 2_{5}, 5_{9}\rangle &
\langle0_{1}, 6_{3}, 4_{9}, 0_{4}\rangle &
\langle4_{2}, 5_{3}, 4_{8}, 6_{7}\rangle \\
\end{array}$$
}

\noindent $n=53$:
{\scriptsize
$$\begin{array}{llllllllll}
\langle\infty, 1_{1}, 1_{7}, 4_{6}\rangle &
\langle\infty, 0_{3}, 1_{9}, 0_{8}\rangle &
\langle\infty, 0_{0}, 0_{4}, 3_{5}\rangle &
\langle3_{2}, 3_{3}, 3_{5}, a_{1}\rangle &
\langle2_{1}, 4_{0}, 0_{6}, a_{1}\rangle &
\langle6_{2}, 0_{3}, 3_{5}, a_{2}\rangle &
\langle5_{1}, 4_{0}, 0_{4}, a_{2}\rangle &
\langle6_{2}, 1_{3}, 2_{6}, a_{3}\rangle \\
\langle4_{1}, 5_{0}, 5_{5}, a_{3}\rangle &
\langle4_{1}, 1_{0}, 5_{4}, a_{4}\rangle &
\langle6_{3}, 1_{2}, 1_{6}, a_{4}\rangle &
\langle1_{3}, 5_{2}, 5_{4}, a_{5}\rangle &
\langle0_{0}, 4_{1}, 1_{6}, a_{5}\rangle &
\langle3_{0}, 2_{2}, 0_{7}, a_{6}\rangle &
\langle2_{1}, 0_{3}, 1_{5}, a_{6}\rangle &
\langle5_{0}, 5_{1}, 0_{5}, a_{7}\rangle \\
\langle0_{3}, 1_{2}, 2_{7}, a_{7}\rangle &
\langle1_{2}, 1_{0}, 6_{4}, a_{8}\rangle &
\langle4_{3}, 1_{1}, 2_{7}, a_{8}\rangle &
\langle5_{3}, 6_{1}, 1_{6}, a_{9}\rangle &
\langle3_{0}, 1_{2}, 4_{7}, a_{9}\rangle &
\langle2_{2}, 1_{0}, 6_{7}, a_{10}\rangle &
\langle5_{3}, 3_{1}, 3_{4}, a_{10}\rangle &
\langle1_{0}, 4_{2}, 3_{6}, a_{11}\rangle \\
\langle6_{1}, 6_{3}, 5_{7}, a_{11}\rangle &
\langle0_{0}, 6_{3}, 3_{9}, 0_{8}\rangle &
\langle2_{3}, 5_{1}, 0_{9}, 3_{4}\rangle &
\langle0_{0}, 2_{2}, 6_{4}, 0_{6}\rangle &
\langle2_{0}, 6_{0}, 1_{7}, 0_{9}\rangle &
\langle1_{1}, 6_{2}, 1_{9}, 5_{4}\rangle &
\langle3_{2}, 1_{1}, 1_{5}, 5_{7}\rangle &
\langle1_{0}, 1_{3}, 0_{8}, 6_{6}\rangle \\
\langle6_{3}, 3_{3}, 3_{7}, 3_{6}\rangle &
\langle0_{1}, 1_{1}, 0_{6}, 4_{9}\rangle &
\langle4_{2}, 1_{1}, 2_{8}, 3_{8}\rangle &
\langle2_{3}, 0_{3}, 6_{5}, 2_{9}\rangle &
\langle2_{3}, 3_{3}, 2_{4}, 6_{8}\rangle &
\langle3_{2}, 6_{0}, 2_{7}, 4_{5}\rangle &
\langle0_{3}, 2_{0}, 6_{6}, 6_{9}\rangle &
\langle1_{2}, 5_{2}, 1_{9}, 2_{6}\rangle \\
\end{array}$$
}

\noindent $n=55$:
{\scriptsize
$$\begin{array}{llllllllll}
\langle\infty, 1_{3}, 0_{9}, 2_{6}\rangle &
\langle\infty, 0_{2}, 1_{5}, 4_{7}\rangle &
\langle\infty, 3_{0}, 3_{8}, 0_{4}\rangle &
\langle4_{1}, 6_{3}, 2_{7}, a_{1}\rangle &
\langle3_{0}, 0_{2}, 6_{6}, a_{1}\rangle &
\langle1_{0}, 2_{2}, 1_{4}, a_{2}\rangle &
\langle4_{1}, 5_{3}, 2_{5}, a_{2}\rangle &
\langle5_{1}, 1_{2}, 3_{4}, a_{3}\rangle \\
\langle4_{3}, 2_{0}, 3_{6}, a_{3}\rangle &
\langle4_{2}, 0_{3}, 5_{6}, a_{4}\rangle &
\langle0_{0}, 2_{1}, 3_{7}, a_{4}\rangle &
\langle1_{3}, 1_{2}, 1_{4}, a_{5}\rangle &
\langle0_{0}, 3_{1}, 4_{5}, a_{5}\rangle &
\langle2_{3}, 1_{0}, 4_{5}, a_{6}\rangle &
\langle1_{2}, 2_{1}, 5_{6}, a_{6}\rangle &
\langle3_{1}, 3_{2}, 5_{7}, a_{7}\rangle \\
\langle0_{0}, 3_{3}, 6_{4}, a_{7}\rangle &
\langle3_{3}, 0_{1}, 0_{7}, a_{8}\rangle &
\langle3_{0}, 2_{2}, 2_{5}, a_{8}\rangle &
\langle2_{3}, 4_{2}, 2_{7}, a_{9}\rangle &
\langle3_{1}, 5_{0}, 6_{4}, a_{9}\rangle &
\langle2_{1}, 0_{2}, 2_{5}, a_{10}\rangle &
\langle1_{0}, 1_{3}, 3_{7}, a_{10}\rangle &
\langle6_{2}, 4_{0}, 2_{6}, a_{11}\rangle \\
\langle1_{3}, 1_{1}, 0_{4}, a_{11}\rangle &
\langle2_{1}, 6_{2}, 6_{6}, a_{12}\rangle &
\langle5_{0}, 4_{3}, 3_{5}, a_{12}\rangle &
\langle1_{3}, 4_{0}, 3_{6}, a_{13}\rangle &
\langle0_{1}, 2_{2}, 3_{7}, a_{13}\rangle &
\langle2_{0}, 1_{1}, 3_{8}, 0_{7}\rangle &
\langle1_{0}, 6_{3}, 0_{7}, 4_{8}\rangle &
\langle4_{2}, 1_{3}, 0_{7}, 2_{4}\rangle \\
\langle0_{2}, 6_{1}, 5_{5}, 5_{9}\rangle &
\langle5_{2}, 1_{2}, 2_{8}, 0_{8}\rangle &
\langle6_{2}, 5_{3}, 2_{9}, 5_{5}\rangle &
\langle2_{2}, 4_{3}, 2_{7}, 6_{9}\rangle &
\langle2_{3}, 6_{3}, 3_{8}, 0_{9}\rangle &
\langle2_{1}, 0_{3}, 0_{6}, 3_{8}\rangle &
\langle1_{3}, 2_{1}, 4_{6}, 1_{8}\rangle &
\langle6_{1}, 5_{0}, 4_{8}, 5_{6}\rangle \\
\end{array}$$
}

\noindent $n=57$:
{\scriptsize
$$\begin{array}{llllllllll}
\langle\infty, 2_{0}, 6_{11}, 2_{5}\rangle &
\langle\infty, 0_{2}, 2_{7}, 1_{4}\rangle &
\langle\infty, 5_{1}, 2_{10}, 0_{9}\rangle &
\langle\infty, 4_{3}, 5_{6}, 2_{8}\rangle &
\langle1_{3}, 6_{2}, 3_{4}, a_{1}\rangle &
\langle0_{1}, 2_{0}, 6_{7}, a_{1}\rangle &
\langle0_{2}, 4_{2}, 6_{6}, 3_{11}\rangle &
\langle0_{1}, 6_{2}, 1_{5}, 5_{9}\rangle \\
\langle3_{0}, 4_{0}, 2_{10}, 5_{9}\rangle &
\langle0_{1}, 3_{3}, 3_{5}, 6_{9}\rangle &
\langle1_{2}, 1_{3}, 2_{9}, 1_{7}\rangle &
\langle0_{1}, 2_{2}, 0_{5}, 2_{5}\rangle &
\langle2_{2}, 4_{2}, 5_{7}, 0_{10}\rangle &
\langle2_{3}, 1_{2}, 2_{10}, 4_{9}\rangle &
\langle0_{3}, 1_{3}, 6_{5}, 3_{5}\rangle &
\langle1_{3}, 4_{2}, 3_{8}, 1_{9}\rangle \\
\langle1_{0}, 0_{1}, 4_{4}, 2_{10}\rangle &
\langle4_{0}, 6_{1}, 1_{4}, 3_{9}\rangle &
\langle4_{2}, 5_{0}, 1_{8}, 3_{5}\rangle &
\langle6_{0}, 1_{0}, 6_{9}, 4_{6}\rangle &
\langle1_{0}, 3_{3}, 3_{8}, 0_{6}\rangle &
\langle6_{0}, 2_{3}, 0_{4}, 1_{10}\rangle &
\langle1_{2}, 6_{3}, 0_{10}, 5_{11}\rangle &
\langle0_{3}, 3_{1}, 0_{11}, 4_{4}\rangle \\
\langle1_{1}, 2_{2}, 1_{7}, 4_{9}\rangle &
\langle2_{0}, 3_{1}, 0_{8}, 1_{5}\rangle &
\langle5_{3}, 0_{1}, 6_{8}, 3_{7}\rangle &
\langle4_{2}, 0_{3}, 5_{6}, 4_{9}\rangle &
\langle4_{2}, 3_{2}, 0_{5}, 5_{11}\rangle &
\langle2_{1}, 2_{3}, 6_{7}, 6_{5}\rangle &
\langle3_{1}, 0_{2}, 3_{4}, 5_{11}\rangle &
\langle4_{0}, 0_{1}, 6_{5}, 1_{6}\rangle \\
\langle6_{0}, 6_{3}, 3_{10}, 5_{8}\rangle &
\langle3_{3}, 2_{0}, 5_{10}, 5_{11}\rangle &
\langle2_{3}, 5_{3}, 4_{6}, 1_{7}\rangle &
\langle0_{1}, 6_{1}, 6_{10}, 1_{7}\rangle &
\langle0_{3}, 2_{3}, 3_{11}, 3_{4}\rangle &
\langle1_{2}, 3_{1}, 3_{8}, 6_{4}\rangle &
\langle4_{1}, 6_{1}, 3_{6}, 2_{10}\rangle &
\langle3_{0}, 0_{2}, 0_{8}, 5_{6}\rangle \\
\end{array}$$
}
\end{proof}

\begin{lem}
T$(2,33;2,n;6) =8n$ for each odd $n$ and $33 \leq n \leq 65$.
\end{lem}

\begin{proof} Let $X_1=(\bbZ_8\times \{0,1,2,3\})\cup \{\infty\}$. For $33 \leq n \leq 47$, let $X_2= (\bbZ_8\times \{4,5,6,7\})\cup (\{a\}\times \{1,\ldots,n-32\})$; for $49 \leq n \leq 63$, let $X_2= (\bbZ_8\times \{4,5,\ldots,9\})\cup (\{a\}\times \{1,\ldots,n-48\})$; for $n = 65$, let $X_2= (\bbZ_8\times \{4,5,\ldots,11\})\cup (\{a\}\times \{1\})$.
Denote $X=X_1\cup X_2$. The desired codes of size $8n$ are constructed on $\bbZ_2^X$ and the base codewords are listed as follows.

\noindent $n=33$:
{\scriptsize
$$
\begin{array}{llllllllll}
\langle\infty, 7_{1}, 2_{6}, 4_{7}\rangle &
\langle\infty, 0_{0}, 5_{5}, 0_{4}\rangle &
\langle1_{0}, 6_{1}, 3_{5}, a_{1}\rangle &
\langle4_{2}, 6_{3}, 3_{4}, a_{1}\rangle &
\langle3_{3}, 1_{0}, 2_{4}, 1_{7}\rangle &
\langle4_{2}, 0_{0}, 7_{4}, 0_{5}\rangle &
\langle3_{3}, 6_{3}, 4_{6}, 3_{7}\rangle &
\langle7_{1}, 6_{2}, 6_{6}, 5_{6}\rangle \\
\langle1_{2}, 7_{0}, 4_{7}, 2_{4}\rangle &
\langle0_{0}, 6_{0}, 4_{6}, 4_{7}\rangle &
\langle1_{3}, 5_{1}, 4_{5}, 5_{7}\rangle &
\langle4_{2}, 5_{2}, 6_{7}, 2_{5}\rangle &
\langle5_{3}, 3_{3}, 2_{5}, 6_{4}\rangle &
\langle3_{1}, 2_{1}, 3_{5}, 4_{4}\rangle &
\langle7_{0}, 2_{2}, 6_{6}, 1_{7}\rangle &
\langle7_{1}, 1_{1}, 5_{7}, 7_{4}\rangle \\
\langle3_{3}, 0_{2}, 7_{5}, 2_{6}\rangle &
\langle2_{0}, 3_{1}, 4_{6}, 5_{7}\rangle &
\langle4_{2}, 2_{3}, 2_{6}, 4_{5}\rangle &
\langle6_{1}, 1_{1}, 6_{6}, 4_{5}\rangle &
\langle6_{2}, 6_{1}, 2_{4}, 0_{5}\rangle &
\langle4_{0}, 7_{0}, 3_{5}, 5_{5}\rangle &
\langle4_{3}, 2_{1}, 6_{6}, 3_{7}\rangle &
\langle2_{2}, 7_{1}, 6_{7}, 3_{5}\rangle \\
\langle7_{0}, 6_{1}, 5_{4}, 0_{6}\rangle &
\langle7_{0}, 2_{3}, 7_{6}, 2_{5}\rangle &
\langle3_{0}, 3_{2}, 0_{4}, 0_{6}\rangle &
\langle0_{3}, 3_{2}, 3_{7}, 6_{5}\rangle &
\langle3_{3}, 7_{2}, 4_{7}, 7_{4}\rangle &
\langle0_{3}, 1_{0}, 2_{7}, 4_{6}\rangle &
\langle7_{2}, 7_{3}, 2_{6}, 1_{4}\rangle &
\langle0_{0}, 4_{3}, 2_{4}, 4_{4}\rangle \\
\langle0_{1}, 5_{2}, 3_{4}, 3_{7}\rangle &
\end{array}$$
}

\noindent $n=35$:
{\scriptsize
$$\begin{array}{llllllllll}
\langle\infty, 4_{1}, 4_{6}, 4_{7}\rangle &
\langle\infty, 7_{0}, 7_{4}, 7_{5}\rangle &
\langle4_{3}, 2_{2}, 7_{5}, a_{1}\rangle &
\langle7_{1}, 1_{0}, 0_{4}, a_{1}\rangle &
\langle6_{2}, 4_{3}, 1_{5}, a_{2}\rangle &
\langle4_{0}, 5_{1}, 7_{4}, a_{2}\rangle &
\langle5_{2}, 5_{3}, 0_{7}, a_{3}\rangle &
\langle2_{1}, 5_{0}, 2_{4}, a_{3}\rangle \\
\langle3_{0}, 7_{3}, 2_{6}, 7_{6}\rangle &
\langle6_{2}, 5_{3}, 4_{7}, 5_{5}\rangle &
\langle4_{1}, 3_{2}, 7_{6}, 5_{6}\rangle &
\langle5_{1}, 7_{3}, 5_{5}, 1_{7}\rangle &
\langle1_{3}, 0_{3}, 6_{7}, 0_{4}\rangle &
\langle4_{1}, 7_{3}, 3_{7}, 1_{6}\rangle &
\langle7_{0}, 7_{3}, 4_{6}, 5_{6}\rangle &
\langle4_{2}, 0_{3}, 2_{5}, 1_{6}\rangle \\
\langle3_{3}, 6_{3}, 7_{4}, 2_{6}\rangle &
\langle5_{0}, 0_{2}, 5_{7}, 7_{6}\rangle &
\langle5_{1}, 6_{3}, 6_{7}, 3_{4}\rangle &
\langle1_{2}, 0_{2}, 1_{5}, 1_{6}\rangle &
\langle2_{1}, 1_{3}, 5_{5}, 7_{4}\rangle &
\langle1_{0}, 7_{0}, 3_{5}, 0_{5}\rangle &
\langle2_{2}, 0_{2}, 1_{4}, 2_{4}\rangle &
\langle2_{1}, 3_{1}, 0_{5}, 6_{4}\rangle \\
\langle7_{0}, 1_{1}, 3_{7}, 6_{7}\rangle &
\langle1_{2}, 4_{2}, 7_{6}, 0_{7}\rangle &
\langle3_{2}, 6_{3}, 5_{5}, 1_{4}\rangle &
\langle4_{0}, 1_{2}, 1_{7}, 2_{7}\rangle &
\langle4_{3}, 7_{1}, 6_{4}, 5_{7}\rangle &
\langle4_{0}, 7_{1}, 1_{5}, 5_{6}\rangle &
\langle3_{0}, 4_{2}, 7_{4}, 1_{4}\rangle &
\langle5_{1}, 0_{1}, 1_{5}, 7_{6}\rangle \\
\langle4_{0}, 5_{0}, 7_{7}, 6_{4}\rangle &
\langle1_{0}, 0_{1}, 7_{5}, 4_{6}\rangle &
\langle0_{0}, 7_{2}, 3_{5}, 1_{7}\rangle &
\end{array}%
$$
}
Note that each of the codewords marked $s$ only generates four codewords.
\end{proof}

\begin{lem}
T$(2,37;2,n;6) =9n$ for each odd $n$ and $37 \leq n \leq 73$.
\end{lem}

\begin{proof} Let $X_1=(\bbZ_9\times \{0,1,2,3\})\cup \{\infty\}$. For $37 \leq n \leq 53$, let $X_2= (\bbZ_9\times \{4,5,6,7\})\cup (\{a\}\times \{1,\ldots,n-36\})$; for $55 \leq n \leq 71$, let $X_2= (\bbZ_9\times \{4,5,\ldots,9\})\cup (\{a\}\times \{1,\ldots,n-54\})$; for $n = 73$, let $X_2= (\bbZ_9\times \{4,5,\ldots,11\})\cup (\{a\}\times \{1\})$.
Denote $X=X_1\cup X_2$. The desired codes of size $9n$ are constructed on $\bbZ_2^X$ and the base codewords are listed as follows.

\noindent $n=37$:
{\scriptsize
$$
$$
}
\end{proof}

\subsection{SFSs of type $(4,2)^a(2,4)^b$ with $a+b\in\{5,6,7,8,9\}$}

To save space, we only list the non-empty cells of the SFSs. We use $(a,b;i,j)$ to denote a cell which is indexed by $(i,j)$ and contains a pair $\{a,b\}$. We use $I_n$ to denote the set $\{0,1,2,\ldots,n-1\}$.

\begin{lem}
There exists an SFS of type $(4,2)^a(2,2)^{5-a}$ for each $a\in\{0,1,2,3,4,5\}$.
\end{lem}

\begin{proof}
Let $V=I_{10}$ and $S=I_{10+2a}$. $V$ can be partitioned as $V=\cup_{i=0}^4\{2i,2i+1\}$ and $S$ can  be partitioned as  $S=(\cup_{i=0}^{a-1}\{4i,4i+1,4i+2,4i+3\})\cup (\cup_{i=a}^{4}\{2i,2i+1\})$. The required SFSs are presented as follows.

\noindent $a=0$:
{\scriptsize
$$
$$
}
\end{proof}

\begin{lem}
There exists an SFS of type $(4,2)^a(2,2)^{6-a}$ for each $a\in\{0,1,\ldots,6\}$.
\end{lem}

\begin{proof}
Let $V=I_{12}$ and $S=I_{12+2a}$. $V$ can be partitioned as $V=\cup_{i=0}^5\{2i,2i+1\}$ and $S$ can  be partitioned as  $S=(\cup_{i=0}^{a-1}\{4i,4i+1,4i+2,4i+3\})\cup (\cup_{i=a}^{5}\{2i,2i+1\})$. The required SFSs are presented as follows.

\noindent $a=0$:
{\scriptsize
$$%
$$
}
\end{proof}

\begin{lem}
There exists an SFS of type $(4,2)^a(2,2)^{7-a}$ for each $a\in\{0,1,\ldots,7\}$.
\end{lem}

\begin{proof}
Let $V=I_{14}$ and $S=I_{14+2a}$. $V$ can be partitioned as $V=\cup_{i=0}^6\{2i,2i+1\}$ and $S$ can  be partitioned as  $S=(\cup_{i=0}^{a-1}\{4i,4i+1,4i+2,4i+3\})\cup (\cup_{i=a}^{6}\{2i,2i+1\})$. The required SFSs are presented as follows.

\noindent $a=0$:
{\scriptsize
$$%
$$
}
\end{proof}

\begin{lem}
There exists an SFS of type $(4,2)^a(2,2)^{8-a}$ for each $a\in\{0,1,\ldots,8\}$.
\end{lem}

\begin{proof}
Let $V=I_{16}$ and $S=I_{16+2a}$. $V$ can be partitioned as $V=\cup_{i=0}^7\{2i,2i+1\}$ and $S$ can  be partitioned as  $S=(\cup_{i=0}^{a-1}\{4i,4i+1,4i+2,4i+3\})\cup (\cup_{i=a}^{7}\{2i,2i+1\})$. The required SFSs are presented as follows.

\noindent $a=0$:
{\scriptsize
$$%
$$
}
\end{proof}

\subsection{HSAS$(s,v;3,3)$ and HSAS$(s,v;5,3)$ with $v\in\{11,15,19\}$}

\begin{lem}
There exists an HSAS$(s,11;3,3)$ for each $s\in \{11,13,15,17,19\}$.
\end{lem}

\begin{proof} Let $V=I_{11}$ and $S=I_s$. Let $W=\{8,9,10\}$ and $T=\{s-3,s-2,s-1\}$. The desired HSASs filled with pairs of points from  $V$  and indexed by $S$ are presented as follows.

\noindent $s=11$:
{\scriptsize
$$%
$$
}
\end{proof}

\begin{lem}
There exists an HSAS$(s,15;3,3)$ for each $s\in \{15,17,\ldots,27\}$.
\end{lem}

\begin{proof} Let $V=I_{15}$ and $S=I_s$. Let $W=\{12,13,14\}$ and $T=\{s-3,s-2,s-1\}$. The desired HSASs filled with pairs of points from  $V$  and indexed by $S$ are presented as follows.

\noindent $s=15$:
{\scriptsize
$$%
$$
}
\end{proof}

\begin{lem}
There exists an HSAS$(s,19;3,3)$ for each $s\in \{19,21,\ldots,35\}$.
\end{lem}

\begin{proof} Let $V=I_{19}$ and $S=I_s$. Let $W=\{16,17,18\}$ and $T=\{s-3,s-2,s-1\}$. The desired HSASs filled with pairs of points from  $V$  and indexed by $S$ are presented as follows.

\noindent $s=19$:
{\scriptsize
$$%
$$
}
\end{proof}

\begin{lem}
There exists an HSAS$(s,v;5,3)$ for each $(s,v)\in \{(21,11),(29,15),(37,19)\}$.
\end{lem}

\begin{proof} Let $V=I_v$ and $S=I_s$. Let $W=\{v-3,v-2,v-1\}$ and $T=\{s-5,s-4,s-3,s-2,s-1\}$. The desired HSASs filled with pairs of points from  $V$  and indexed by $S$ are presented as follows.

\noindent $(s,v)=(21,11)$:
{\scriptsize
$$%
$$
}
\end{proof}
\end{document}